\documentclass[runningheads]{llncs}
\pdfoutput=1
\newif\ifanonym
\anonymfalse

\usepackage{graphicx}
\usepackage{amsmath,amssymb}       

\usepackage{hyperref}
\usepackage{cleveref}
\usepackage[noadjust,nospace]{cite}

\usepackage{nicefrac}
\usepackage{mleftright}
\usepackage{cases}
\usepackage{empheq}
\usepackage{mathbbol}
\usepackage{stmaryrd}
\usepackage{bm}
\usepackage{bbm}
\DeclareSymbolFontAlphabet{\amsmathbb}{AMSb}
\usepackage{mathtools}
\usepackage{accents}
\DeclareSymbolFontAlphabet{\amsmathbb}{AMSb}

\usepackage{listings}
\usepackage{enumitem}
\usepackage[colorinlistoftodos]{todonotes}

\usepackage{siunitx}
\usepackage{caption}
\usepackage{subcaption}
\usepackage{tabularx}
\usepackage{booktabs}
\usepackage{scrextend}
\usepackage{threeparttable}

\newcommand\munderbar[1]{%
  \underaccent{\bar}{#1}}

\usepackage{xspace}
\usepackage{microtype}

\usepackage{cleveref}

\usepackage{tikz}
\usepackage{pgfplots}
\DeclareUnicodeCharacter{2212}{−}
\usepgfplotslibrary{external,fillbetween,groupplots,dateplot}
\usetikzlibrary{patterns,shapes.arrows}
\pgfplotsset{compat=newest}

\definecolor{red}{rgb}{0.745,0.192,0.102}
\definecolor{darkgreen}{RGB}{34,161,55}
\definecolor{ruhuisstijlrood}{rgb}{0.745,0.192,0.102}
\definecolor{ruhuisstijlzwart}{rgb}{0,0,0}
\definecolor{ruhuisstijlwit}{rgb}{0.98,0.98,0.98}

\definecolor{plotblue}{rgb}{0.1,0.498039215686275,0.9549019607843137}

\usepackage{mdframed}

\usetikzlibrary{backgrounds,patterns,patterns.meta,arrows,arrows.meta,calc,shapes,shadows,decorations.pathmorphing,decorations.pathreplacing,automata,shapes.multipart,positioning,shapes.geometric,fit,circuits,trees,shapes.gates.logic.US,fit, shadows.blur,shapes.symbols,intersections}

\usepackage{fontawesome5}

\crefname{equation}{Eq.}{Eqs.}
\crefname{pluralequation}{Eqs.}{Eqs.}

\crefname{algorithm}{Algorithm}{Algorithm}

\crefname{figure}{Fig.}{Figs.}
\crefname{pluralfigure}{Figs.}{Figs.}

\crefname{section}{Sect.}{Sects.}
\crefname{pluralsection}{Sects.}{Sects.}

\crefname{table}{Table}{Table}
\crefname{pluraltable}{Tables}{Tables}

\crefname{definition}{Def.}{Def.}
\crefname{pluraldefinition}{Defs.}{Defs.}

\crefname{theorem}{Theorem}{Theorems}
\crefname{pluraltheorem}{Theorems}{Theorems}

\crefname{lemma}{Lemma}{Lemmas}
\crefname{plurallemma}{Lemmas}{Lemmas}

\crefname{example}{Example}{Example}
\crefname{pluralexample}{Examples}{Examples}

\crefname{problem}{Problem}{Problem}
\crefname{pluralproblem}{Problems}{Problems}

\crefname{assumption}{Assumption}{Assumption}
\crefname{pluralassumption}{Assumptions}{Assumptions}

\crefname{remark}{Remark}{Remark}
\crefname{pluralremark}{Remarks}{Remarks}

\crefname{proposition}{Proposition}{Proposition}
\crefname{pluralproposition}{Propositions}{Propositions}

\crefname{corollary}{Corollary}{Corollary}
\crefname{pluralcorollary}{Corollaries}{Corollaries}

\crefname{appendix}{Appendix}{Appendices}
\crefname{pluralappendix}{Appendices}{Appendices}

\tikzset{ 
    every state/.style={
        rectangle, 
        rounded corners,
        semithick,
        fill=gray!10,
        minimum height=0.8cm,
        inner sep=3pt,
        minimum size=3pt
        }
}

\definecolor{labelA}{rgb}{0.992156863,0.682352941,0.380392157}
\definecolor{empty_color}{rgb}{1,1,0.749019608}
\definecolor{nonempty_color}{rgb}{0.670588235,0.866666667,0.643137255}

\newcommand{\bigcdot}[1]{%
 \begin{tikzpicture}
  \node [circle, draw=black, outer sep=2pt, inner sep=2pt, minimum size=2pt, fill=#1] {};
 \end{tikzpicture}}
 
\tikzset{empty label/.style={draw, fill=empty_color}}
\tikzset{nonempty label/.style={draw, fill=nonempty_color}}
\tikzset{no label/.style={draw}}

\tikzset{small/.style={draw, inner sep=1pt, minimum size=1pt, font=\scriptsize}}
\tikzset{medium/.style={draw, inner sep=1pt, minimum size=15pt, font=\scriptsize}}
\newcommand{\Prob}{\amsmathbb{P}}

\newcommand{\Real}{\amsmathbb{R}}
\newcommand{\Q}{\amsmathbb{Q}}

\newcommand{\N}{\amsmathbb{N}}
\newcommand{\distr}[1]{\ensuremath{\mathit{Dist(#1)}}}

\newcommand{\diag}[1]{\ensuremath{{\mathrm{diag}(#1)}}}

\newcommand{\Always}{\Box \, }
\newcommand{\Finally}{\lozenge \, }

\newcommand{\ctmc}{\ensuremath{\mathcal{C}}}
\newcommand{\mdp}{\ensuremath{\mathcal{M}}}
\newcommand{\imdp}{\ensuremath{\mathcal{I}}}
\newcommand{\sched}{\ensuremath{\sigma}}
\newcommand{\Sched}{\ensuremath{\mathrm{Sched}}}

\newcommand{\imphist}{\ensuremath{\Omega}}
\newcommand{\imp}[2]{\ensuremath{[\,\munderbar{#1}_{#2}, \bar{#1}_{#2}\,]}}

\newcommand{\timegraph}{\ensuremath{\mathcal{G}}}
\newcommand{\nodes}{\ensuremath{\mathcal{N}}}
\newcommand{\edges}{\ensuremath{\mathcal{E}}}
\newcommand{\post}{\ensuremath{\mathsf{post}}}

\newcommand{\Unfold}{\ensuremath{\mathsf{Unfold}}}

\newcommand{\Abstract}{\ensuremath{\mathsf{Abstract}}}

\newcommand{\pr}{\ensuremath{\text{Pr}}}

\newcommand{\sTr}{\ensuremath{\mathsf{sTr}}}

\newcommand{\hor}{\ensuremath{d}}

\newcommand{\impT}{\ensuremath{\mathcal{T}}}

\newcommand{\interval}{\ensuremath{\mathbbm{I}}}

\newcommand{\partition}{\ensuremath{\Psi}}
\newcommand{\Partition}{\ensuremath{\mathsf{partition}}}

\newcommand{\calP}{\ensuremath{\mathcal{P}}}

\newcommand{\consistent}{\ensuremath{\sim}}

\newcommand{\tuple}[1]{\ensuremath{\langle #1 \rangle}} 
\newcommand{\closure}{\ensuremath{\mathrm{cl}}} 
\newcommand{\weight}{\ensuremath{w}} 
\newcommand{\Weight}{\ensuremath{W}} 

\makeatletter
\def\Ddots{\mathinner{\mkern1mu\raise\p@
\vbox{\kern7\p@\hbox{.}}\mkern2mu
\raise4\p@\hbox{.}\mkern2mu\raise7\p@\hbox{.}\mkern1mu}}
\makeatother

\DeclareMathOperator*{\argmax}{arg\,max}

\newcommand{\ex}[1]{\textsc{#1}}
\newcommand{\storm}{\textsc{Storm}}

\renewcommand{\paragraph}[1]{\smallskip\noindent\emph{#1}\,\,}
\renewcommand{\subsubsection}[1]{\medskip\noindent\textbf{#1}}

\makeatletter
\def\orcidID#1{\smash{\href{http://orcid.org/#1}{\protect\raisebox{-1.25pt}{\protect\includegraphics{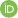}}}}}
\makeatother

\newif\ifappendix
\global\appendixtrue
 
\newif\iftikzcompile
\tikzcompiletrue 

\pgfdeclarelayer{bg}    
\pgfsetlayers{bg,main}  

\begin{document}

\title{CTMCs with Imprecisely Timed Observations
\ifanonym\else
\thanks{This work has been funded by the NWO grant \href{https://primavera-project.com}{PrimaVera} (NWA.1160.18.238), 
the EU Consolidator Grant 864075 (CAESAR)
the EU Horizon 2020 project MISSION (101008233), and the ERC Starting Grant 101077178 (DEUCE).}
\fi
}
\titlerunning{CTMCs with Imprecisely Timed Observations}

\ifanonym
\author{}
\institute{}
\authorrunning{Authors omitted}
\else
\author{
Thom Badings\inst{1}\orcidID{0000-0002-5235-1967}
\and Matthias~Volk\inst{2}\orcidID{0000-0002-3810-4185}
\and Sebastian~Junges\inst{1}\orcidID{0000-0003-0978-8466}
\and Marielle~Stoelinga\inst{1,3}\orcidID{0000-0001-6793-8165}
\and Nils~Jansen\inst{1,4}\orcidID{0000-0003-1318-8973}
}
\authorrunning{T.\ Badings, M.\ Volk, S.\ Junges, M.\ Stoelinga, N.\ Jansen}
\institute{
Radboud University, Nijmegen, the Netherlands \and Eindhoven University of Technology, Eindhoven, the Netherlands\and
University of Twente, Enschede, the Netherlands\and
Ruhr-University Bochum, Germany
}
\fi

\maketitle
\begin{abstract}
Labeled continuous-time Markov chains (CTMCs) describe processes subject to random timing and partial observability. In applications such as runtime monitoring, we must incorporate past observations. The timing of these observations matters but may be uncertain. Thus, we consider a setting in which we are given a sequence of imprecisely timed labels called the evidence. The problem is to compute reachability probabilities, which we condition on this evidence. Our key contribution is a method that solves this problem by unfolding the CTMC states over all possible timings for the evidence. We formalize this unfolding as a Markov decision process (MDP) in which each timing for the evidence is reflected by a scheduler. This MDP has infinitely many states and actions in general, making a direct analysis infeasible. Thus, we abstract the continuous MDP into a finite interval MDP (iMDP) and develop an iterative refinement scheme to upper-bound conditional probabilities in the CTMC. We show the feasibility of our method on several numerical benchmarks and discuss key challenges to further enhance the performance.
\end{abstract}

\section{Introduction}
\label{sec:introduction}

Continuous-time Markov chains (CTMCs) are stochastic processes subject to random timing, which are ubiquitous in reliability engineering~\cite{DBLP:journals/csr/RuijtersS15}, network processes~\cite{DBLP:conf/srds/HaverkortHK00,hermanns1999multi}, and systems biology~\cite{DBLP:journals/acta/CeskaDPKB17,DBLP:conf/tacas/BortolussiS18}.
Here, we consider finite-state labeled CTMCs, which exhibit partial observability through a labeling function, such that analysis can only be done based on observations of the state.
Specific techniques such as model checking algorithms compute quantitative aspects of CTMC behavior under the assumption of a static and known initial state~\cite{DBLP:journals/tse/BaierHHK03,aziz2000model}.

\paragraph{Conditional probabilities}
In applications such as runtime monitoring~\cite{DBLP:journals/fmsd/SanchezSABBCFFK19,DBLP:series/lncs/BartocciDDFMNS18}, we need to analyze an already running system without a static initial state.
Instead, we must incorporate past observations, which are given as a sequence of CTMC labels, each of which is observed at a specific time.
We call this sequence of timed labels the \emph{evidence}.
We want to incorporate this evidence by conditioning the state of the CTMC on the evidence.
For example, ``what is the probability of a failure for a production machine (modeled as a CTMC) before time $T$, given that we have observed particular labels at earlier times $t_1, t_2, \ldots, t_n$?''

\paragraph{Imprecise observation times}
These conditional probabilities depend on the exact time at which each label was observed.
However, in realistic scenarios, the times for the labels in the evidence may not be known precisely.
For example, inspections are always done in the first week of a month, but the precise moment of inspection may be unknown.
Intuitively, we can interpret such \emph{imprecisely timed evidence} as a potentially infinite set of (precisely timed) \emph{instances} of the evidence that vary only in the observation times.
For example, an inspection done on ``\emph{January 2 exactly at noon}'' is an instance of the imprecise observation time of ``\emph{the first week of January}.''
This perspective motivates a robust version of the previous question:
``Given the imprecisely timed evidence, what is the maximal probability of a failure before time $T$ over all instances of the evidence?''

\paragraph{Problem statement}
In this paper, we are given a labeled CTMC together with imprecisely timed evidence.
For each instance of the evidence, we can define the probability of reaching a set of target states, conditioned on that evidence.
The problem is to compute the supremum over these conditional probabilities for all instances of the evidence.
We generalize this problem by considering \emph{weighted} conditional reachability probabilities (or simply the \emph{weighted reachability}), where we assign to each state a nonnegative weight.
Standard conditional reachability is then a special case with a weight of one for the target states and zero elsewhere.

\paragraph{Contributions}
Our main contribution is the first method to compute weighted conditional reachability probabilities in CTMCs with imprecisely timed evidence.
Our approach consists of the following three steps.

\paragraph{1) Unfolding}
In \cref{sec:unfolding}, we introduce a method that \emph{unfolds} the CTMC over all possible timings of the imprecisely timed evidence.
We formalize this unfolding as a Markov decision process (MDP)~\cite{DBLP:books/wi/Puterman94}, in which the timing imprecision is reflected by nondeterminism.
We show that the weighted reachability can be computed via (unconditional) reachability probabilities on a transformed version of this MDP~\cite{DBLP:conf/tacas/BaierKKM14,DBLP:conf/cav/JungesTS20}.
For the special case of evidence with precise observation times, we obtain a precise solution to the problem that we can directly compute.

\paragraph{2) Abstraction}
In general, imprecisely timed evidence yields an unfolded MDP with infinitely many states and actions.
In \cref{sec:abstraction}, we propose an abstraction of this continuous MDP as a finite interval MDP (iMDP)~\cite{DBLP:journals/ai/GivanLD00}, similar to game-based abstractions~\cite{DBLP:journals/fmsd/KattenbeltKNP10}.
A robust analysis of the iMDP yields upper and lower bounds on the weighted reachability for the CTMC.
Moreover, we propose an iterative refinement scheme that converges to the weighted reachability in the limit.

\paragraph{3) Computing bounds in practice}
In \cref{sec:algorithm}, we use the iMDP abstraction and refinement to obtain sound upper and lower bounds on the weighted reachability in practice.
In \cref{sec:experiments}, we show the feasibility of our method across several numerical benchmarks.
Concretely, we show that we obtain reasonably tight bounds on the weighted reachability within a reasonable time.
Finally, we discuss the key challenges in further enhancing the performance of our method in \cref{sec:conclusion}.

\paragraph{Related work}
Closest to our problem are works on model checking CTMCs against deterministic timed automata (DTA)~\cite{DBLP:journals/corr/abs-1101-3694,DBLP:journals/pe/AmparoreD18,DBLP:conf/cav/FengKLXZ18}.
Evidence can be expressed as a single-clock DTA, and tools such as \ex{MC4CSL}~\cite{DBLP:conf/qest/AmparoreD10a} can calculate the weighted reachability for precise timings. However, for imprecisely timed evidence, checking CTMCs against DTAs yields the \emph{sum of probabilities} over all instances of the evidence, whereas we are interested in the \emph{maximal probability} over all instances.

Our setting is also similar to synthesizing timeouts in CTMCs with fixed-delay transitions~\cite{DBLP:conf/qest/BrazdilKKNR15,DBLP:conf/mascots/KorenciakKR16,DBLP:journals/tomacs/BaierDKKR19}.
Finding optimal timeouts is similar to our objective of finding an instance of the imprecisely timed evidence such that the weighted reachability is maximized.
While timeouts can model the time \emph{between} observations, we consider \emph{global} observation times, i.e., the time between observations depends on the previous time of observation---which cannot be modeled with timeouts.

We discuss other related work in more detail in \cref{sec:related}.
\section{Problem Statement}
\label{sec:preliminaries}
We recap continuous-time Markov chains (CTMCs)~\cite{DBLP:journals/tse/BaierHHK03,aziz2000model} and formalize the problem statement.
The set of all probability distributions over a finite set $X$ is denoted as $\distr{X}$.
We write tuples $\tuple{ a, b }$ with square brackets, and $\mathbbm{1}_x$ is the indicator function over $x$, i.e., $\mathbbm{1}_{(y=z)}$ is one if $y=z$ and zero otherwise.
We use the standard temporal operators $\Finally$ and $\Always$ to denote \emph{eventually} reaching or \emph{always} being in a state~\cite{BK08}.

\begin{definition}[CTMC]
\label{def:ctmc}
A (labeled) \emph{continuous-time Markov chain} $\ctmc$ is a tuple $\tuple{ S, s_I, \Delta, E, C, L }$ with a finite set $S$ of \emph{states}, an \emph{initial state} $s_I \in S$, a \emph{transition matrix} $\Delta \colon S \rightarrow \distr{S}$, \emph{exit-rates} $E \colon S \rightarrow \Q_{\geq 0}$, a finite set of \emph{colors} $C$, and a labeling function $L\colon S \rightarrow C$.
\end{definition}

A \emph{(timed) CTMC path} $\pi = s_0 t_0 s_1 t_2 s_3 \cdots \in \Pi = S \times (\Real_{\geq 0} \times S)^*$ is an alternating sequence of states and residence times, where $\Delta(s_i)(s_{i+1}) > 0 \, \forall i \in \N$.
The path $s_0 3 s_1 4 s_2$ means we stayed exactly 3 time units in $s_0$, then transition to $s_1$, where we stayed 4 time units before moving to $s_2$.
The CTMC state at time $t \in \Real_{\geq 0}$ is denoted by $\pi(t) \in S$, e.g., $\pi(6.2) = s_1$ for the example path above.

An alternative (and equivalent) view of CTMCs is to combine the transition matrix $\Delta$ and exit-rates $E$ in a transition rate matrix $R \colon S \times S \to \Q_{\geq 0}$, where $R(s,s') = \Delta(s,s') \cdot E(s)$ $\forall s, s' \in S$~\cite{DBLP:conf/lics/Katoen16}.
From state $s \in S$, the \emph{transient probability distribution} $\pr_s (t) \in \distr{S}$ after time $t \geq 0$ is $\pr_s (t) = \delta_s \cdot e^{(R - \diag{E})t}$, where $\delta_s \in \{0,1\}^{|S|}$ is the Dirac distribution for state $s$, and $\diag{E}$ is the diagonal matrix with the exit rates $E$ on the diagonal.
Thus, the probability of starting in state $s$ and being in state $s' \in S$ after time $t$ is $\Pr_s(t)(s') \in [0,1]$.

\begin{example}
\label{example:queue}
    Consider a simple, single-product inventory where the number of items in stock ranges from 0 to 2, but we can only observe if the inventory is empty or not.
    This system is modeled by the CTMC shown in \cref{fig:ctmc} with states $S = \{s_0,s_1,s_2\}$ (modeling the stock) and labels shown by the two colors (\bigcdot{empty_color}~for empty and \bigcdot{nonempty_color}~for nonempty).
    The rates at which items arrive and deplete are $R(s_0,s_1) = R(s_1,s_2) = 3$ and $R(s_1,s_0) = R(s_2,s_1) = 2$, respectively.
\end{example}

\begin{figure}[t]
\begin{subfigure}[b]{0.2\linewidth}
    \centering
    \iftikzcompile
    \begin{tikzpicture}[node distance=0.8cm]
      \node[state, empty label, initial,initial where=left, initial text=] (s0) {$s_0$};
      \node[state, nonempty label, below of=s0] (s1) {$s_1$};
      \node[state, nonempty label, below of=s1] (s2) {$s_2$};
      \path[->] (s0) edge [bend right=30] node [left, bend right=45] {$\lambda_a$} (s1)
                (s1) edge [bend right=30] node [left, bend right=45] {$\lambda_a$} (s2)
                (s2) edge [bend right=30] node [right] {$\lambda_d$} (s1)
                (s1) edge [bend right=30] node [right] {$\lambda_d$} (s0);
    \end{tikzpicture}
    \fi
    \caption{CTMC.}
    \label{fig:ctmc}
    \end{subfigure}
\hfill 
\begin{subfigure}[b]{0.3\linewidth}
    \centering
    \iftikzcompile
    \begin{tikzpicture}[node distance=1.2cm]
    \def\xshift{-0.3cm}
      \node[state, no label, initial,initial where=left, initial text=, align=center] (t0) {$0$};
      \node[state, no label, right of=t0, align=center, xshift=\xshift] (t1) {$t_1$};
      \node[state, no label, right of=t1, align=center, xshift=\xshift] (t2) {$t_2$};
      \node[state, no label, right of=t2, align=center, xshift=\xshift] (tstar) {$t_\star$};
      \path[->] (t0) edge [] node [above] {} (t1)
                (t1) edge [] node [above] {} (t2)
                (t2) edge [] node [above] {} (tstar);
    \end{tikzpicture}
    \fi
     \caption{Evidence graph.}
    \label{fig:inspection_graph}
\end{subfigure}
\hfill
\begin{subfigure}[b]{0.48\linewidth}
    \centering
    \iftikzcompile
    \begin{tikzpicture}[node distance=0.8cm]

    \def\xshift{0.6cm}
      \node[state, no label, initial,initial where=left, initial text=] (s0t0) {$s_0, 0$};
      \node[state, no label, below of=s0t0] (s1t0) {$s_1, 0$};
      \node[state, no label, below of=s1t0] (s2t0) {$s_2, 0$};
      \node[state, empty label, right of=s0t0, xshift=\xshift] (s0t1) {$s_0, t_1$};
      \node[state, nonempty label, below of=s0t1] (s1t1) {$s_1, t_1$};
      \node[state, nonempty label, below of=s1t1] (s2t1) {$s_2, t_1$};
      \node[state, empty label, right of=s0t1, xshift=\xshift] (s0t2) {$s_0, t_2$};
      \node[state, nonempty label, below of=s0t2] (s1t2) {$s_1, t_2$};
      \node[state, nonempty label, below of=s1t2] (s2t2) {$s_2, t_2$};
      \node[state, no label, right of=s0t2, xshift=\xshift] (s0star) {$s_0, t_\star$};
      \node[state, no label, below of=s0star] (s1star) {$s_1, t_\star$};
      \node[state, no label, below of=s1star] (s2star) {$s_2, t_\star$};
      \begin{pgfonlayer}{bg}    
      \path[->,shorten <= -2pt+\pgflinewidth] 
                (s0t0) edge [] node [above] {} (s0t1)
                (s0t0) edge [] node [above] {} (s1t1)
                (s0t0) edge [] node [above] {} (s2t1)
                (s1t0) edge [] node [above] {} (s0t1)
                (s1t0) edge [] node [above] {} (s1t1)
                (s1t0) edge [] node [above] {} (s2t1)
                (s2t0) edge [] node [above] {} (s0t1)
                (s2t0) edge [] node [above] {} (s1t1)
                (s2t0) edge [] node [above] {} (s2t1)
                (s0t1) edge [] node [above] {} (s0t2)
                (s0t1) edge [] node [above] {} (s1t2)
                (s0t1) edge [] node [above] {} (s2t2)
                (s1t1) edge [] node [above] {} (s0t2)
                (s1t1) edge [] node [above] {} (s1t2)
                (s1t1) edge [] node [above] {} (s2t2)
                (s2t1) edge [] node [above] {} (s0t2)
                (s2t1) edge [] node [above] {} (s1t2)
                (s2t1) edge [] node [above] {} (s2t2)
                (s0t2) edge [] node [above] {} (s0star)
                (s1t2) edge [] node [above] {} (s1star)
                (s2t2) edge [] node [above] {} (s2star);
    \end{pgfonlayer}
    \end{tikzpicture}
    \fi
     \caption{Unfolded MDP.}
    \label{fig:unfolding}
\end{subfigure}
    \caption{The CTMC (a) for \cref{example:queue}, (b) the graph for the precise evidence $\rho = \tuple{ t_1, o_1 }, \tuple{ t_2, o_2 }$, and (c) the states of the MDP unfolding defined by \cref{def:unfolded_mdp}.}
\end{figure}
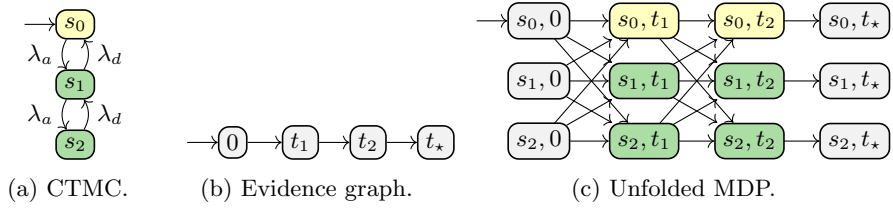

\subsection{Problem statement}
\label{subsec:problem_statement}

The key problem we want to solve is to compute reachability probabilities for the CTMC conditioned on a timed sequence of labels, which we call the \emph{evidence}.

\paragraph{Evidence}
The \emph{evidence} $\rho = \tuple{ t_1, o_1 }, \ldots, \tuple{ t_\hor, o_\hor } \in (\Real_{> 0} \times C)^\hor$ is a sequence of $d$ times and labels such that $t_i < t_{i+1}$ for all $i \in \{1,\ldots,\hor-1\}$.
A timed label $\tuple{ t_i, o_i }$ means that at time $t_i$, the CTMC was in a state $s \in S$, that is, $L(s) = o_i$.
Since each time $t \in \Real_{> 0}$ can only occur once in $\rho$, we overload $\rho$ and denote the evidence at time $t \in \{t_1,\ldots,t_\hor\}$ by $\rho(t) = o \in C$, such that $\tuple{ t,o } \in \rho$.
While a timed path of a CTMC describes the state at every continuous point in time, the evidence only contains the observations at $\hor$ points in time.
We say that a path $\pi$ is \emph{consistent} with evidence $\rho$, written as $\pi \models \rho$, if each timed label in $\rho$ matches the label of path $\pi$ at time $t$, i.e., if $L(\pi(t)) = \rho(t) \, \forall t \in \{t_1,\ldots,t_\hor\}$.

\paragraph{Conditional probabilities}
We want to compute the conditional probability $\Prob_\ctmc(\pi(t_\hor) = s ) \mid [\pi \models \rho])$ that the CTMC $\ctmc$ with initial state $s_I$ generates a path being in state $s$ at time $t_\hor$, conditioned on the evidence $\rho$.
Using Bayes' rule, we can characterize this conditional probability as follows (assuming $\frac{0}{0} = 0$, for brevity):
\begin{equation}
\begin{split}
    \label{eq:BayesRule}
    \Prob_\ctmc(\pi(t_\hor) = s \mid [\pi \models \rho]) 
    &= \frac{\Prob_\ctmc([\pi(t_\hor) = s] \cap [\pi \models \rho])}{\Prob_\ctmc(\pi \models \rho)}.
\end{split}
\end{equation}

\paragraph{Imprecise timings}
We extend evidence with uncertainty in the timing of each label.
The \emph{imprecisely timed evidence} (or \emph{imprecise evidence}) $\imphist = \tuple{ \impT_1, o_1 }, \ldots, \tuple{ \impT_\hor, o_\hor }$ is a sequence of $\hor$ labels and uncertain timings $\impT_i = \cup_{j=1}^q \imp{t}{j}$, with $\munderbar{t}_j \leq \bar{t}_j$ and $q \in \N$.
Observe that $\impT$ can model both singletons ($\impT_i = \{1,2,3\}$) and unions of intervals ($\impT_i = [1,1.5] \cup [2,2.5]$).
We require that $\max_{t \in \impT_i} (t) < \min_{t' \in \impT_{i+1}} (t')$ for all $i \in \{1,\ldots,d-1\}$, i.e., the order of the labels is known, despite the uncertainty in the observation times.
Again, we overload notation and denote the evidence at time $t$ by $\imphist(t) = o$, such that $\exists \tuple{ \impT,o } \in \imphist$ with $t \in \impT$.
Imprecise evidence induces a set of \emph{instances} of the evidence that only differ in the label times.
This set of instances is uncountably infinite whenever one of the imprecise timings $\impT$ is a continuous set.
Formally, the evidence $\rho = \tuple{ t_1, o_1 }, \ldots, \tuple{ t_\hor, o_\hor }$ is an instance of the imprecise evidence $\imphist$, written as $\rho \in \imphist$, if $t_i \in \impT_i$ for all $i=1,\ldots,\hor$.

\begin{example}
    An example of imprecise evidence for the CTMC in \cref{example:queue} is $\imphist = \tuple{[0.2, 0.8], \bigcdot{empty_color}}, \tuple{[1.4,2.1], \bigcdot{nonempty_color}}$.
    The precise evidence $\rho = \tuple{0.4, \bigcdot{empty_color}}, \tuple{1.9, \bigcdot{nonempty_color}}$ is an instance of $\imphist$, i.e., $\rho \in \imphist$.
    However, $\rho' = \tuple{0.1, \bigcdot{empty_color}}, \tuple{1.9, \bigcdot{nonempty_color}}$ and $\rho'' = \tuple{0.4, \bigcdot{nonempty_color}}, \tuple{1.9, \bigcdot{nonempty_color}}$ are not, i.e., $\rho' \notin \imphist$, $\rho'' \notin \imphist$, as the timings and labels do not match, respectively.
\end{example}

\paragraph{State-weights}
Let $\weight \colon S \to \Real_{\geq 0}$ be a \emph{state-weight function}, which assigns to each CTMC state $s \in S$ a non-negative weight.
The weight $w(s)$ represents a general measure of risk associated with each state $s \in S$, as used in~\cite{DBLP:conf/cav/JungesTS20}.
For example, $w(s)$ may represent the probability of reaching a set of target states $S_T$ from $s$ within some time horizon $h \geq 0$.
We then consider the following problem.
\begin{mdframed}
\begin{problem}[Weighted conditional reachability probability]
\label{problem}
Given a CTMC $\ctmc$, a state-weight function $\weight$, and the imprecisely timed evidence $\imphist$, compute the (maximal) weighted conditional reachability probability $\Weight(\imphist)$:
\begin{equation}
    \label{eq:problem}
    \Weight(\imphist) = \sup_{\rho \in \imphist} \, \sum_{s \in S} \Prob_\ctmc(\pi(t_\hor) = s \mid [\pi \models \rho]) \cdot \weight(s).
\end{equation}
\end{problem}
\end{mdframed}

\begin{example}
    \label{ex:weights}
    For the CTMC in \cref{example:queue}, consider the state-weight function that assigns to each state the probability of reaching state $s_0$ within time $t=0.1$.
    Then, the problem above is interpreted as: \emph{Given the imprecisely timed evidence $\imphist$, compute the probability (conditioned on $\imphist$) of reaching state $s_0$ within time $t=0.1$ (after the end of the evidence).}
\end{example}

Our overall workflow to solve \cref{problem} is summarized in \cref{fig:Approach} and consists of four blocks, which we discuss in \cref{sec:unfolding,sec:abstraction,sec:algorithm}, respectively.

\paragraph{Variations}
To instead minimize \cref{eq:problem}, we would swap every $\inf$ and $\sup$ (and $\max$ and $\min$) in the paper, but our general approach remains the same.
Furthermore, by setting $w(s) = 1$ for all $s \in S_T$ and zero otherwise, we can also compute the probability of being in a state in $S_T$ \emph{immediately} after the evidence.
Finally, we remark that \cref{problem} only considers events \emph{after} the end of the evidence.
This setting is motivated by applications where the exact system state is not observable, but actual system failures can be observed.
Thus, one can typically assume that the system has not failed yet and the problem as formalized in  \cref{problem} is to predict the conditional probability of a future system failure.

\begin{figure}[t!]
	\centering
\iftikzcompile
	\usetikzlibrary{calc}
\usetikzlibrary{arrows.meta}
\usetikzlibrary{positioning}
\usetikzlibrary{fit}
\usetikzlibrary{shapes}

\tikzstyle{node} = [rectangle, rounded corners, minimum width=3.0cm, text width=3.0cm, minimum height=1.1cm, text centered, draw=black, fill=plotblue!15]
\tikzstyle{nodebig} = [rectangle, rounded corners, minimum width=3.5cm, text width=3.5cm, minimum height=1.1cm, text centered, draw=black, fill=plotblue!15]

\resizebox{\linewidth}{!}{%
	\begin{tikzpicture}[node distance=5cm,->,>=stealth,line width=0.3mm,auto]
		
		\newcommand\xshift{-1cm}
		\newcommand\yshift{-3.9cm}
		
		\node (unfolding) [node] {\textbf{(1)} MDP unfolding \\ $\mdp = \Unfold(\ctmc, \timegraph_\imphist)$ \\ \textbf{(\cref{def:unfolded_mdp})}};
		\node (condition) [node, right of=unfolding, xshift=-1cm] {\textbf{(2)} Conditioning on the evidence $\imphist$ \\ \textbf{(\cref{def:conditioned_MDP})}};
		\node (abstraction) [nodebig, right of=condition, xshift=-0.7cm] {\textbf{(3)} Abstract iMDP \\ $\imdp = \Abstract(\mdp_{|\imphist}, \partition)$ \\ \textbf{(\cref{def:abstraction})}};
		\node (compute) [node, right of=abstraction, xshift=-0.7cm] {\textbf{(4)} Compute upper and lower bounds \\ \textbf{(\cref{lemma:computable_bounds})}};
  
		
		\node (in1) [above of=unfolding, yshift=\yshift] {};
		\node (in2) [above of=condition, yshift=\yshift] {};
		\node (in3) [above of=abstraction, yshift=\yshift] {};
		\node (out) [right of=compute, xshift=-1.2cm] {};
		
		
		\draw [->] (in1) -- (unfolding) node [pos=0.0, above, align=left, font=\sffamily\small] {1) CTMC $\ctmc$ \\ 2) State-weight function $\weight$ \\ 3) Imprecise evidence $\imphist$};
		
		\draw [->] ($(in3) + (-1cm, 0cm)$) -- ($(abstraction.north) + (-1cm, 0cm)$) node [pos=0.0, above, align=center, font=\sffamily\small] {Initial time \\ partition $\partition$};

            \draw [->] (compute) -- (out) node [pos=0.55, above, align=center, font=\sffamily\small] {Upper and \\ lower bounds \\ on $\Weight(\imphist)$};
		
		
		\draw [->] (unfolding) -- (condition) node [pos=0.5, above, align=center, font=\sffamily\small] {$\mdp$};
		
		\draw [->] (condition) -- (abstraction) node [pos=0.5, above, align=center, font=\sffamily\small] {$\mdp_{|\imphist}$};
		
		\draw [->] (abstraction) -- (compute) node [pos=0.5, above, align=center, font=\sffamily\small] {$\imdp$};		

            \draw [->] (compute.north) -- ($(compute.north) + (0, 0.4cm)$) -| (abstraction) node[pos=0.25, above, align=center, font=\sffamily\small] {Refinement (splitting \\ elements of $\partition$)};


		\draw [-,decorate,decoration={brace,amplitude=5pt,mirror,raise=3ex}] (-1.55cm, -0.5cm) -- (5.6cm, -0.5cm) node[midway, below, align=center, yshift=-0.6cm, font=\normalsize] {\textbf{\cref{sec:unfolding}}};
          \draw [-,decorate,decoration={brace,amplitude=5pt,mirror,raise=3ex}] (6.4cm, -0.5cm) -- (10.15cm, -0.5cm) node[midway, below, align=center, yshift=-0.6cm, font=\normalsize] {\textbf{\cref{sec:abstraction}}};
          \draw [-,decorate,decoration={brace,amplitude=5pt,mirror,raise=3ex}] (11cm, -0.5cm) -- (14.2cm, -0.5cm) node[midway, below, align=center, yshift=-0.6cm, font=\normalsize] {\textbf{\cref{sec:algorithm}}};
  
	\end{tikzpicture}
}
\fi
	\vspace{-2em}
	\caption{
		Conceptual workflow of our approach for solving \cref{problem}.
	}
	\label{fig:Approach}
\end{figure}
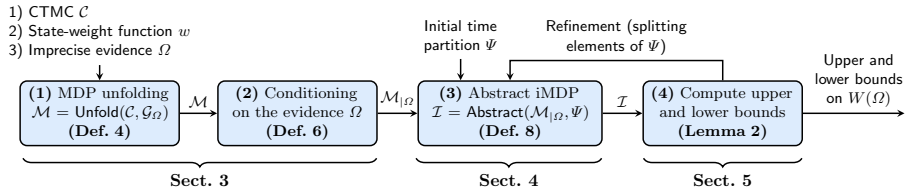

\subsection{Interval Markov decision processes}
\label{sec:preliminaries:MDPs}
We recap interval MDPs (iMDPs)~\cite{DBLP:journals/ai/GivanLD00} and define standard MDPs as special case.
We denote (i)MDP states by $q \in Q$, whereas CTMC states are denoted $s \in S$.
\begin{definition}[iMDP]
\label{def:iMDP}
An \emph{interval MDP} $\imdp$ is a tuple $\tuple{ Q, q, A, \calP }$, with $Q$ a set of states, $q \in Q$ the initial state, $A$ a set of actions, and where the uncertain transition function $\calP \colon Q \times A \times Q \rightharpoonup \interval \cup \{ [0,0] \}$ is defined over intervals $\interval = \{ [a,b] \mid a,b \in (0,1] \text{ and } a \leq b \}$.
The actions enabled in state $q \in Q$ are $A(q) \subseteq A$.
\end{definition}
The assumption that an interval cannot have a lower bound of $0$ except the $[0,0]$ interval is standard, see, e.g.,~\cite{DBLP:conf/cdc/WolffTM12,DBLP:conf/cav/PuggelliLSS13}.
An MDP is a special case of iMDP, where the upper and lower bounds coincide, i.e., $\calP(q,a,q') = [b,b], \, b \in [0,1]$ for all intervals, and each $\calP(q,a,\cdot) \in \distr{Q}$ is a distribution over states.
We denote an MDP as $\mdp = \tuple{ Q,q,A,P }$, with transition function $P \colon Q \times A \times Q \rightharpoonup [0,1]$.
For an MDP $\mdp$ with transition function $P$, we write $P \in \calP$ if for all $q,q' \in Q$ and $a \in A$ we have $P(q,a,q') \in \calP (q,a,q')$ and each $P(q, a, \cdot) \in \distr{Q}$. 
Fixing a transition function $P \in \calP$ for iMDP $\imdp$ yields an induced MDP $\imdp[P]$.

The nondeterminism in an iMDP $\imdp$ is resolved by a memoryless scheduler $\sched \colon Q \to A$, with $\sched \in \Sched_{\imdp}$ the set of all schedulers.
We denote a finite (i)MDP path by $\xi = q_0, \ldots, q_n \in \Xi_{\imdp}^\sched$, where $\Xi_{\imdp}^\sched$ is the set of all paths under scheduler $\sched$.
For the Markov chain induced by scheduler $\sched$ in $\imdp[P]$, we use the standard probability measure $\Prob_{\imdp[P]}^\sched$ over the smallest sigma-algebra containing the cylinder sets of all finite paths $\xi \in \Xi_{\imdp}^\sched$; see, e.g.,~\cite{BK08}.
If $\Sched_{\imdp}$ is a singleton (i.e., $\imdp$ has only one scheduler), we omit the script $\sched$ and simply write $\Prob_{\imdp[P]}$ and $\Xi_{\imdp}$.
For MDPs $\mdp$, we use the analogous notation with subscripts $\mdp$.
\section{Conditional Reachability with Imprecise Evidence}
\label{sec:unfolding}

In this section, we treat the first two blocks of \cref{fig:Approach}.
In \cref{sec:unfolding:MDP}, we \emph{unfold} the CTMC over the times in the imprecise evidence into an MDP.
The main result of this section, \cref{thm1:risk_on_MDP}, states that the conditional reachability on the CTMC in \cref{problem} is equal to the \emph{maximal} conditional reachability probabilities in the MDP over a \emph{subset of schedulers} (those that we call \emph{consistent}; see \cref{def:consistent_schedulers}).
In \cref{sec:unfolding:conditioning}, we use results from~\cite{DBLP:conf/tacas/BaierKKM14} to determine these conditional probabilities via unconditional reachability probabilities on a transformed version of the MDP.

\subsection{Unfolding the CTMC into an MDP}
\label{sec:unfolding:MDP}

We interpret the (precisely timed) evidence $\rho = \tuple{ t_1, o_1 },\ldots,\tuple{ t_\hor, o_\hor}$ as a directed graph that encodes the trivial progression over the time steps $t_1,\ldots,t_\hor$.
\begin{definition}[Evidence graph]
    \label{def:graph}
    An \emph{evidence graph} $\timegraph = \tuple{ \nodes, \edges }$ is a directed graph where each node $t \in \nodes \subseteq \Real_{> 0}$ is a point in time, and with directed edges $\edges \subset \{ t \to t' : t,t' \in \nodes \}$, such that $t' > t$ for all $t \to t' \in \edges$.
\end{definition}
The graph $\timegraph_\rho = \tuple{ \nodes_\rho, \edges_\rho }$ for the precise evidence $\rho$ has nodes $\nodes_\rho = \{ 0, t_1, \ldots, t_\hor, t_\star \}$ and edges $\edges_\rho = \{ t_{i-1} \to t_i : i = 2,\ldots,\hor \} \cup \{ 0 \to  t_1, t_\hor \to t_\star \}$.
As illustrated in \cref{fig:inspection_graph}, the graph $\timegraph_\rho$ has exactly one path, which follows the time points $t_1,\ldots,t_d$ of the evidence $\rho$ itself.
Likewise, we model the imprecise evidence $\imphist$ as a graph $\timegraph_\imphist$ which is the union of all graphs $\timegraph_\rho$ for all instances $\rho \in \imphist$, i.e.,
\begin{equation}
        \timegraph_\imphist 
        = \tuple{ \nodes_\imphist, \edges_\imphist }
        = \cup_{\rho \in \imphist} (\timegraph_\rho) 
        = \tuple{ \cup_{\rho \in \imphist} (\nodes_\rho), \cup_{\rho \in \imphist} (\edges_\rho) }.
\end{equation}
If $\imphist$ has infinitely many instances, then $\timegraph_\imphist$ has infinite branching.
Every path $t_0 t_1 \ldots t_\hor t_\star$ through graph $\timegraph_\imphist$ corresponds to the time points of the precise evidence $\rho = \tuple{ t_1, o_1 },\ldots,\tuple{ t_\hor, o_\hor} \in \imphist$ (and vice versa).

We denote the successor nodes of $t \in \nodes$ by $\post(t) = \{ t' \in \nodes : t \to t' \in \edges \}$. 
For example, the graph in \cref{fig:inspection_graph} has $\post(0) = t_1$, $\post(t_1) = t_2$ and $\post(t_2) = t_\star$.
We introduce the \emph{unfolding operator} $\mdp = \Unfold(\ctmc, \timegraph)$, which takes a CTMC $\ctmc$ and a graph $\timegraph$, and returns the \emph{unfolded MDP} $\mdp$ defined as follows.
\begin{definition}[Unfolded MDP]
    \label{def:unfolded_mdp}
    For a CTMC $\ctmc = \tuple{ S, s_I, \Delta, E, C, L }$ and a graph $\timegraph = \tuple{ \nodes, \edges }$, the \emph{unfolded MDP} $\Unfold(\ctmc, \timegraph) = \tuple{ Q, q_I, A, P }$ has states
    states $Q = S \times \nodes$, initial state $q_I = \tuple{ s_I, 0 }$,
    actions $A = \nodes$, and
    transition function $P$, which is defined for all $ \tuple{ s,t } \in Q$, $t' \in \post(t)$, $s' \in S$ as
    \begin{align}
        P\big( \tuple{ s,t }, t', \tuple{ s',t' } \big) = \begin{cases}
            \pr_s(t' - t)(s') \enskip & \text{if } t' \neq t_\star,
            \\
            \mathbbm{1}_{(s = s')} & \text{if } t' = t_\star,
        \end{cases}
    \end{align}
\end{definition}
The unfolding of the CTMC in \cref{fig:ctmc} over the graph in \cref{fig:inspection_graph} is shown in \cref{fig:unfolding}.
A state $\tuple{s,t} \in Q$ in the unfolded MDP is interpreted as being in CTMC state $s \in S$ at time $t$.
In state $\tuple{s,t}$, the set of enabled actions is $A(\tuple{s,t}) = \post(t) \subset \nodes$, and taking an action $t' \in \post(t)$ corresponds to \emph{deterministically} jumping to time $t'$. 
The effect of this action is \emph{stochastic} and determines the next CTMC state.
The transition probability $P( \tuple{ s,t }, t', \tuple{ s',t' } )$ for $t' \neq t_\star$ models the probability of starting in CTMC state $s \in S$ and being in state $s' \in S$ after time $t'-t$ has elapsed, which is precisely the transient probability $\pr_s(t' - t)(s')$ defined in \cref{sec:preliminaries}.
Finally, the (terminal) states $\tuple{s,t_\star}$ for all $s \in S$ are absorbing.

\paragraph{Interpretation of schedulers}
Every instance $\rho \in \imphist$ of the imprecise evidence $\imphist = \tuple{ \impT_1, o_1 }, \ldots, \tuple{ \impT_\hor, o_\hor }$ corresponds to fixing a precise time $t_i \in \impT_i$ for all $i=1,\ldots,\hor$.
For each such $\rho \in \imphist$, there exists a scheduler $\sched \in \Sched_\mdp$ for MDP $\mdp = \Unfold(\ctmc, \timegraph_\imphist)$ that induces a Markov chain which only visits those time points $t_1,\ldots,t_\hor$.
We call such a scheduler $\sched$ \emph{consistent} with the evidence $\rho$.
\begin{definition}[Consistent scheduler]
    \label{def:consistent_schedulers}
    A scheduler $\sched \in \Sched_\mdp$ is consistent with $\rho = \tuple{ t_1, o_1 }, \ldots, \tuple{ t_\hor, o_\hor } \in \imphist$, written as $\sched \consistent \rho$, if for all CTMC states $s \in S$:
    \begin{equation}
        \label{eq:def:consistent_schedulers}
        \sched(\tuple{ s, 0 }) = t_1,
        \quad
        \sched(\tuple{ s, t_{i} }) = t_{i+1} 
        \, \forall i \in \{0,\ldots,d-1\},
        \quad 
        \sched(\tuple{ s, t_d }) = t_\star.
    \end{equation}
    We denote the set of all consistent schedulers by $\Sched_\mdp^\mathsf{con} \subseteq \Sched_\mdp$.
\end{definition}
A consistent scheduler chooses the same action $\sched(\tuple{ s, t }) = \sched(\tuple{ s', t' })$ in any two MDP states $\tuple{ s, t }, \tuple{ s', t' } \in Q$ for which $t = t'$.
There is a one-to-one correspondence between choices $\rho \in \imphist$ and consistent schedulers: for every $\rho \in \imphist$, there exists a scheduler $\sched \in \Sched_\mdp^\mathsf{con}$ such that $\sched \sim \rho$, and vice versa.

\begin{example}
    \label{example:inconsistent_scheduler}
    Consider imprecise evidence $\imphist = \tuple{[0.2, 0.8], \bigcdot{empty_color}}, \tuple{[1.4,2.1], \bigcdot{nonempty_color}}$ for the CTMC in \cref{example:queue}.
    A scheduler with $\sched(\tuple{s_0, 0.4}) = 1.5$, $\sched(\tuple{s_1, 0.4}) = 1.8$ is inconsistent as it chooses different actions in MDP states with the same time.
\end{example}

\begin{remark}
    \label{remark:schedulers_precise_evidence}
    The unfolded MDP $\mdp' = \Unfold(\ctmc, \timegraph_\rho)$ for the precise evidence $\rho$ has only a single action enabled in every state (i.e., $\mdp'$ directly reduces to a discrete-time Markov chain).
    Hence, $\mdp'$ has only one scheduler, and $\Sched_{\mdp'}^\mathsf{con} = \Sched_{\mdp'}$.
\end{remark}

\paragraph{Conditional reachability on unfolded MDP}
As a main result, we show that $\Weight(\imphist)$ in \cref{problem} can be expressed as maximizing conditional reachability probabilities in the unfolded MDP $\mdp$ over the consistent schedulers $\Sched_\mdp^\mathsf{con} \subset \Sched_\mdp$.
\begin{theorem}
    \label{thm1:risk_on_MDP}
    For a CTMC $\ctmc$ and the imprecise evidence $\imphist$ with graph $\timegraph_\imphist$, let $\mdp = \Unfold(\ctmc, \timegraph_\imphist)$ be the unfolded MDP.
    Then, using the notation from \cref{sec:preliminaries:MDPs} (for the probability measure $\Prob_\mdp^\sched$ over paths $\xi \in \Xi_\mdp^\sched$), \cref{eq:problem} is rewritten as
    \begin{equation}
    \begin{split}
        \label{eq:thm1:risk_on_MDP}
        \Weight(\imphist) 
        &= 
        \sup_{\sched \in \Sched_\mdp^\mathsf{con}} \, \sum_{s \in S} \Prob_\mdp^\sched( \Finally\tuple{ s, t_\star } \mid [\xi \models \rho, \,\, \sched \sim \rho]) \cdot \weight(s).
    \end{split}
    \end{equation}
\end{theorem}
\begin{proof}
    The proof is in 
    \ifappendix
        \cref{appendix:proof:trans} 
    \else
        \cite[Appendix~A]{Badings2024TACAS:extended}
    \fi
    and shows that for every instance $\rho \in \imphist$, the conditional transient probabilities in the CTMC are equivalent to conditional reachability probabilities in the unfolded MDP under a $\sched \sim \rho$ consistent to $\rho$.
    \qed
\end{proof}

\begin{figure}[t]
\begin{subfigure}[b]{0.49\linewidth}
    \centering
    \iftikzcompile
    \begin{tikzpicture}[node distance=0.8cm]

    \def\xshift{0.6cm}
      \node[state, initial,initial where=left, initial text=] (s0t0) {$s_0, 0$};
      \node[state, below of=s0t0] (s1t0) {$s_1, 0$};
      \node[state, below of=s1t0] (s2t0) {$s_2, 0$};
      \node[state, empty label, right of=s0t0, xshift=\xshift] (s0t1) {$s_0, t_1$};
      \node[state, nonempty label, below of=s0t1] (s1t1) {$s_1, t_1$};
      \node[state, nonempty label, below of=s1t1] (s2t1) {$s_2, t_1$};
      \node[state, empty label, right of=s0t1, xshift=\xshift] (s0t2) {$s_0, t_2$};
      \node[state, nonempty label, below of=s0t2] (s1t2) {$s_1, t_2$};
      \node[state, nonempty label, below of=s1t2] (s2t2) {$s_2, t_2$};
      \node[state, right of=s0t2, xshift=\xshift] (s0star) {$s_0, t_\star$};
      \node[state, below of=s0star] (s1star) {$s_1, t_\star$};
      \node[state, below of=s1star] (s2star) {$s_2, t_\star$};
      \begin{pgfonlayer}{bg}    
      \path[->,shorten <= -2pt+\pgflinewidth] 
                (s0t0) edge [] node [above] {} (s0t1)
                (s0t0) edge [] node [above] {} (s1t1)
                (s0t0) edge [] node [above] {} (s2t1)
                (s1t0) edge [] node [above] {} (s0t1)
                (s1t0) edge [] node [above] {} (s1t1)
                (s1t0) edge [] node [above] {} (s2t1)
                (s2t0) edge [] node [above] {} (s0t1)
                (s2t0) edge [] node [above] {} (s1t1)
                (s2t0) edge [] node [above] {} (s2t1)
                %
                (s0t1) edge [red, thick, bend right=10] node [] {} (s0t0)
                (s1t1) edge [] node [above] {} (s0t2)
                (s1t1) edge [] node [above] {} (s1t2)
                (s1t1) edge [] node [above] {} (s2t2)
                (s2t1) edge [] node [above] {} (s0t2)
                (s2t1) edge [] node [above] {} (s1t2)
                (s2t1) edge [] node [above] {} (s2t2)
                %
                (s0t2) edge [red, thick, bend right=25] node [] {} (s0t0)
                (s1t2) edge [] node [above] {} (s1star)
                (s2t2) edge [] node [above] {} (s2star);
    \end{pgfonlayer}
    \end{tikzpicture}
    \fi
    \caption{For $\rho = \langle t_1, \bigcdot{nonempty_color} \rangle, \langle t_2, \bigcdot{nonempty_color} \rangle$.}
    \label{fig:unfolded_inspection1}
    \end{subfigure}
\hfill 
\begin{subfigure}[b]{0.49\linewidth}
    \centering
    \iftikzcompile
    \begin{tikzpicture}[node distance=0.8cm]

    \def\xshift{0.6cm}
      \node[state, initial,initial where=left, initial text=] (s0t0) {$s_0, 0$};
      \node[state, below of=s0t0] (s1t0) {$s_1, 0$};
      \node[state, below of=s1t0] (s2t0) {$s_2, 0$};
      \node[state, empty label, right of=s0t0, xshift=\xshift] (s0t1) {$s_0, t_1$};
      \node[state, nonempty label, below of=s0t1] (s1t1) {$s_1, t_1$};
      \node[state, nonempty label, below of=s1t1] (s2t1) {$s_2, t_1$};
      \node[state, empty label, right of=s0t1, xshift=\xshift] (s0t2) {$s_0, t_2$};
      \node[state, nonempty label, below of=s0t2] (s1t2) {$s_1, t_2$};
      \node[state, nonempty label, below of=s1t2] (s2t2) {$s_2, t_2$};
      \node[state, right of=s0t2, xshift=\xshift] (s0star) {$s_0, t_\star$};
      \node[state, below of=s0star] (s1star) {$s_1, t_\star$};
      \node[state, below of=s1star] (s2star) {$s_2, t_\star$};
      \begin{pgfonlayer}{bg}    
      \path[->,shorten <= -2pt+\pgflinewidth] 
                (s0t0) edge [] node [above] {} (s0t1)
                (s0t0) edge [] node [above] {} (s1t1)
                (s0t0) edge [] node [above] {} (s2t1)
                (s1t0) edge [] node [above] {} (s0t1)
                (s1t0) edge [] node [above] {} (s1t1)
                (s1t0) edge [] node [above] {} (s2t1)
                (s2t0) edge [] node [above] {} (s0t1)
                (s2t0) edge [] node [above] {} (s1t1)
                (s2t0) edge [] node [above] {} (s2t1)
                (s0t1) edge [] node [above] {} (s0t2)
                (s0t1) edge [] node [above] {} (s1t2)
                (s0t1) edge [] node [above] {} (s2t2)
                (s1t1) edge [red, thick, bend right=20] node [] {} (s0t0)
                (s2t1) edge [red, thick, bend left=65] node [] {} (s0t0)
                %
                (s0t2) edge [red, thick, bend right=25] node [] {} (s0t0)
                (s1t2) edge [] node [above] {} (s1star)
                (s2t2) edge [] node [above] {} (s2star);
    \end{pgfonlayer}
    \end{tikzpicture}
    \fi
    \caption{For $\rho = \langle t_1, \bigcdot{empty_color} \rangle, \langle t_2, \bigcdot{nonempty_color} \rangle$.}
    \label{fig:unfolded_inspection2}
\end{subfigure}
    \caption{The unfolded MDP from \cref{fig:unfolding} conditioned on different precise evidences. States that do not agree with the evidence are looped back to the initial state.}
    \label{fig:unfolded_inspection}
\end{figure}
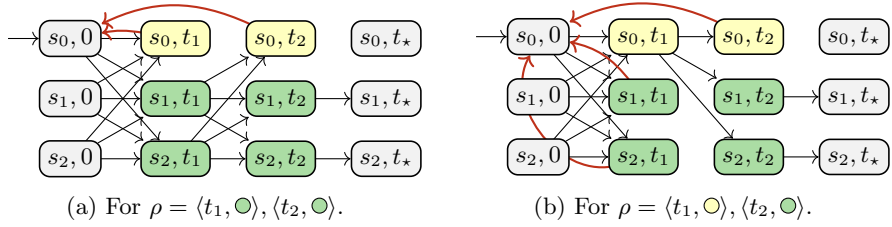

\subsection{Computing conditional probabilities in MDPs}
\label{sec:unfolding:conditioning}

We describe a transformation of the unfolded MDP to compute the conditional reachability probabilities in \cref{eq:thm1:risk_on_MDP}.
Intuitively, we \emph{refute} all paths through the MDP that do not agree with the labels in the evidence.
Specifically, we find the subset of MDP states $Q_\mathsf{reset}(\imphist) \subset Q$ that disagree with the evidence, defined as 
\begin{equation}
    \label{eq:inconsistent_states}
    Q_\mathsf{reset}(\Omega) = \big\{ \tuple{ s,t } \in Q : \, L(s) \neq \Omega(t) \big\} \subset Q.
\end{equation}
We \emph{reset} all states in $Q_\mathsf{reset}(\Omega)$ by adding transitions back to the initial state with probability one.
Formally, we define the \emph{conditioned MDP} $\mdp_{|\imphist}$ as follows.
\begin{definition}[Conditioned MDP]
    \label{def:conditioned_MDP}
    For $\mdp = \Unfold(\ctmc, \timegraph_\imphist) = \tuple{ Q, q_I, A, P }$, the conditioned MDP $\mdp_{|\imphist} = \tuple{ Q, q_I, A, P_{|\Omega} }$ has the same states and actions, but the transition function is defined for all $ \tuple{ s,t } \in Q$, $t' \in \post(t)$, $s' \in S$ as
    \begin{align}
        \label{eq:conditioned_MDP}
        P_{|\Omega} \big( \tuple{ s,t }, t', \tuple{ s', t' } \big) = \begin{cases}
            P \big( \tuple{ s,t }, t', \tuple{ s', t' } \big) \! \! & \text{if } \tuple{ s,t } \notin Q_\mathsf{reset}(\imphist),
            \\
            \mathbbm{1}_{(s' = s_I)} & \text{if } \tuple{ s,t } \in Q_\mathsf{reset}(\imphist).
        \end{cases}
    \end{align}
\end{definition}
Two examples of conditioning on precise evidence are shown in \cref{fig:unfolded_inspection}.
Compared to \cref{fig:unfolding}, we removed all probability mass over paths that are not consistent with the evidence and normalized the probabilities for all other paths.
The following result from \cite{DBLP:conf/tacas/BaierKKM14} shows that conditional reachabilities in the unfolded MDP are equal to \emph{unconditional} reachabilities in the conditioned MDP.
\begin{lemma}[Thm. 1 in \cite{DBLP:conf/tacas/BaierKKM14}]
    \label{lemma:Baier}
    For the imprecise evidence $\imphist$, unfolded MDP $\mdp = \Unfold(\ctmc, \timegraph_\imphist)$, and conditioned MDP $\mdp_{|\imphist}$ defined by \cref{def:conditioned_MDP}, it holds that
    \begin{align}
        \label{eq:lemma:Baier}
        \Prob_\mdp^\sched( \Finally\tuple{ s, t_\star } \mid [\xi \models \rho, \,\, \sched \sim \rho])
        \, = \,
        \Prob_{\mdp_{|\imphist}}^\sched( \Finally\tuple{ s, t_\star } ) \quad \forall \sched \in \Sched_\mdp \,\, \forall s \in S.
    \end{align}
\end{lemma}

Finally, combining \cref{lemma:Baier} with \cref{thm1:risk_on_MDP} directly expresses the conditional reachability $\Weight(\imphist)$ in terms of reachability probabilities on the conditioned MDP.
\begin{theorem}
    \label{thm2:problem_solution}
    Given a CTMC $\ctmc$, a state-weight function $\weight$, and the imprecisely timed evidence $\imphist$, let $\mdp = \Unfold(\ctmc, \timegraph_\imphist)$.
    Then, it holds that
    \begin{equation}
    \begin{split}
        \label{eq:thm2}
        \Weight(\imphist)
        &= 
        \sup_{\sched \in \Sched_\mdp^\mathsf{con}} \sum_{s \in S} \Prob^\sched_{\mdp_{|\imphist}}( \Finally\tuple{ s, t_\star } ) \cdot \weight(s).
    \end{split}
    \end{equation}
\end{theorem}

Solving Problem~1 with precisely timed evidence is now straightforward by solving a finite DTMC, see \cref{remark:schedulers_precise_evidence}. 
Furthermore, if the imprecise evidence has finitely many instances, then the MDP is finite. 
A naive approach to optimize over the consistent schedulers is enumeration, which we discuss in details \cref{sec:algorithm}.

\begin{remark}[Variations on \cref{problem}]
\label{remark:variation}
With minor modifications to our approach, we can compute, e.g., the likelihood that a CTMC generates precise evidence $\rho$.
Concretely, we define a transformed version $\mdp_\rho$ of the unfolded MDP in which all states in $Q_\mathsf{reset}$ are absorbing.
We discuss this variation in 
\ifappendix
    \cref{appendix:path_probabilities}.
\else
    \cite[Appendix~C]{Badings2024TACAS:extended}
\fi
\end{remark}
\section{Abstraction of Conditioned MDPs}
\label{sec:abstraction}

For imprecisely timed evidence with \emph{infinitely many instances} (e.g., imprecise timings over intervals), the conditioned MDP from \cref{sec:unfolding} has infinitely many states and actions.
In this section, we treat block (3) of \cref{fig:Approach} and propose an abstraction of this continuous MDP into a finite interval MDP (iMDP).
Similar to game-based abstractions~\cite{DBLP:journals/fmsd/KattenbeltKNP10,DBLP:conf/tacas/HahnHWZ10,DBLP:conf/qest/HahnNPWZ11}, we capture abstraction errors as nondeterminism in the transition function of the iMDP.
Robust reachability probabilities in the iMDP yield sound bounds on the conditional reachability $\Weight(\imphist)$.
The crux of our abstraction is to create a finite \emph{partition} of the (infinite) sets of uncertain timings in the evidence, as illustrated by \cref{fig:time_partition}.
\begin{figure}[t]
\begin{subfigure}[b]{0.47\linewidth}
    \centering
    \iftikzcompile
    \begin{tikzpicture}[xscale=1.5, yscale=0.8]
    
    \draw[shift={(0,0)},color=black] (0pt,0pt) -- (0pt,-3pt) node[below] {$0$};
    \draw[shift={(1,0)},color=black] (0pt,0pt) -- (0pt,-3pt) node[below] {};
    \draw[shift={(2,0)},color=black] (0pt,0pt) -- (0pt,-3pt) node[below] {};

    \draw[thick, shift={(0,0)},color=red] (0pt,5pt) -- (0pt,-0pt);
    \draw[shift={(2.8,0)},color=black] (0pt,0pt) -- (0pt,-3pt) node[below] 
    {$t_\star$};
    \draw[thick, shift={(2.8,0)},color=red] (0pt,5pt) -- (0pt,-0pt);

    \foreach \from/\to/\txt in 
    {0.2/0.8/$\tilde\impT_1^1$,
     1.4/2.1/$\tilde\impT_2^1$}
    {
    \fill[red!20] (\from,0pt) rectangle (\to,5pt);
    \draw[thick, shift={(\from,0)},color=red] (0pt,5pt) -- (0pt,-0pt);
    \draw[thick, shift={(\to,0)},color=red] (0pt,5pt) -- (0pt,-0pt);
    \draw [red,decorate,decoration={brace,amplitude=5pt,mirror,raise=2ex}] (\from+0.02,0.2) -- (\to-0.02,0.2) node[red, midway, below, yshift=-0.4cm]{\txt};
    }

    \draw[-latex] (0,0) -- (3,0) ; 
    \foreach \x in  {0} 
    \draw[shift={(\x,0)},color=black] (0pt,3pt) -- (0pt,-3pt);
    
    \path[-stealth] 
    \foreach \x/\xx in {0/0.5,0.5/1.75,1.75/2.8}{
    (\x,0.2) edge [bend left=70] node [] {} (\xx-0.02,0.2)
    };
    
    \end{tikzpicture}
    \fi
    \caption{Coarsest time partition.}
    \label{fig:time_partition_coarse}
    \end{subfigure}
\hfill 
\begin{subfigure}[b]{0.47\linewidth}
    \centering
    \iftikzcompile
    \begin{tikzpicture}[xscale=1.5, yscale=0.8]
    
    \draw[shift={(0,0)},color=black] (0pt,0pt) -- (0pt,-3pt) node[below] {$0$};
    \draw[shift={(1,0)},color=black] (0pt,0pt) -- (0pt,-3pt) node[below] {};
    \draw[shift={(2,0)},color=black] (0pt,0pt) -- (0pt,-3pt) node[below] {};

    \draw[thick, shift={(0,0)},color=red] (0pt,5pt) -- (0pt,-0pt);
    \draw[shift={(2.8,0)},color=black] (0pt,0pt) -- (0pt,-3pt) node[below] 
    {$t_\star$};
    \draw[thick, shift={(2.8,0)},color=red] (0pt,5pt) -- (0pt,-0pt);

    \foreach \from/\to/\txt in 
    {0.2/0.5/$\tilde\impT_1^1$,
     0.5/0.8/$\tilde\impT_1^2$,
     1.4/1.75/$\tilde\impT_2^1$,
     1.75/2.1/$\tilde\impT_2^2$}
    {
    \fill[red!20] (\from,0pt) rectangle (\to,5pt);
    \draw[thick, shift={(\from,0)},color=red] (0pt,5pt) -- (0pt,-0pt);
    \draw[thick, shift={(\to,0)},color=red] (0pt,5pt) -- (0pt,-0pt);
    \draw [red,decorate,decoration={brace,amplitude=5pt,mirror,raise=2ex}] (\from+0.02,0.2) -- (\to-0.02,0.2) node[red, midway, below, yshift=-0.4cm]{\txt};
    }

    \draw[-latex] (0,0) -- (3,0) ; 
    \foreach \x in  {0} 
    \draw[shift={(\x,0)},color=black] (0pt,3pt) -- (0pt,-3pt);
    
    \path[-stealth] 
    \foreach \x/\xx in 
    {0/0.35,
     0/0.65,
     0.40/1.55,
     0.35/1.90,
     0.70/1.50,
     0.65/1.85,
     1.6/2.8,
     1.925/2.8}
    {
    (\x,0.2) edge [bend left=70] node [] {} (\xx-0.02,0.2)
    };
    
    \end{tikzpicture}
    \fi
    \caption{Refined time partition.}
    \label{fig:time_partition_refined}
\end{subfigure}
    \caption{Two partitions of imprecise evidence $\imphist = \tuple{[0.2, 0.8], o_1}, \tuple{[1.4, 2.1], o_2}$. The partition in (a) consists of two elements, such that $\tilde{\impT}_1^1 = [0.2, 0.8]$ and $\tilde{\impT}_2^1 = [1.4, 2.1]$, where (b) refines this to $\tilde{\impT}_1^1 \cup \tilde{\impT}_1^2 = [0.2, 0.8]$ and $\tilde{\impT}_2^1 \cup \tilde{\impT}_2^2 = [1.4, 2.1]$.}
    \label{fig:time_partition}
\end{figure}
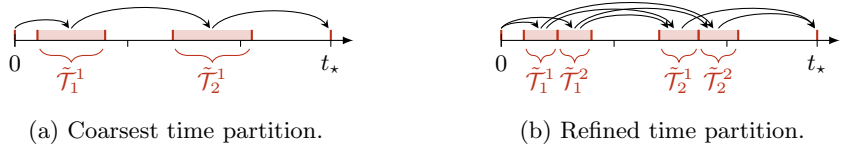
\begin{definition}[Time partition]
    \label{def:partition}
    A \emph{time partition} $\partition$ of the imprecise evidence $\imphist = \tuple{\impT_1, o_1}, \ldots, \tuple{\impT_\hor, o_\hor}$ is a set $\partition = \cup_{i=1}^d \Partition(\impT_i) \cup \{0, t_\star\}$, where each $\Partition(\impT_i) = \{\impT_i^1, \ldots, \impT_i^{n_i}\}$ is a finite partition\footnote{A partition $\Partition(X) = (X_1, \ldots, X_n)$ covers $X$ (i.e., $X = \cup_{i=1}^n X_i$) and the interior of each element is disjoint (i.e., $\text{int}(X_i) \cap \text{int}(X_j) = \varnothing, \enskip i,j \in \{1,\ldots,n\}, \, i \neq j$).} of $\impT_i$ into $n_i \in \N$ elements.
\end{definition}
With abuse of notation, the element of $\partition$ containing time $t$ is $\partition(t) \in \partition$, and $\partition^{-1}(\psi) = \{ t : \partition(t) = \psi \}$ is the set of times mapping to $\psi \in \partition$.
As shown by \cref{fig:time_partition}, for each $i \in \{1,\ldots,\hor\}$, the sets $\tilde\impT_i^1, \ldots, \tilde\impT_i^{n_i}$ are a partition of the set $\impT_i$.

To illustrate the abstraction, let $\tuple{s,t} \xrightarrow{t' : P'} \tuple{s',t'}$ denote the MDP transition from state $\tuple{s,t} \in Q$, under action $t' \in A(\tuple{s,t})$ to state $\tuple{s',t'} \in Q$, which has probability $P'$.
With this notation, we can express any MDP path as
\begin{alignat}{7}
    \label{eq:MDP_path}
    \tuple{s_I, 0} 
    &\xrightarrow{t : P}
    &&\tuple{s,t}
    &&\xrightarrow{t' : P'}
    &&\tuple{s',t'}
    &&\xrightarrow{t'' : P''}
    \cdots
    &&\xrightarrow{t''' : P'''}
    &&\tuple{s,t_\star}. \\
    \intertext{For every element $\psi \in \partition$ of partition $\partition$, the abstraction merges all MDP states $\tuple{s,t} \in Q$ for which the time $t$ belongs to the element $\psi$, that is, for which $t \in \partition^{-1}(\psi)$.
    Thus, we merge infinitely many MDP states into finitely many abstract states.
    The MDP path in \cref{eq:MDP_path} matches the next path in the abstraction:}
    \label{eq:iMDP_path}
    \tuple{s_I, 0} 
    &\xrightarrow{\impT : \calP} \,
    &&\tuple{s,\impT}
    &&\xrightarrow{\impT' : \calP'} \,
    &&\tuple{s',\impT'}
    &&\xrightarrow{\impT'' : \calP''}\,
    \cdots
    &&\xrightarrow{\impT''' : \calP'''} \,
    &&\tuple{s,t_\star},
\end{alignat}
where each $t \in \impT$, and each $\calP$ is a \emph{set of probabilities}.
The abstraction contains the behavior of the continuous MDP if $P \in \calP$ at every step in \cref{eq:MDP_path,eq:iMDP_path}, see, e.g.,~\cite{DBLP:conf/lics/JonssonL91}.
The following iMDP abstraction satisfies these requirements.
\begin{definition}[iMDP abstraction]
    \label{def:abstraction}
    For a conditioned MDP $\mdp_{|\imphist} = \tuple{ Q, q_I, A, P }$ and a time partition $\partition$ of $\imphist$, the iMDP abstraction $\imdp = \Abstract(\mdp_{|\imphist}, \partition) = \tuple{ \tilde{Q}, q_I, \tilde{A}, \calP }$, with states $\tilde{Q} = \big\{ \tuple{s, \partition(t)} : \tuple{s,t} \in Q \big\}$, actions $\tilde{A} = \big\{ \partition(t) : t \in A \big\}$, and uncertain transition function $\calP$ defined for all $\tuple{ s, \impT }, \tuple{ s', \impT' } \in \tilde{Q}$ as
        \begin{equation}
        \begin{split}
            \label{eq:probability_intervals}
            \calP \big( \tuple{ s,\impT }, \impT', \tuple{ s',\impT' } \big) \enskip = \enskip
            \closure\Big(
            \!\!\!\!\! \bigcup_{t \in \partition^{-1}(\impT), t' \in \partition^{-1}(\impT')}  \!\!\!\!\!
            P\big( \tuple{ s,t }, t', \tuple{ s', t'} \big)\enskip
            \Big),
        \end{split}
        \end{equation}
        where $\closure(x) = [\min(x), \max(x)]$ is the interval closure of $x$.
\end{definition}
%

\tikzset{box around red/.style={
    thick, draw=red, rounded corners,
    prefix after command= {\pgfextra{\tikzset{every label/.style={red, anchor=south}}}}
}}
\tikzset{box around blue/.style={
    thick, draw=plotblue, rounded corners,
    prefix after command= {\pgfextra{\tikzset{every label/.style={plotblue, anchor=south}}}}
}}

\newcommand{\spacing}{-0.3em}

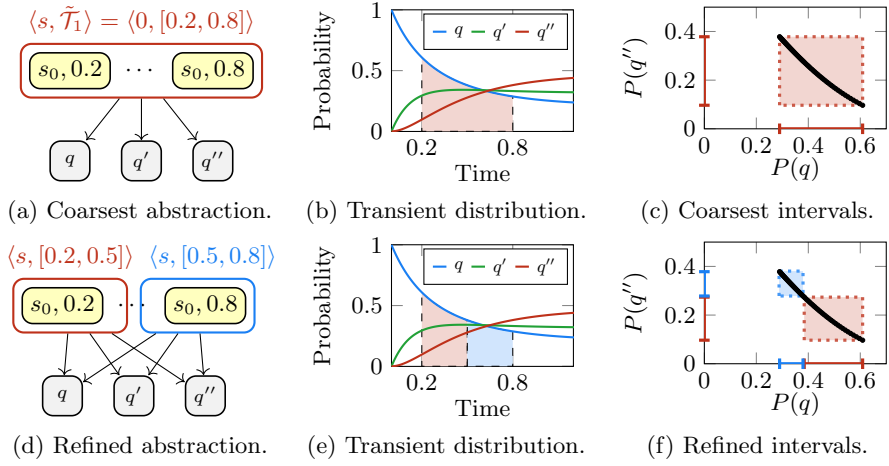
\begin{figure}[t]
\begin{subfigure}[b]{0.32\linewidth}
    \centering
    \iftikzcompile
    \begin{tikzpicture}[node distance=1.2cm]

    \def\xshift{0.3cm}
    \def\xshiftB{-0.25cm}
    \def\yshift{0cm}
      \node[state, empty label, xshift=\xshift] (s0t1a) {$s_0, 0.2$};
      \node[right of=s0t1a, xshift=\xshiftB] (s0t1b) {$\cdots$};
      \node[state, empty label, right of=s0t1b, xshift=\xshiftB] (s0t1c) {$s_0, 0.8$};
      \node[box around red, fit=(s0t1a) (s0t1c), label={above:${\tuple{s, \tilde\impT_1} = \tuple{ 0, [0.2, 0.8] }}$}] (s0T1) {};
      %
      \node[state, medium, below of=s0t1a, yshift=\yshift] (s0t2) {$q$};
      \node[state, medium, below of=s0t1b, yshift=\yshift] (s1t2) {$q'$};
      \node[state, medium, below of=s0t1c, yshift=\yshift] (s2t2) {$q''$};
      
      \path[->] (s0T1) edge [] node [left] {} (s0t2)
                (s0T1) edge [] node [left] {} (s1t2)  
                (s0T1) edge [] node [right] {} (s2t2);   
                
    \end{tikzpicture}
    \fi
    \caption{Coarsest abstraction.}
\label{fig:abstraction_states}
\end{subfigure}
\begin{subfigure}[b]{0.33\linewidth}
    \centering
    \iftikzcompile
    \begin{tikzpicture}[baseline]
\definecolor{color1}{rgb}{0.1,0.498039215686275,0.9549019607843137}
\begin{axis}[
        width=4cm,
        height=3.2cm,
        ymin=-0.001,
        ymax=1.001,
        xmin=0,
        xmax=120,
        xlabel={Time},
        xtick={20, 80},
        xticklabels={0.2, 0.8},
        x label style={at={(axis description cs:0.5,-0.2)}},
        ylabel={Probability},
        ytick={0,0.5,1.0},
        legend columns=-1,
        legend entries={$q$, $q'$, $q''$},
        legend pos=north east,
        legend image post style={xscale=0.3},
        legend style={nodes={scale=0.75, transform shape},
                      column sep=0.3ex,
                      row sep=0pt},
]

\path[name path=axis] (axis cs:0,0) -- (axis cs:120,0);

\addplot[name path=q0, thick, mark=none, color=plotblue] table[x=t, y=q0, col sep=comma] {Figures/Abstraction/Data/ctmc_queue_curves.csv};

\addplot[name path=q1, thick, mark=none, color=darkgreen] table[x=t, y=q1, col sep=comma] {Figures/Abstraction/Data/ctmc_queue_curves.csv};

\addplot[name path=q2, thick, mark=none, color=red] table[x=t, y=q2, col sep=comma] {Figures/Abstraction/Data/ctmc_queue_curves.csv};

\addplot[draw=black, dashed, fill=red!20] fill between[of=q0 and axis,soft clip={domain=20:80}];

\end{axis}
\end{tikzpicture}
    \vspace{\spacing}
    \fi
    \caption{Transient distribution.}
\label{fig:abstraction_transient}
\end{subfigure}
\begin{subfigure}[b]{0.33\linewidth}
    \centering
    \iftikzcompile
    \begin{tikzpicture}
    \begin{axis}[%
        width=4cm,
        height=3.2cm,
        xmin=0,
        xmax=0.7,
        ymin=0.0,
        ymax=0.5,
        xlabel={$P(q)$},
        ylabel={$P(q'')$},
        xtick={0,0.2,0.4,0.6},
        ytick={0,0.2,0.4,0.6},
        x label style={at={(axis description cs:0.5,-0.16)}},
    ]
    
    \addlegendimage{dashed, very thick, color=red}
    \addlegendimage{dotted, very thick, color=plotblue}
    
    \addplot[only marks,fill=black,mark size=0.8pt] table {Figures/Abstraction/Data/points.dat}; 

    \addplot[mark=|, red, very thick] table[] { 
    x                       y
    0.28950187683227596     0 
    0.609037086538144       0 
    }; 
    \addplot[mark=-, red, very thick] table[] { 
    x                       y
    0                       0.09677975126692534
    0                       0.3783307322408702 
    }; 
    
    \coordinate (low) at (axis cs:0.28950187683227596, 0.09677975126692534) {};
    \coordinate (upp) at (axis cs:0.609037086538144, 0.3783307322408702) {};
    
    \end{axis}
    
    \begin{pgfonlayer}{bg}
        \filldraw[dotted, color=red!80, fill=red!20, very thick] (low) rectangle (upp);
    \end{pgfonlayer}
    
    \end{tikzpicture}
    \vspace{\spacing}
    \fi
    \caption{Coarsest intervals.}
\label{fig:abstraction_intervals}
\end{subfigure}


\begin{subfigure}[b]{0.32\linewidth}
    \centering
    \iftikzcompile
    \begin{tikzpicture}[node distance=1.2cm]

    \def\xshift{0.3cm}
    \def\xshiftB{-0.25cm}
    \def\yshift{0cm}
    
    \node[state, empty label, xshift=\xshift] (s0t1a) {$s_0, 0.2$};
    \node[right of=s0t1a, xshift=\xshiftB] (s0t1b) {$\cdots$};
    \node[state, empty label, right of=s0t1b, xshift=\xshiftB] (s0t1c) {$s_0, 0.8$};
    %
    \node[state, medium, below of=s0t1a, yshift=\yshift] (s0t2) {$q$};
    \node[state, medium, below of=s0t1b, yshift=\yshift] (s1t2) {$q'$};
    \node[state, medium, below of=s0t1c, yshift=\yshift] (s2t2) {$q''$};

    \coordinate (s0_left)  at ($(s0t1b) + (-0.2cm, 0cm)$);
    \coordinate (s0_right) at ($(s0t1b) + (0.2cm, 0cm)$);
    
    \node[box around red, fit=(s0t1a) (s0_left), label={[anchor=south]above:${\tuple{ s, [0.2, 0.5] }}$}] (s0T1) {};    
    \node[box around blue, fit=(s0t1c) (s0_right), label={[anchor=south, xshift=0.2cm]above:${\tuple{ s, [0.5, 0.8] }}$}] (s0T2) {};
    
    \path[->] 
        (s0T1) edge [] node [above] {} (s0t2)
        (s0T1) edge [] node [above] {} (s1t2)  
        (s0T1) edge [] node [above] {} (s2t2)
        (s0T2) edge [] node [above] {} (s0t2)
        (s0T2) edge [] node [above] {} (s1t2)  
        (s0T2) edge [] node [above] {} (s2t2);   
    \end{tikzpicture}
    \fi
    \caption{Refined abstraction.}
\label{fig:abstraction_states_refined}
\end{subfigure}
\begin{subfigure}[b]{0.33\linewidth}
    \centering
    \iftikzcompile
    \begin{tikzpicture}[baseline]
\definecolor{color1}{rgb}{0.1,0.498039215686275,0.9549019607843137}
\begin{axis}[
        width=4cm,
        height=3.2cm,
        ymin=-0.001,
        ymax=1.001,
        xmin=0,
        xmax=120,
        xlabel={Time},
        xtick={20, 80},
        xticklabels={0.2, 0.8},
        x label style={at={(axis description cs:0.5,-0.2)}},
        ylabel={Probability},
        ytick={0,0.5,1.0},
        legend columns=-1,
        legend entries={$q$, $q'$, $q''$},
        legend pos=north east,
        legend image post style={xscale=0.3},
        legend style={nodes={scale=0.75, transform shape},
                      column sep=0.3ex,
                      row sep=0pt},
]

\path[name path=axis] (axis cs:0,0) -- (axis cs:120,0);

\addplot[name path=q0, thick, mark=none, color=plotblue] table[x=t, y=q0, col sep=comma] {Figures/Abstraction/Data/ctmc_queue_curves.csv};

\addplot[name path=q1, thick, mark=none, color=darkgreen] table[x=t, y=q1, col sep=comma] {Figures/Abstraction/Data/ctmc_queue_curves.csv};

\addplot[name path=q2, thick, mark=none, color=red] table[x=t, y=q2, col sep=comma] {Figures/Abstraction/Data/ctmc_queue_curves.csv};

\addplot[draw=black, dashed, fill=red!20] fill between[of=q0 and axis,soft clip={domain=20:50}];

\addplot[draw=black, dashed, fill=plotblue!20] fill between[of=q0 and axis,soft clip={domain=50:80}];

\end{axis}
\end{tikzpicture}
    \vspace{\spacing}
    \fi
    \caption{Transient distribution.}
\label{fig:abstraction_transient_refined}
\end{subfigure}
\begin{subfigure}[b]{0.33\linewidth}
    \centering
    \iftikzcompile
    \begin{tikzpicture}
    \begin{axis}[%
        width=4cm,
        height=3.2cm,
        xmin=0,
        xmax=0.7,
        ymin=0.0,
        ymax=0.5,
        xlabel={$P(q)$},
        ylabel={$P(q'')$},
        xtick={0,0.2,0.4,0.6},
        ytick={0,0.2,0.4,0.6},
        x label style={at={(axis description cs:0.5,-0.16)}},
    ]
    
    \addlegendimage{dashed, very thick, color=red}
    \addlegendimage{dotted, very thick, color=plotblue}
    
    \addplot[only marks,fill=black,mark size=0.8pt] table {Figures/Abstraction/Data/points.dat}; 

    \addplot[mark=|, red, very thick] table[] { 
    x                       y
    0.38429912159797264     0 
    0.609037086538144       0 
    }; 
    \addplot[mark=-, red, very thick] table[] { 
    x                       y
    0                       0.09677975126692534
    0                       0.27367298878305124
    }; 

    \addplot[mark=|, plotblue, very thick] table[] { 
    x                       y
    0.28950187683227596     0 
    0.37967922538548715     0 
    }; 
    \addplot[mark=-, plotblue, very thick] table[] { 
    x                       y
    0                       0.27841114126188643
    0                       0.3783307322408702 
    }; 
    
    \coordinate (low) at (axis cs:0.38429912159797264, 0.09677975126692534) {};
    \coordinate (upp) at (axis cs:0.609037086538144, 0.27367298878305124) {};

    \coordinate (low2) at (axis cs:0.28748697365694736, 0.27841114126188643) {};
    \coordinate (upp2) at (axis cs:0.37967922538548715, 0.3807000222790507) {};
    
    \end{axis}

    \begin{pgfonlayer}{bg}
        \filldraw[dotted, color=red!80, fill=red!20, very thick] (low) rectangle (upp);
    \end{pgfonlayer}
    
    \begin{pgfonlayer}{bg}
        \filldraw[dotted, color=plotblue!80, fill=plotblue!20, very thick] (low2) rectangle (upp2);
    \end{pgfonlayer}
    
    \end{tikzpicture}
    \vspace{\spacing}
    \fi
    \caption{Refined intervals.}
\label{fig:abstraction_intervals_refined}
\end{subfigure}
\caption{Abstraction of an infinite set of MDP states for all times $t \in [0.2, 0.8]$ into (a) a single iMDP state $\tuple{s, [0.2, 0.8]}$ with probability intervals that overapproximate the transient distribution (b) as the rectangular set in (c), where the line shows the MDP transition probabilities for all $t \in [0.2,0.8]$.
The refinement (d) into two iMDP states $\tuple{s, [0.2, 0.5]}$ and $\tuple{s, [0.5, 0.8]}$ splits the approximation of the transient (e) into the two (less conservative) rectangular sets in (f).}
\label{fig:abstraction}
\end{figure}
An abstraction under the coarse time partition from \cref{fig:time_partition} is shown in \cref{fig:abstraction_states}.
The transition probabilities for each MDP state are defined by transient probabilities for the CTMC.
Thus, the uncertain transition function $\calP$ of the iMDP overapproximates these transient probabilities over a \emph{range of times} (as shown in \cref{fig:abstraction_transient}), yielding probability intervals as in \cref{fig:abstraction_intervals}.

\paragraph{Conditional reachability on iMDP}
We show that the iMDP abstraction can be used to obtain sound upper and lower bounds on the conditional reachability $\Weight(\imphist)$.
Let $\Weight_\imdp(\tilde{P}, \sched) \geq 0$ denote the value for the MDP $\imdp[\tilde{P}]$ induced by iMDP $\imdp$ under transition function $\tilde{P}$, and with scheduler $\sched \in \Sched_\imdp$:
\begin{equation}
    \Weight_\imdp(\tilde{P}, \sched) \coloneqq \sum_{s \in S} \Prob_{\imdp[\tilde{P}]}^\sched( \Finally \tuple{ s, t_\star } ) \cdot \weight(s).
\end{equation}
The next theorem, proven in 
\ifappendix
    \cref{appendix:proof:sandwich},
\else
    \cite[Appendix~B]{Badings2024TACAS:extended},
\fi
is the main result of this section.
\begin{theorem}
    \label{thm:sandwich}
    Let $\imdp = \Abstract(\mdp_{|\imphist}, \partition)$ be the iMDP abstraction for a conditioned MDP $\mdp_{|\imphist}$ and a time partition $\partition$ of $\imphist$.
    Then, it holds that
    \begin{equation}
    \begin{split}
        \max_{\sched \in \Sched_{\imdp}^\mathsf{con}} \min_{\tilde{P} \in \calP} 
        \Weight_\imdp(\tilde{P},\sched)
        \leq
        \Weight(\imphist)
        \leq
        \max_{\sched \in \Sched_{\imdp}^\mathsf{con}} \max_{\tilde{P} \in \calP} 
        \Weight_\imdp(\tilde{P},\sched).
        \label{eq:thm:sandwich}
    \end{split}
    \end{equation}
\end{theorem}

\paragraph{Construction of the iMDP}
We want to construct the abstract iMDP directly from the CTMC without first constructing the continuous MDP $\mdp_{|\imphist}$.
Consider computing the probability interval $\calP ( \tuple{ s,\impT }, \impT', \tuple{ s',\impT' } )$ for the iMDP transition from state $\tuple{s,\impT}$ to $\tuple{s',\impT'}$.
This interval is given by the minimum and maximum transient probabilities $\Pr_s(t' - t)(s')$ over all $t \in \impT$ and $t' \in \impT'$.
However, the problem is that the transient probabilities are not monotonic over time in general (see \cref{fig:abstraction_transient}), so it is unclear how to compute this interval.

Instead, we compute upper and lower bounds for the transient probabilities.
Let $\munderbar{t} = \min(\impT)$ and $\bar{t} = \max(\impT)$.
An upper bound on the transient probability is given by the probability to reach $s'$ from $s$ at \emph{some} time $t'-t$, $t \in \impT$, $t' \in \impT'$:
\begin{equation}
    \label{eq:transient_upper}
    \sup_{t \in \impT, t' \in \impT'} {\Pr}_s(t' - t)(s')
    \leq \sup_{t \in \impT, t' \in \impT'} \Prob_{\ctmc,s}( \Finally^{[t, t']} s' )
    = \Prob_{\ctmc,s}( \Finally^{[\munderbar{t}, \bar{t}']} s' ),
\end{equation}
where $\Prob_{\ctmc,s}$ is the probability measure for the CTMC starting in initial state $s$, and $\bar{t}' - \munderbar{t}$ is the maximal time difference.
A lower bound is given symmetrically by the transient probability to reach $s'$ in the CTMC at the \emph{earliest} possible time $\munderbar{t'} - \bar{t}$ and staying there for the \emph{full} remaining time $(\bar{t'} - \munderbar{t}) - (\munderbar{t}' - \bar{t})$:
\begin{equation}
    \label{eq:transient:lower}
    \inf_{t \in \impT, t' \in \impT'} {\Pr}_s(t' - t)(s')
    \geq {\Pr}_s(\munderbar{t'} - \bar{t})(s') \cdot \Prob_{\ctmc,s'}(\Always^{[0, (\bar{t'} - \munderbar{t}) - (\munderbar{t}' - \bar{t})]} s').
\end{equation} 

\subsection*{Abstraction refinement}
\label{subsec:abstraction:refinement}

To improve the tightness of the bounds in \cref{thm:sandwich}, we propose a refinement step that splits elements of the time partition $\partition$.
For example, we may split the single abstract state in \cref{fig:abstraction_states} into the two states in \cref{fig:abstraction_states_refined}.
\begin{definition}[Refinement of time partition]
    \label{def:refinement}
    Let $\partition$ and $\partition'$ be partitions as per \cref{def:partition}, for which $|\partition'| > |\partition|$.
    We call $\partition'$ a refinement of $\partition$  if for all $\psi' \in \partition'$, there exists a $\psi \in \partition$ such that $\psi' \subseteq \psi$.
\end{definition}
Any refinement $\partition$' of partition $\partition$ can be constructed by finitely many splits.
We lift the refinement to the iMDP, see also \cref{fig:abstraction_intervals,fig:abstraction_intervals_refined}. 
The refined iMDP $\imdp' = \Abstract(\mdp_{|\imphist}, \partition')$ has more states and actions, but each union in \cref{eq:probability_intervals} is over a smaller set than in iMDP $\Abstract(\mdp_{|\imphist}, \partition)$. 
Thus, the refinement leads to smaller probability intervals and, in general, to tighter bounds in \cref{thm:sandwich}.
Repeatedly refining every element of the partition yields an iMDP with arbitrarily many states and actions and with arbitrarily small probability intervals.
Hence, in the limit, we may recover the original continuous MDP by refinements, which also implies that the bounds in \cref{thm:sandwich} on the refined iMDP converge.

\paragraph{Refinement strategy}
By splitting every element of the partition $\partition$, the number of iMDP states and actions double per iteration, and the number of transitions grows exponentially.
Thus, we employ the following \emph{guided refinement strategy}.
At each iteration, we extract the scheduler $\sched^\star$ that attains the upper bound in \cref{thm:sandwich} and determine the set $\tilde{Q}_\mathsf{reach}^{\sched^\star} \subset \tilde{Q}$ of reachable iMDP states.
We only refine the reachable elements $\psi \in \partition$, that is, for which there exists a $t \in \psi$ and $s \in S$ such that $\tuple{s,t} \in \tilde{Q}_\mathsf{reach}^{\sched^\star}$.
Using this guided strategy, we iteratively shrink only the relevant probability intervals, resulting in the same convergence behavior as the naive strategy but without the severe increase in abstraction size.
\section{Computing Bounds on the Conditional Reachability}
\label{sec:algorithm}
\cref{thm:sandwich} provides bounds on the conditional reachability $\Weight(\imphist)$ in \cref{problem}, but computing these bounds involves optimizing over the subset of consistent schedulers.
Recall from \cref{def:consistent_schedulers} that a consistent scheduler chooses the same actions in different states.\footnote{Consistent schedulers are similar to (memoryless) schedulers in partially observable MDPs that choose the same action in states with the same observation label.}
As we are not aware of any efficient algorithm to optimize over the consistent schedulers, we compute the following straightforward bounds:
\begin{lemma}[Bounds on \cref{problem}]
    \label{lemma:computable_bounds}
    Let $\imdp = \Abstract(\mdp_{|\imphist}, \partition)$ be the iMDP abstraction for the unfolded MDP $\mdp_{|\imphist}$ and a time partition $\partition$.
    It holds that
    \begin{equation}
        \Weight(\imphist)
        \leq
        \max_{\sched \in \Sched_{\imdp}^\mathsf{con}} \max_{\tilde{P} \in \calP} 
        \Weight_\imdp(\tilde{P},\sched)
        \leq
        \max_{\sched \in \Sched_{\imdp}} \max_{\tilde{P} \in \calP} 
        \Weight_\imdp(\tilde{P},\sched).
        \label{eq:upper_bound_all_schedulers}
    \end{equation}
    Moreover, any consistent scheduler $\hat\sched \in \Sched_\imdp^\mathsf{cons}$ results in a lower bound. 
\end{lemma}

\paragraph{Obtaining lower bounds} While we can use \emph{any} consistent scheduler in \cref{lemma:computable_bounds} to compute a lower bound on $\Weight(\imphist)$, we obtain better bounds by modifying a (potentially non-consistent) optimal scheduler $\sched^-$ under the worst-case choice of probabilities, i.e., $\sched^- = \argmax_{\sched \in \Sched_{\imdp}} \min_{\tilde{P} \in \calP} \Weight_\imdp(\tilde{P},\sched)$.
We check for inconsistency of scheduler $\sched^-$ by evaluating the following condition in all pairs of states $\tuple{s,t}, \tuple{s',t'} \in \tilde{Q}_\mathsf{reach}^{\sched^-} \subset \tilde{Q}$ reachable under $\sched^-$:
\begin{equation}
    t = t' \implies \sched(\tuple{s,t}) = \sched(\tuple{s',t})
    \quad
    \forall \tuple{s,t}, \tuple{s',t'} \in \tilde{Q}_\mathsf{reach}^{\sched^-}.
\end{equation}
We remove inconsistencies by changing the action in one of the states to match the others.
We take a greedy approach and always adapt to the action chosen most often across all iMDP states $\tuple{s,t} \in \tilde{Q}$ for the same time $t$.
For example, if $\sched(\tuple{s,t}) = \sched(\tuple{s',t}) \neq \sched(\tuple{s'',t})$, then we only modify $\sched(\tuple{s'',t})$ to match the other actions.
Because the set $\tilde{Q}_\mathsf{reach}^{\sched^-}$ is finite by construction, a finite number of modifications suffices to render any scheduler consistent.
The experiments in \cref{sec:experiments} show that modifying an inconsistent scheduler yields tighter lower bounds than taking the maximum over many sampled consistent schedulers.

\paragraph{Obtaining upper bounds}
The set of consistent schedulers is finite but prohibitively large, so enumerating over all consistent schedulers is infeasible.
For a sound upper bound, we instead optimize over all schedulers. 
The experiments in \cref{sec:experiments} show that we obtain (relatively) tight bounds. 
To further refine these upper bounds, the literature suggests another abstraction refinement loop, which can be formulated either directly on the imprecise evidence~\cite{DBLP:conf/tacas/CeskaJJK19} or on the consistent schedulers~\cite{DBLP:journals/tac/WintererJWJTKB21}. 
The latter approach leverages the fact that consistent schedulers can also be modeled as searching for (memoryless) schedulers in partially observable MDPs, where the schedulers would only observe the time but not the state. 
Finally, the hardness of optimizing over consistent schedulers in the iMDP remains open: Classical NP-hardness results for the problems above do not carry over.
\section{Numerical Experiments}
\label{sec:experiments}

We implemented our approach in a prototypical Python tool, which is available at \url{https://doi.org/10.5281/zenodo.10438984}.
The tool builds on top of \storm{}~\cite{DBLP:journals/sttt/HenselJKQV22} for the analysis of CTMCs and iMDPs.
It takes as input a CTMC $\ctmc$, a property defining the state-weight function~$\weight$, and imprecisely timed evidence $\imphist$.
The tool constructs the abstract iMDP for the coarsest time partition, computing the probability intervals as per \cref{eq:transient_upper,eq:transient:lower}.
The bounds on the conditional reachability in \cref{lemma:computable_bounds} are computed using robust value iteration.
Then, the tool applies guided refinements, as in \cref{subsec:abstraction:refinement}, and starts a new iteration with the refined partition.
After a predefined time limit, the tool returns the lower bound $\underline{\Weight(\imphist)}$ and upper bound $\overline{\Weight(\imphist)}$ on the conditional reachability $\Weight(\imphist)$:
\begin{equation}
        \underline{\Weight(\imphist)}
        =
        \min_{\tilde{P} \in \calP} \Weight_\imdp(\tilde{P},\hat\sched)
        \leq \Weight(\imphist) \leq
        \max_{\sched \in \Sched_{\imdp}} \max_{\tilde{P} \in \calP} \Weight_\imdp(\tilde{P},\sched)
        =
        \overline{\Weight(\imphist)},
        \label{eq:bounds_tool}
\end{equation}
where the consistent scheduler $\hat\sched$ for the lower bound is obtained by fixing all inconsistencies in the scheduler $\sched^-$ defined in \cref{sec:algorithm}.
The tool can also compute minimal conditional reachabilities (by swapping all $\min$ and $\max$ operators).

\begin{table*}[t]
\setlength{\tabcolsep}{3pt}

\centering
\caption{Overview of considered benchmarks.}
\label{tab:benchmarks_overview}

\scalebox{1}{
\begin{tabular}{lrrrl}
	\toprule
    \multicolumn{2}{c}{Example} & \multicolumn{2}{c}{{CTMC size}} & \multicolumn{1}{c}{State-weight function}\\
    \cmidrule(l){1-2} \cmidrule(l){3-4} \cmidrule(l){5-5}
    Name
    & Evid. len. ($|\imphist|$)
    & States
    & Transit.
    & Property
    \\
    \midrule
    \ex{Invent} & 3-14 & 3 & 4 & ``Prob. empty inventory within time 0.1''\\
    \ex{Ahrs}   & 4 & 74 & 196 & ``Prob. system failure within time 50''\\
    \ex{Phil}   & 4 & 34 & 89 & ``Prob. deadlock within time 1''\\
    \ex{Tandem} & 2 & 120 & 363 & ``Prob. both queues full within time 10''\\
    \ex{Polling} & 3 & 576 & 2208 & ``Prob. all stations empty within time 10''\\
    \bottomrule
\end{tabular}
}
\end{table*}%

\paragraph{Benchmarks}
We evaluate our approach on several CTMCs from the literature, creating multiple imprecisely timed evidence for each CTMC.
\cref{tab:benchmarks_overview} lists the evidence length (i.e., the number of observed times and labels), the number of CTMC states and transitions, and the property specifying the state-weight function.
More details on the benchmarks are in 
\ifappendix
    \cref{appendix:additional_results:benchmarks}.
\else
    \cite[Appendix~D.1]{Badings2024TACAS:extended},
\fi
All experiments run on an Intel Core i5 with 8GB RAM, using a time limit of 10 minutes.

\pgfmathsetlengthmacro\MajorTickLength{
      \pgfkeysvalueof{/pgfplots/major tick length} * 0.5
    }

\begin{figure}[t!]
\newcommand{\xlabelShift}{0.1cm}
\newcommand{\plotheight}{4cm}
\newcommand{\plotwidth}{0.325\linewidth}
\begin{subfigure}[b]{\plotwidth}
\begin{tikzpicture}

\definecolor{crimson2143940}{RGB}{214,39,40}
\definecolor{darkgray176}{RGB}{176,176,176}
\definecolor{gray127}{RGB}{127,127,127}
\definecolor{steelblue31119180}{RGB}{31,119,180}

\begin{axis}[
    width=\textwidth,
    height=\plotheight,
    xtick={0.1, 1, 10, 60, 600},
    xticklabels={0.1, 1, 10, 60, 600},
    ytick={0, 0.05, 0.1, 0.15},
    yticklabels={0, 0.05, 0.1, 0.15},
    tick label style={font=\scriptsize},
    label style={font=\scriptsize},
    major tick length=\MajorTickLength,
    xlabel style={yshift=\xlabelShift},
log basis x={10},
tick align=outside,
tick pos=left,
x grid style={darkgray176},
xlabel={Time [s]},
xmin=0.05, xmax=600,
xmode=log,
xtick style={color=black},
y grid style={darkgray176},
ylabel={Cond. reach. prob.},
ymin=0.0, ymax=0.15,
ytick style={color=black}
]
\path [draw=gray127, fill=gray127, opacity=0.5]
(axis cs:0.05,0.0764071214681387)
--(axis cs:0.05,0.0823696792960419)
--(axis cs:0.01,0.0823696792960419)
--(axis cs:0.03,0.0823696792960419)
--(axis cs:0.1,0.0823696792960419)
--(axis cs:0.31,0.0823696792960419)
--(axis cs:0.82,0.0823696792960419)
--(axis cs:2.12,0.0823696792960419)
--(axis cs:4.96,0.0823696792960419)
--(axis cs:9.21,0.0823696792960419)
--(axis cs:14.82,0.0823696792960419)
--(axis cs:22.58,0.0823696792960419)
--(axis cs:32.48,0.0823696792960419)
--(axis cs:46.51,0.0823696792960419)
--(axis cs:62.03,0.0823696792960419)
--(axis cs:80.3,0.0823696792960419)
--(axis cs:107.18,0.0823696792960419)
--(axis cs:134.04,0.0823696792960419)
--(axis cs:177.51,0.0823696792960419)
--(axis cs:210.98,0.0823696792960419)
--(axis cs:248.01,0.0823696792960419)
--(axis cs:289.57,0.0823696792960419)
--(axis cs:334.64,0.0823696792960419)
--(axis cs:384.65,0.0823696792960419)
--(axis cs:439.34,0.0823696792960419)
--(axis cs:499.57,0.0823696792960419)
--(axis cs:565.35,0.0823696792960419)
--(axis cs:637.81,0.0823696792960419)
--(axis cs:637.81,0.0764071214681387)
--(axis cs:637.81,0.0764071214681387)
--(axis cs:565.35,0.0764071214681387)
--(axis cs:499.57,0.0764071214681387)
--(axis cs:439.34,0.0764071214681387)
--(axis cs:384.65,0.0764071214681387)
--(axis cs:334.64,0.0764071214681387)
--(axis cs:289.57,0.0764071214681387)
--(axis cs:248.01,0.0764071214681387)
--(axis cs:210.98,0.0764071214681387)
--(axis cs:177.51,0.0764071214681387)
--(axis cs:134.04,0.0764071214681387)
--(axis cs:107.18,0.0764071214681387)
--(axis cs:80.3,0.0764071214681387)
--(axis cs:62.03,0.0764071214681387)
--(axis cs:46.51,0.0764071214681387)
--(axis cs:32.48,0.0764071214681387)
--(axis cs:22.58,0.0764071214681387)
--(axis cs:14.82,0.0764071214681387)
--(axis cs:9.21,0.0764071214681387)
--(axis cs:4.96,0.0764071214681387)
--(axis cs:2.12,0.0764071214681387)
--(axis cs:0.82,0.0764071214681387)
--(axis cs:0.31,0.0764071214681387)
--(axis cs:0.1,0.0764071214681387)
--(axis cs:0.03,0.0764071214681387)
--(axis cs:0.01,0.0764071214681387)
--(axis cs:0.05,0.0764071214681387)
--cycle;

\addplot [semithick, black]
table {%
0.05 0.0823696792960419
0.01 0.0823696792960419
0.03 0.0823696792960419
0.1 0.0823696792960419
0.31 0.0823696792960419
0.82 0.0823696792960419
2.12 0.0823696792960419
4.96 0.0823696792960419
9.21 0.0823696792960419
14.82 0.0823696792960419
22.58 0.0823696792960419
32.48 0.0823696792960419
46.51 0.0823696792960419
62.03 0.0823696792960419
80.3 0.0823696792960419
107.18 0.0823696792960419
134.04 0.0823696792960419
177.51 0.0823696792960419
210.98 0.0823696792960419
248.01 0.0823696792960419
289.57 0.0823696792960419
334.64 0.0823696792960419
384.65 0.0823696792960419
439.34 0.0823696792960419
499.57 0.0823696792960419
565.35 0.0823696792960419
637.81 0.0823696792960419
};
\addplot [semithick, crimson2143940]
table {%
0.01 0.025165
0.03 0.040696
0.1 0.055381
0.31 0.065055
0.82 0.065981
2.12 0.068411
4.96 0.068671
9.21 0.068949
14.82 0.069668
22.58 0.069793
32.48 0.070054
46.51 0.070192
62.03 0.070307
80.3 0.070334
107.18 0.070368
134.04 0.070482
177.51 0.070493
210.98 0.070557
248.01 0.070622
289.57 0.070629
334.64 0.070635
384.65 0.07066
439.34 0.070688
499.57 0.07072
565.35 0.070751
637.81 0.070755
};
\addplot [semithick, crimson2143940, dashed]
table {%
0.01 0.132122
0.03 0.112328
0.1 0.097331
0.31 0.087791
0.82 0.082348
2.12 0.07943
4.96 0.077918
9.21 0.077149
14.82 0.07676
22.58 0.076565
32.48 0.076468
46.51 0.076419
62.03 0.076394
80.3 0.076382
107.18 0.076376
134.04 0.076373
177.51 0.076371
210.98 0.07637
248.01 0.07637
289.57 0.07637
334.64 0.07637
384.65 0.07637
439.34 0.07637
499.57 0.07637
565.35 0.07637
637.81 0.07637
};
\addplot [semithick, steelblue31119180, dashed]
table {%
0.01 0.025165
0.03 0.042245
0.1 0.05885
0.31 0.069813
0.82 0.075965
2.12 0.079201
4.96 0.080856
9.21 0.081693
14.82 0.082114
22.58 0.082325
32.48 0.082431
46.51 0.082484
62.03 0.08251
80.3 0.082523
107.18 0.08253
134.04 0.082533
177.51 0.082535
210.98 0.082536
248.01 0.082536
289.57 0.082536
334.64 0.082536
384.65 0.082536
439.34 0.082536
499.57 0.082536
565.35 0.082536
637.81 0.082536
};
\addplot [semithick, steelblue31119180]
table {%
0.01 0.132122
0.03 0.115587
0.1 0.102293
0.31 0.093477
0.82 0.091266
2.12 0.089242
4.96 0.088743
9.21 0.088276
14.82 0.088195
22.58 0.087931
32.48 0.087661
46.51 0.087526
62.03 0.08743
80.3 0.087392
107.18 0.087324
134.04 0.087262
177.51 0.087257
210.98 0.087223
248.01 0.087193
289.57 0.087189
334.64 0.087188
384.65 0.087172
439.34 0.087155
499.57 0.087154
565.35 0.087147
637.81 0.087138
};
\draw (axis cs:0.1,0.085) node[
  scale=0.5,
  anchor=base west,
  text=black,
  rotate=0.0
]{$W(\Omega)'$};
\end{axis}%
\end{tikzpicture}%
    \caption{\ex{Invent} with evidence 1.}
    \label{subfig:invent_1}
\end{subfigure}
\hfill
\begin{subfigure}[b]{\plotwidth}
\begin{tikzpicture}

\definecolor{crimson2143940}{RGB}{214,39,40}
\definecolor{darkgray176}{RGB}{176,176,176}
\definecolor{gray127}{RGB}{127,127,127}
\definecolor{steelblue31119180}{RGB}{31,119,180}

\begin{axis}[
    width=\textwidth,
    height=\plotheight,
    xtick={0.1, 1, 10, 60, 600},
    xticklabels={0.1, 1, 10, 60, 600},
    ytick={0.9, 0.95, 1},
    yticklabels={0.9, 0.95, 1},
    tick label style={font=\scriptsize},
    label style={font=\scriptsize},
    major tick length=\MajorTickLength,
    xlabel style={yshift=\xlabelShift},
log basis x={10},
tick align=outside,
tick pos=left,
x grid style={darkgray176},
xlabel={Time [s]},
xmin=0.8, xmax=600,
xmode=log,
xtick style={color=black},
y grid style={darkgray176},
ylabel=\empty, 
ymin=0.90, ymax=1.0,
ytick style={color=black}
]
\path [draw=gray127, fill=gray127, opacity=0.5]
(axis cs:0.05,0.957807961117679)
--(axis cs:0.05,0.962605785891122)
--(axis cs:0.46,0.962605785891122)
--(axis cs:1.9,0.962605785891122)
--(axis cs:6.6,0.962605785891122)
--(axis cs:23.56,0.962605785891122)
--(axis cs:77.48,0.962605785891122)
--(axis cs:249.42,0.962605785891122)
--(axis cs:800.39,0.962605785891122)
--(axis cs:800.39,0.957807961117679)
--(axis cs:800.39,0.957807961117679)
--(axis cs:249.42,0.957807961117679)
--(axis cs:77.48,0.957807961117679)
--(axis cs:23.56,0.957807961117679)
--(axis cs:6.6,0.957807961117679)
--(axis cs:1.9,0.957807961117679)
--(axis cs:0.46,0.957807961117679)
--(axis cs:0.05,0.957807961117679)
--cycle;

\addplot [semithick, black]
table {%
0.05 0.962605785891122
0.46 0.962605785891122
1.9 0.962605785891122
6.6 0.962605785891122
23.56 0.962605785891122
77.48 0.962605785891122
249.42 0.962605785891122
800.39 0.962605785891122
};
\addplot [semithick, crimson2143940]
table {%
0.46 0.904601
1.9 0.928636
6.6 0.942189
23.56 0.948768
77.48 0.951069
249.42 0.952714
800.39 0.954194
};
\addplot [semithick, crimson2143940, dashed]
table {%
0.46 0.98367
1.9 0.973195
6.6 0.966025
23.56 0.961857
77.48 0.959615
249.42 0.958452
800.39 0.95786
};
\addplot [semithick, steelblue31119180, dashed]
table {%
0.46 0.904602
1.9 0.932688
6.6 0.94712
23.56 0.956096
77.48 0.959298
249.42 0.961071
800.39 0.962041
};
\addplot [semithick, steelblue31119180]
table {%
0.46 0.983669
1.9 0.97504
6.6 0.969646
23.56 0.966662
77.48 0.965627
249.42 0.965129
800.39 0.964306
};
\draw (axis cs:1,0.965) node[
  scale=0.5,
  anchor=base west,
  text=black,
  rotate=0.0
]{$W(\Omega)'$};
\end{axis}%
\end{tikzpicture}%
    \caption{\ex{Ahrs} with evidence 1.}
    \label{subfig:ahrs_1}
\end{subfigure}
\hfill
\begin{subfigure}[b]{\plotwidth}
\begin{tikzpicture}

\definecolor{crimson2143940}{RGB}{214,39,40}
\definecolor{darkgray176}{RGB}{176,176,176}
\definecolor{gray127}{RGB}{127,127,127}
\definecolor{steelblue31119180}{RGB}{31,119,180}

\begin{axis}[
    width=\textwidth,
    height=\plotheight,
    xtick={0.1, 1, 10, 60, 600},
    xticklabels={0.1, 1, 10, 60, 600},
    ytick={0.06, 0.07, 0.08},
    yticklabels={0.06, 0.07, 0.08},
    scaled ticks=false,
    tick label style={font=\scriptsize},
    label style={font=\scriptsize},
    major tick length=\MajorTickLength,
    xlabel style={yshift=\xlabelShift},
log basis x={10},
tick align=outside,
tick pos=left,
x grid style={darkgray176},
xlabel={Time [s]},
xmin=0.8, xmax=600,
xmode=log,
xtick style={color=black},
y grid style={darkgray176},
ylabel=\empty, 
ymin=0.06, ymax=0.08,
ytick style={color=black}
]
\path [draw=gray127, fill=gray127, opacity=0.5]
(axis cs:0.05,0.0709074678300218)
--(axis cs:0.05,0.0712881603797118)
--(axis cs:0.31,0.0712881603797118)
--(axis cs:1.2,0.0712881603797118)
--(axis cs:3.93,0.0712881603797118)
--(axis cs:12.92,0.0712881603797118)
--(axis cs:33.32,0.0712881603797118)
--(axis cs:86.63,0.0712881603797118)
--(axis cs:192.94,0.0712881603797118)
--(axis cs:367.61,0.0712881603797118)
--(axis cs:646.85,0.0712881603797118)
--(axis cs:646.85,0.0709074678300218)
--(axis cs:646.85,0.0709074678300218)
--(axis cs:367.61,0.0709074678300218)
--(axis cs:192.94,0.0709074678300218)
--(axis cs:86.63,0.0709074678300218)
--(axis cs:33.32,0.0709074678300218)
--(axis cs:12.92,0.0709074678300218)
--(axis cs:3.93,0.0709074678300218)
--(axis cs:1.2,0.0709074678300218)
--(axis cs:0.31,0.0709074678300218)
--(axis cs:0.05,0.0709074678300218)
--cycle;

\addplot [semithick, black]
table {%
0.05 0.0712881603797118
0.31 0.0712881603797118
1.2 0.0712881603797118
3.93 0.0712881603797118
12.92 0.0712881603797118
33.32 0.0712881603797118
86.63 0.0712881603797118
192.94 0.0712881603797118
367.61 0.0712881603797118
646.85 0.0712881603797118
};
\addplot [semithick, crimson2143940]
table {%
0.31 0.063327
1.2 0.067147
3.93 0.068848
12.92 0.069642
33.32 0.069735
86.63 0.069871
192.94 0.070054
367.61 0.070065
646.85 0.070075
};
\addplot [semithick, crimson2143940, dashed]
table {%
0.31 0.076902
1.2 0.073683
3.93 0.072323
12.92 0.071667
33.32 0.071319
86.63 0.071139
192.94 0.071048
367.61 0.071002
646.85 0.070979
};
\addplot [semithick, steelblue31119180, dashed]
table {%
0.31 0.063327
1.2 0.067637
3.93 0.069517
12.92 0.070412
33.32 0.070844
86.63 0.071056
192.94 0.071161
367.61 0.071213
646.85 0.071239
};
\addplot [semithick, steelblue31119180]
table {%
0.31 0.076902
1.2 0.074281
3.93 0.073015
12.92 0.072452
33.32 0.072343
86.63 0.072226
192.94 0.072078
367.61 0.072059
646.85 0.072057
};
\draw (axis cs:1,0.0718) node[
  scale=0.5,
  anchor=base west,
  text=black,
  rotate=0.0
]{$W(\Omega)'$};
\end{axis}%
\end{tikzpicture}%
    \caption{\ex{Ahrs} with evidence 2.}
    \label{subfig:ahrs_2}
\end{subfigure}
\begin{subfigure}[b]{\plotwidth}
\begin{tikzpicture}

\definecolor{crimson2143940}{RGB}{214,39,40}
\definecolor{darkgray176}{RGB}{176,176,176}
\definecolor{gray127}{RGB}{127,127,127}
\definecolor{steelblue31119180}{RGB}{31,119,180}

\begin{axis}[
    width=\textwidth,
    height=\plotheight,
    xtick={0.1, 1, 10, 60, 600},
    xticklabels={0.1, 1, 10, 60, 600},
    ytick={0.4, 0.6, 0.8, 1},
    yticklabels={\phantom{0}0.4, 0.6, 0.8, 1},
    tick label style={font=\scriptsize},
    label style={font=\scriptsize},
    major tick length=\MajorTickLength,
    xlabel style={yshift=\xlabelShift},
log basis x={10},
tick align=outside,
tick pos=left,
x grid style={darkgray176},
xlabel={Time [s]},
xmin=0.8, xmax=600,
xmode=log,
xtick style={color=black},
y grid style={darkgray176},
ylabel={Cond. reach. prob.},
ymin=0.4, ymax=1.0,
ytick style={color=black}
]
\path [draw=gray127, fill=gray127, opacity=0.5]
(axis cs:0.05,0.824094391236942)
--(axis cs:0.05,0.839301713822009)
--(axis cs:0.16,0.839301713822009)
--(axis cs:0.55,0.839301713822009)
--(axis cs:2.3,0.839301713822009)
--(axis cs:7.97,0.839301713822009)
--(axis cs:28.86,0.839301713822009)
--(axis cs:91.66,0.839301713822009)
--(axis cs:295.82,0.839301713822009)
--(axis cs:912.23,0.839301713822009)
--(axis cs:912.23,0.824094391236942)
--(axis cs:912.23,0.824094391236942)
--(axis cs:295.82,0.824094391236942)
--(axis cs:91.66,0.824094391236942)
--(axis cs:28.86,0.824094391236942)
--(axis cs:7.97,0.824094391236942)
--(axis cs:2.3,0.824094391236942)
--(axis cs:0.55,0.824094391236942)
--(axis cs:0.16,0.824094391236942)
--(axis cs:0.05,0.824094391236942)
--cycle;

\addplot [semithick, black]
table {%
0.05 0.839301713822009
0.16 0.839301713822009
0.55 0.839301713822009
2.3 0.839301713822009
7.97 0.839301713822009
28.86 0.839301713822009
91.66 0.839301713822009
295.82 0.839301713822009
912.23 0.839301713822009
};
\addplot [semithick, crimson2143940]
table {%
0.16 0.510816
0.55 0.6763
2.3 0.753933
7.97 0.78932
28.86 0.799641
91.66 0.804624
295.82 0.806785
912.23 0.807689
};
\addplot [semithick, crimson2143940, dashed]
table {%
0.16 0.952249
0.55 0.906499
2.3 0.870919
7.97 0.848703
28.86 0.836396
91.66 0.829932
295.82 0.826621
912.23 0.824945
};
\addplot [semithick, steelblue31119180, dashed]
table {%
0.16 0.510816
0.55 0.691344
2.3 0.771221
7.97 0.806593
28.86 0.823193
91.66 0.831052
295.82 0.834785
912.23 0.836695
};
\addplot [semithick, steelblue31119180]
table {%
0.16 0.952248
0.55 0.913046
2.3 0.882564
7.97 0.863673
28.86 0.857128
91.66 0.853789
295.82 0.852212
912.23 0.851548
};
\draw (axis cs:1,0.85) node[
  scale=0.5,
  anchor=base west,
  text=black,
  rotate=0.0
]{$W(\Omega)'$};
\end{axis}%
\end{tikzpicture}%
    \caption{\ex{Phil} with evidence 1.}
    \label{subfig:phil_1}
\end{subfigure}
\hfill
\begin{subfigure}[b]{\plotwidth}
\begin{tikzpicture}

\definecolor{crimson2143940}{RGB}{214,39,40}
\definecolor{darkgray176}{RGB}{176,176,176}
\definecolor{gray127}{RGB}{127,127,127}
\definecolor{steelblue31119180}{RGB}{31,119,180}

\begin{axis}[
    width=\textwidth,
    height=\plotheight,
    xtick={0.1, 1, 10, 60, 600},
    xticklabels={0.1, 1, 10, 60, 600},
    ytick={0, 0.01, 0.02},
    yticklabels={0, 0.01, 0.02},
    scaled ticks=false,
    tick label style={font=\scriptsize},
    label style={font=\scriptsize},
    major tick length=\MajorTickLength,
    xlabel style={yshift=\xlabelShift},
log basis x={10},
tick align=outside,
tick pos=left,
x grid style={darkgray176},
xlabel={Time [s]},
xmin=0.8, xmax=600,
xmode=log,
xtick style={color=black},
y grid style={darkgray176},
ylabel=\empty, 
ymin=0.0, ymax=0.02,
ytick style={color=black}
]
\path [draw=gray127, fill=gray127, opacity=0.5]
(axis cs:0.05,0.00358341191726812)
--(axis cs:0.05,0.00359636177823899)
--(axis cs:0.86,0.00359636177823899)
--(axis cs:2.61,0.00359636177823899)
--(axis cs:7.04,0.00359636177823899)
--(axis cs:17.79,0.00359636177823899)
--(axis cs:46.14,0.00359636177823899)
--(axis cs:98.73,0.00359636177823899)
--(axis cs:195.88,0.00359636177823899)
--(axis cs:354.2,0.00359636177823899)
--(axis cs:598.37,0.00359636177823899)
--(axis cs:920.42,0.00359636177823899)
--(axis cs:920.42,0.00358341191726812)
--(axis cs:920.42,0.00358341191726812)
--(axis cs:598.37,0.00358341191726812)
--(axis cs:354.2,0.00358341191726812)
--(axis cs:195.88,0.00358341191726812)
--(axis cs:98.73,0.00358341191726812)
--(axis cs:46.14,0.00358341191726812)
--(axis cs:17.79,0.00358341191726812)
--(axis cs:7.04,0.00358341191726812)
--(axis cs:2.61,0.00358341191726812)
--(axis cs:0.86,0.00358341191726812)
--(axis cs:0.05,0.00358341191726812)
--cycle;

\addplot [semithick, black]
table {%
0.05 0.00359636177823899
0.86 0.00359636177823899
2.61 0.00359636177823899
7.04 0.00359636177823899
17.79 0.00359636177823899
46.14 0.00359636177823899
98.73 0.00359636177823899
195.88 0.00359636177823899
354.2 0.00359636177823899
598.37 0.00359636177823899
920.42 0.00359636177823899
};
\addplot [semithick, crimson2143940]
table {%
0.86 0.002904
2.61 0.002934
7.04 0.003052
17.79 0.003117
46.14 0.003216
98.73 0.003216
195.88 0.00328
354.2 0.003281
598.37 0.003283
920.42 0.003318
};
\addplot [semithick, crimson2143940, dashed]
table {%
0.86 0.016526
2.61 0.007741
7.04 0.00516
17.79 0.004296
46.14 0.003916
98.73 0.003749
195.88 0.003671
354.2 0.003633
598.37 0.003614
920.42 0.003605
};
\addplot [semithick, steelblue31119180, dashed]
table {%
0.86 0.002904
2.61 0.002936
7.04 0.003062
17.79 0.003225
46.14 0.003357
98.73 0.003453
195.88 0.003515
354.2 0.003549
598.37 0.003567
920.42 0.003577
};
\addplot [semithick, steelblue31119180]
table {%
0.86 0.016526
2.61 0.008203
7.04 0.005352
17.79 0.004923
46.14 0.004349
98.73 0.004349
195.88 0.004117
354.2 0.004114
598.37 0.004111
920.42 0.004009
};
\draw (axis cs:1,0.004) node[
  scale=0.5,
  anchor=base west,
  text=black,
  rotate=0.0
]{$W(\Omega)'$};
\end{axis}%
\end{tikzpicture}%
    \caption{\ex{Tandem} with evidence 1.}
    \label{subfig:tandem_1}
\end{subfigure}
\hfill
\begin{subfigure}[b]{\plotwidth}
\begin{tikzpicture}

\definecolor{crimson2143940}{RGB}{214,39,40}
\definecolor{darkgray176}{RGB}{176,176,176}
\definecolor{gray127}{RGB}{127,127,127}
\definecolor{steelblue31119180}{RGB}{31,119,180}

\begin{axis}[
    width=\textwidth,
    height=\plotheight,
    xtick={0.1, 1, 10, 60, 600},
    xticklabels={0.1, 1, 10, 60, 600},
    ytick={0.4, 0.6, 0.8, 1},
    yticklabels={\phantom{0}0.4, 0.6, 0.8, 1},
    tick label style={font=\scriptsize},
    label style={font=\scriptsize},
    major tick length=\MajorTickLength,
    xlabel style={yshift=\xlabelShift},
log basis x={10},
tick align=outside,
tick pos=left,
x grid style={darkgray176},
xlabel={Time [s]},
xmin=0.8, xmax=600,
xmode=log,
xtick style={color=black},
y grid style={darkgray176},
ylabel=\empty, 
ymin=0.4, ymax=1.0,
ytick style={color=black}
]
\path [draw=gray127, fill=gray127, opacity=0.5]
(axis cs:0.05,0.744682293404531)
--(axis cs:0.05,0.74839342375536)
--(axis cs:93.14,0.74839342375536)
--(axis cs:367.09,0.74839342375536)
--(axis cs:2951.89,0.74839342375536)
--(axis cs:2951.89,0.744682293404531)
--(axis cs:2951.89,0.744682293404531)
--(axis cs:367.09,0.744682293404531)
--(axis cs:93.14,0.744682293404531)
--(axis cs:0.05,0.744682293404531)
--cycle;

\addplot [semithick, black]
table {%
0.05 0.74839342375536
93.14 0.74839342375536
367.09 0.74839342375536
2951.89 0.74839342375536
};
\addplot [semithick, crimson2143940]
table {%
93.14 0.687551
367.09 0.713631
2951.89 0.726752
};
\addplot [semithick, crimson2143940, dashed]
table {%
93.14 0.791948
367.09 0.778071
2951.89 0.777599
};
\addplot [semithick, steelblue31119180, dashed]
table {%
93.14 0.687551
367.09 0.717634
2951.89 0.73141
};
\addplot [semithick, steelblue31119180]
table {%
93.14 0.791948
367.09 0.779804
2951.89 0.781912
};
\draw (axis cs:1,0.75) node[
  scale=0.5,
  anchor=base west,
  text=black,
  rotate=0.0
]{$W(\Omega)'$};
\end{axis}%
\end{tikzpicture}%
    \caption{\ex{Polling} with evidence 1.}
    \label{subfig:polling_1}
\end{subfigure}
\caption{Results for different CTMCs and different imprecisely timed evidence. The blue lines are the upper bound $\overline{\Weight(\imphist)}$ (solid) and lower bound $\underline{\Weight(\imphist)}$ (dashed) on $\Weight(\imphist)$; red lines show the analogous lower bounds.}
\label{fig:results}
\end{figure}
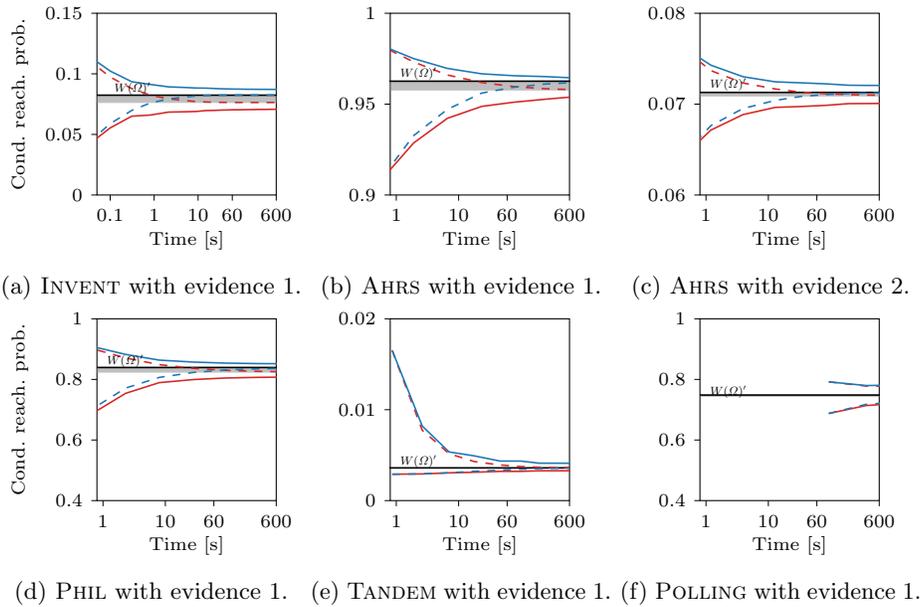

\subsubsection{Feasibility of our approach}
We investigate if our approach yields tight bounds on the weighted reachability.
\cref{fig:results} shows the results for each example with different imprecise evidences.
The gray area shows the weighted reachabilities (as per \cref{thm2:problem_solution}) for 500 precisely timed instances $\rho \in \imphist$ sampled from the imprecise evidence. 
Recall that the weighted reachability $\Weight(\imphist)$ is an upper bound to the weighted reachability for each precisely timed evidence $\rho \in \imphist$.
Thus, the upper bound of the gray areas in \cref{fig:results}, indicated as ${\Weight(\imphist)}'$, is a lower bound of the actual (but unknown) value $\Weight(\imphist)$.
The blue lines are the upper bound $\overline{\Weight(\imphist)}$ (solid) and lower bound $\underline{\Weight(\imphist)}$ (dashed) on $\Weight(\imphist)$ returned by our approach over the runtime (note the log-scale).
Similarly, the red lines are the bounds obtained for \emph{minimizing} the minimal weighted reachability.

\paragraph{Tightness of bounds}
\Cref{fig:results} shows that we obtain reasonably tight bounds within a minute.
In all examples, the lower bound converges close to the maximum of the samples.
The improvement is steepest at the start, indicating that the bounds can be quickly improved by only a few refinement steps.
In the long run, the improvement of the bounds diminishes, both because each refinement takes longer, and the improvement in each iteration gets smaller.

While not clearly visible in \cref{subfig:invent_1}, the lower bound $\underline{\Weight(\imphist)}$ (dashed blue line) slightly exceeds the maximal sampled value ${\Weight(\imphist)}'$ (gray area) in the end.
Thus, the lower bound $\underline{\Weight(\imphist)}$ is closer to the actual weighted reachability $\Weight(\imphist)$ than the maximal lower bound obtained by sampling.
We observed the same results when increasing the number of samples used to compute ${\Weight(\imphist)}'$ to $10\,000$.

\Cref{subfig:ahrs_1,subfig:ahrs_2} show the general benefit of conditioning on evidence.
While evidence 1 for \ex{AHRS} results in a state in which a system failure within the next 50 time units is very likely, a failure conditioned on evidence 2 is very unlikely.

\begin{table*}[t!]
\setlength{\tabcolsep}{3pt}

\centering
\caption{Results for all benchmarks (evidence length $|\imphist|$ is given after the name).}

\scalebox{0.82}{
\begin{tabular}{lrrlrrrrrr}
    \toprule
    Example & \multicolumn{2}{c}{Refine} & \multicolumn{1}{c}{{Results}} & \multicolumn{3}{c}{{iMDP size}} & \multicolumn{3}{c}{{Timings [s]}}\\
    \cmidrule(l){1-1} \cmidrule(l){2-3} \cmidrule(l){4-4} \cmidrule(l){5-7} \cmidrule(l){8-10}
    Name ($|\imphist|$)
    & Iter.
    & \#split
    & Bounds on $\Weight(\imphist)$
    & States
    & Actions
    & Transit.
    & Unfold
    & Analysis
    & Total
    \\
    \midrule
    \ex{Invent-1} (4)  & 25 & 555 & [0.082536, 0.087138] & 898 & 128307 & 278163 & 537.51 & 100.28 & 637.81\\
    \ex{Invent-2} (4)  & 27 & 585 & [0.071768, 0.078328] & 1180 & 167917 & 503537 & 606.91 & 43.85 & 650.74\\
    \ex{Invent-3} (9)  & 14 & 1176& [0.071757, 0.078577] & 2372 & 369329 & 1107877 & 658.77 & 127.83 & 786.57\\
    \ex{Invent-4} (15) & 7  & 528 & [0.070924, 0.080409] & 1016 & 39927 & 115119 & 42.63 & 974.89 & 1017.50\\
    \ex{Ahrs-1} (4)    & 6  & 177 & [0.962041, 0.964306] & 6283 & 282538 & 1415346 & 620.75 & 179.65 & 800.39\\
    \ex{Ahrs-2} (4)    & 8  & 154 & [0.071239, 0.072057] & 727 & 20626 & 81362 & 577.64 & 69.19 & 646.85\\
    \ex{Ahrs-3} (4)    & 6  & 176 & [0.964936, 0.969535] & 6112 & 280954 & 1334231 & 749.38 & 152.61 & 902.00\\
    \ex{Ahrs-4} (4)    & 7  & 300 & [0.209591, 0.213820] & 7179 & 535763 & 3618439 & 1801.81 & 111.39 & 1913.18\\
    \ex{Phil-1} (5)    & 7  & 339 & [0.836695, 0.851548] & 4122 & 370091 & 3887339 & 851.92 & 60.32 & 912.23\\
    \ex{Phil-2} (5)    & 6  & 209 & [0.236734, 0.246067] & 4050 & 203549 & 3669721 & 419.97 & 376.73 & 796.70\\
    \ex{Tandem-1} (2)  & 9  & 77  & [0.003577, 0.004009] & 1203 & 24561 & 362657 & 917.29 & 3.11 & 920.42\\
    \ex{Tandem-2} (2)  & 7  & 80  & [0.130187, 0.162762] & 587 & 25096 & 75548 & 549.03 & 327.93 & 876.96\\
    \ex{Polling-1} (3) & 2  & 9   & [0.731410, 0.781912] & 3267 & 9798 & 2379462 & 348.83 & 2603.08 & 2951.89\\
    \bottomrule
\end{tabular}
}
\label{tab:benchmarks_summary}
\end{table*}%

\subsubsection{Scalability}
We investigate the scalability of our approach.
\Cref{tab:benchmarks_summary} provides the refinement statistics, bounds, model sizes, and runtimes for all benchmarks.
The refinement statistics show the number of iterations (Iter.) and the total number of splits made in the partition.
The bounds on $W(\imphist)$ (which are the solid and dashed blue lines in \cref{fig:results}) and the iMDP sizes are both given for the final iteration.
For the timings, we provide the total time (over all iterations) and distinguish between the time spent on unfolding the model, i.e., constructing the iMDP, and analyzing it.
Our approach terminates if after an iteration, the total run time so far exceeds the time limit of 10 minutes.
The total runtime can, therefore, be significantly longer than 10 minutes.

\paragraph{CTMC size}
The size of the CTMC has a large impact on the total runtime.
For example, for evidence with 4 labels, we can perform up to 27 iterations for \ex{Invent} (3 CTMC states) but only 6-8 for \ex{Ahrs} (74 CTMC states).
For \ex{Polling} (576 states) with evidence of length 2, performing 2 iterations takes nearly 50 minutes.
The CTMC size affects the unfolding, which requires computing the transient probabilities from all states in one layer to all states in the next one.
A clear example is \ex{Tandem-1} (120 CTMC states), where nearly all of the runtime is spent on the unfolding.
A larger CTMC also leads to more transitions in the iMDP and thus, can increase the analysis time.
An example is \ex{Polling-1} (576 CTMC states), where most of the runtime is spent in the analysis.

\paragraph{Length of evidence}
The time per refinement step increases with the length of the evidence.
For example, for \ex{Invent-4} (with 15 labels), only 7 iterations are performed because the resulting iMDP has 15 layers, so the value iteration becomes the bottleneck (nearly 96\% of the runtime for this example is spent on analyzing the iMDP). This is consistent with experiments on unfolded MDPs in~\cite{DBLP:conf/cav/JungesTS20,DBLP:conf/tacas/HartmannsJQW23}, where policy iteration-based methods lead to better results.

\paragraph{Caching improves performance}
To reduce runtimes, we implemented caching in our tool, which allows reusing transient probability computations.
For example, if all labels in the evidence have a time interval of the same width (which is the case for \ex{Ahrs-1}), transient probabilities are the same between layers of the unfolding.
\cref{tab:benchmarks_overview} shows that the unfolding times for \ex{Ahrs-1} are indeed lower than for, e.g., \ex{Ahrs-3}, which has time intervals of different widths.

\paragraph{Likelihood of evidence}
The size of the iMDP is influenced by the number of CTMC states corresponding to the observed labels.
Less likely observations can, therefore, mean that fewer CTMC states need to be considered in each layer.
For example, the evidence in \ex{Ahrs-2} is 17 times less likely (probability of 0.01, with 569 states) than \ex{Ahrs-4} (probability of 0.17, with 4007 states), and as a result the total runtime of \ex{Ahrs-2} is less than for \ex{Ahrs-4}.
\section{Related work}
\label{sec:related}

Beyond the related work discussed in \cref{sec:introduction} on DTAs~\cite{DBLP:journals/corr/abs-1101-3694,DBLP:journals/pe/AmparoreD18,DBLP:conf/cav/FengKLXZ18} and synthesis of timeouts \cite{DBLP:conf/qest/BrazdilKKNR15,DBLP:conf/mascots/KorenciakKR16,DBLP:journals/tomacs/BaierDKKR19}, the following work is related to ours.

Imprecisely timed evidence can also be expressed via multiphase timed until formulas in continuous-time linear logic~\cite{DBLP:conf/tacas/GuanY22}.
However, similar to DTA, conditioning and computing the maximal weighted reachability are not supported.

Conditional probabilities naturally appear in runtime monitoring~\cite{DBLP:journals/fmsd/SanchezSABBCFFK19,DBLP:series/lncs/BartocciDDFMNS18} and speech recognition~\cite{DBLP:journals/ftsig/GalesY07}, and is, e.g., studied for hidden Markov models~\cite{DBLP:conf/rv/StollerBSGHSZ11} and MDPs~\cite{DBLP:conf/tacas/BaierKKM14,DBLP:conf/cav/JungesTS20}.
Approximate model checking of conditional continuous stochastic logic for CTMCs is studied in~\cite{DBLP:journals/ipl/GaoXZZ13,DBLP:conf/atva/GaoHZ013} by means of a product construction formalized as CTMC, but their algorithm is incompatible with imprecise observation times.
Conditional sampling in CTMCs is studied by~\cite{hobolth2009simulation}, and maximum likelihood inference of paths in CTMCs by~\cite{DBLP:conf/nips/Perkins09}.

The abstraction of continuous stochastic models into iMDPs is well-studied~\cite{LSAZ21}.
Various papers develop abstractions of stochastic hybrid and dynamical systems into iMDPs~\cite{Badings2022JAIR,Badings2023AAAI,DBLP:conf/tacas/CauchiA19} and relate to early work in~\cite{DBLP:conf/lics/JonssonL91}.
Our abstraction in \cref{sec:abstraction} is similar to a game-based abstraction, in which the (possibly infinite-state) model is abstracted into a two-player stochastic game~\cite{DBLP:journals/fmsd/KattenbeltKNP10,DBLP:conf/tacas/HahnHWZ10,DBLP:conf/qest/HahnNPWZ11}. In particular, iMDPs are a special case of a stochastic game in which the actions of the second player in each state only differ in transition probabilities~\cite{DBLP:journals/ior/NilimG05,DBLP:journals/mor/Iyengar05}.
An interesting extension of our approach is to consider CTMCs with uncertain \emph{transition rates}, which have recently also been studied extensively, e.g., in \cite{DBLP:conf/rtss/HanKM08,DBLP:journals/acta/CeskaDPKB17,DBLP:journals/jss/CalinescuCGKP18,DBLP:conf/qest/CardelliGLTTV21,Cardelli2023TAC,DBLP:conf/cav/BadingsJJSV22}.
\section{Conclusion}
\label{sec:conclusion}

We have presented the first method for computing reachability probabilities in CTMCs that are conditioned on evidence with imprecise observation times.
The method combines an unfolding of the problem into an infinite MDP with an iterative abstraction into a finite iMDP.
Our experiments have shown the applicability of our method across several benchmarks.

A natural next step is to embed our method in a predictive runtime monitoring framework, which introduces the challenge of running our algorithm in realtime.
Another interesting extension is to consider uncertainty in the observed labels.
Furthermore, this paper gives rise to four concrete challenges.
First, finding better methods to overapproximate the union over MDP probabilities in \cref{eq:probability_intervals} may lead to tighter bounds on the weighted reachability.
Second, we want to optimize over the consistent schedulers only, potentially via techniques used in~\cite{DBLP:conf/cav/AndriushchenkoC21}.
Third, we wish to explore better refinement strategies for the iMDP.
The final challenge is to improve the computational performance of our implementation.
One promising option to improve performance is to adapt symbolic policy iteration~\cite{DBLP:journals/tomacs/BaierDKKR19}, which only considers small sets of candidate actions instead of all actions.

\clearpage
\bibliographystyle{splncs04}
\bibliography{literature}

\ifappendix
    \clearpage
    \appendix
    
    \section{Proof of \cref{thm1:risk_on_MDP}}
\label{appendix:proof:trans}

The proof of \cref{thm1:risk_on_MDP} is based on \cref{lemma:trans} below, which states that, for every $\rho \in \imphist$ with consistent scheduler $\sched \consistent \rho$, it holds that
\begin{equation}
    \label{eq:thm1:risk_on_MDP:proof1}
    \Prob_\ctmc(\pi(t_\hor) = s \mid [\pi \models \rho])
    = \Prob_\mdp^\sched( \Finally \tuple{ s, t_\star } \mid [\xi \models \rho]).
\end{equation}
That is, the conditional transient probability $\Prob_\ctmc(\pi(t_\hor) = s \mid [\pi \models \rho])$ equals the conditional reachability probabilities in \cref{eq:thm1:risk_on_MDP:proof1} for the unfolded MDP $\mdp$, under a scheduler $\sched \sim \rho$ consistent to $\rho$.
We then use \cref{eq:thm1:risk_on_MDP:proof1} to rewrite  \cref{problem} as
\begin{equation}
\begin{split}
    \label{eq:thm1:risk_on_MDP:proof2}
    \Weight(\imphist) = \sup_{\rho \in \imphist} \, \sum_{s \in S} \Prob_\mdp^\sched( \Finally \tuple{ s, t_\star } \mid [\xi \models \rho]) \cdot \weight(s),
\end{split}
\end{equation}
where $\sched \consistent \rho$, as per \cref{def:consistent_schedulers}.
Due to the one-to-one correspondence between choices $\rho \in \imphist$ and consistent schedulers, we can replace the supremum over $\rho \in \imphist$ by the supremum over consistent schedulers, which yields the expression in \cref{eq:thm1:risk_on_MDP}.

Next, we formalize the lemma that shows \cref{eq:thm1:risk_on_MDP:proof1}.
In the proof of this lemma, we use the notion of the \emph{state-trace} $\sTr_\rho(\pi) \in S^\hor$ of a CTMC path $\pi$ onto the time points $t_1,\ldots,t_\hor$ of the precisely timed evidence $\rho$, which is defined as follows:
\begin{equation}
    \label{eq:state_trace}
    \sTr_\rho(\pi) = \big( \pi(t_1), \pi(t_2), \ldots, \pi(t_\hor) \big).
\end{equation}
Conditional reachability probabilities in the CTMC and in the unfolded MDP are then related as follows.
\begin{lemma}
    \label{lemma:trans}
    For a CTMC $\ctmc$ and the imprecise evidence $\imphist$, let $\mdp = \Unfold(\ctmc, \timegraph_\imphist)$ be the unfolded MDP.
    For every instance $\rho \in \imphist$ with corresponding consistent scheduler $\sched \in \Sched_\mdp^\mathsf{cons}$, i.e., such that $\sched \consistent \rho$, it holds that
    \begin{equation}
        \label{eq:lemma:trans}
        \Prob_\ctmc(\pi(t_\hor) = s \mid [\pi \models \rho])
        = \Prob_\mdp^\sched( \Finally \tuple{ s, t_\star } \mid [\xi \models \rho]).
    \end{equation}
\end{lemma}

\begin{proof}
First, let us use Bayes' rule to rewrite the right-hand side of \cref{eq:lemma:trans} as
\begin{equation}
    \label{eq:lemma:trans:proof1}
    \Prob_\mdp^\sched( \Finally \tuple{ s, t_\star } \mid [\xi \models \rho]) = 
    \frac{\Prob_\mdp^\sched( \Finally \tuple{ s, t_\star } \cap [\xi \models \rho])}
    {\Prob_\mdp^\sched( \xi \models \rho)}.
\end{equation}
We will prove \cref{lemma:trans} by showing that the numerator and denominator in \cref{eq:lemma:trans:proof1} are equivalent to those in \cref{eq:BayesRule}.
In other words, we will show that
\begin{align}
    \label{eq:lemma2a}
    \Prob_\ctmc \big( [\pi(t_\hor) = s] \cap [\pi \models \rho] \big) &= 
    \Prob_\mdp^\sched \big( \Finally \tuple{ s, t_\star } \cap [\pi \models \rho] \big) \enskip \forall s \in S
    \\
    \label{eq:lemma2b}
    \Prob_\ctmc \big( \pi \models \rho \big) &= 
    \Prob_\mdp^\sched \big( \xi \models \rho \big),
\end{align}
where $\sched \consistent \rho$ are consistent as per \cref{def:consistent_schedulers}.
We prove \cref{eq:lemma2b} first and then prove \cref{eq:lemma2a} in a largely analogous manner.
    
\textbf{Proof of \cref{eq:lemma2b}}.
From \cref{eq:BayesRule}, we have for every $\rho \in \imphist$ that
\begin{equation}
    \label{eq:Lemma2:proof1}
    \Prob_\ctmc(\pi \models \rho)
    =
    \int_\Pi \mathbbm{1}_{(\pi \models \rho)} \Pr(\pi) d\pi
    =
    \Prob_\ctmc(\pi \in \Pi_\rho),
\end{equation}
where $\Pi_\rho = \{ \pi \in \Pi : \pi \models \rho \} \subset \Pi$ is the subset of CTMC paths consistent with evidence $\rho$.
Let $\Gamma_\rho$ be the set of state-traces that are consistent with evidence $\rho$:
\begin{equation}
    \label{eq:Gamma}
    \Gamma_\rho = 
    \bigcup_{\pi \in \Pi} \{ \sTr_\rho(\pi) : \pi \models \rho \} \subseteq S^{\hor}.
\end{equation}
Let us denote $(x_1,\ldots,x_k)$ by $x_{1:k}$ for brevity.
Using this notation, the preimages $\sTr^{-1}(s_{1:\hor})$ for all $s_{1:\hor} \in \Gamma_\rho$ form a partition of $\Pi_\rho$, that is:
\begin{equation}
\label{eq:Lemma2:proof2}
    \bar{\Pi} = \bigcup_{s_{1:\hor} \in \Gamma_\rho}{ \sTr^{-1}_\rho(s_{1:\hor}) }
    \,\, \text{and} \,\,
    \sTr^{-1}_\rho(s_{1:\hor}) \cup \sTr^{-1}_\rho(s_{1:\hor}') = \varnothing \, \enskip \forall s_{1:\hor}, s_{1:\hor}' \in \Gamma_\rho.
\end{equation}
Thus, we can rewrite \cref{eq:Lemma2:proof1} as a finite sum over all state-traces $s_{1:\hor} \in \Gamma_\rho$:
\begin{equation}
\label{eq:Lemma2:proof3}
    \Prob_\ctmc(\pi \models \rho)
    =
    \sum_{s_{1:\hor} \in \Gamma_\rho} \Prob_\ctmc \big( 
    \pi \in \Pi 
    : \sTr_\rho(\pi) = s_{1:\hor} 
    \big).
\end{equation}
The term $\Prob_\ctmc(\pi \in \Pi : \sTr_\rho(\pi) = s_{1:\hor})$ is the probability for a path $\pi$ whose state-trace is $s_{1:\hor}$.
This probability is equal to the product of the appropriate transient probabilities $\pr_{s_{i-1}}(t_i - t_{i-1})(s_i)$ for all $s \in \{1,\ldots,\hor\}$, as defined in \cref{sec:preliminaries}:
\begin{equation}
\begin{split}
    \Prob_\ctmc(\pi \models \rho) &= \sum_{s_{1:\hor} \in \Gamma_\rho} 
    \prod_{i=1}^{\hor} \pr_{s_{i-1}} (t_i - t_{i-1}) (s_i)
    \label{eq:Lemma2:proof4},
\end{split}
\end{equation}
where $s_0 = s_I$ and $t_0 = 0$.
Recall from \cref{def:unfolded_mdp} that the unfolded MDP has transition probabilities $P( \tuple{ s,t }, t', \tuple{ s',t' }) = \pr_s(t' - t)(s')$.
Hence, we obtain
\begin{equation}
\begin{split}
    \label{eq:Lemma2:proof5}
    \Prob_\ctmc(\pi \models \rho)
    &= \sum_{s_{1:\hor} \in \Gamma_\rho} \prod_{i=1}^\hor P \big( \tuple{ s_{i-1},t_{i-1} }, t_i, \tuple{ s_i,t_i } \big)
    \\
    &= \sum_{s_{1:\hor} \in \Gamma_\rho} \Prob_\mdp^\sched \big( \xi \in \Xi_\mdp : \xi = \tuple{ s_I, 0 }, \tuple{ s_1,t_1 },\ldots, \tuple{ s_\hor, t_\hor }
    \big).
    \nonumber
\end{split}
\end{equation}
A state-trace $s_{1:\hor}$ belongs to $\Gamma_\rho$ if and only if the associated MDP path $\xi = \tuple{ s_I, 0 }, \tuple{ s_1,t_1 },\ldots, \tuple{ s_\hor, t_\hor }
\in \Xi_\mdp$ is consistent with $\rho$, i.e., $\xi \models \rho$.
Thus, we can rewrite \cref{eq:Lemma2:proof5} as the desired expression:
\begin{equation}
    \Prob_\ctmc(\pi \models \rho) 
    = \sum_{\xi \in \Xi_\mdp} \Prob_\mdp^\sched(\xi) \cdot \mathbbm{1}_{(\xi \models \rho)}
    = \Prob_\mdp^\sched(\xi \models \rho).
    \label{eq:Lemma2:proof6}
\end{equation}

\textbf{Proof of \cref{eq:lemma2a}.}
Again, using the fact that the preimages $\sTr^{-1}(s_{1:\hor})$ for all $s_{1:\hor} \in \Gamma_\rho$ form a partition of $\Pi_\rho$ (where $\sTr$ is defined by \cref{eq:state_trace}), we obtain
\begin{equation}
    \label{eq:Lemma2:proof7}
    \Prob_\ctmc([\pi(t_\hor) = s] \cap [\pi \models \rho]) = 
    \sum_{s_{1:\hor} \in \Gamma_\rho} \Prob_\ctmc \big( 
    \pi \in \Pi 
    : [\pi(t_\hor) = s] \cap [\sTr_\rho(\pi) = s_{1:\hor}]
    \big).
\end{equation}
Compared to \cref{eq:lemma2a}, we additionally require that $\pi(t_\hor) = s$, which corresponds with reaching the terminal state $\tuple{ s,t_\star } \in Q$ in the unfolded MDP $\mdp$ corresponding with CTMC state $s \in S$.
As a result, we have that
\begin{align}
    \Prob_\ctmc([\pi(t_\hor) = s] \cap [\pi \models \rho]) 
    &= \sum_{s_{1:\hor} \in \Gamma_\rho} \prod_{i=1}^\hor P \big( \tuple{ s_{i-1},t_{i-1} }, t_i, \tuple{ s_i,t_i }\big) \cdot \mathbbm{1}_{(s_\hor = s)}
    \nonumber
    \\
    &= \sum_{s_{1:\hor} \in \Gamma_\rho} \Prob_\mdp^\sched \big( \tuple{ s_I, 0 },\tuple{ s_1,t_1 },\ldots,\tuple{ s_\hor, t_\hor }
    \big) \cdot \mathbbm{1}_{(s_\hor = s)}
    \nonumber
    \\
    &= \sum_{\xi \in \Xi_\mdp} \Prob_\mdp^\sched(\xi) \cdot \mathbbm{1}_{(\xi \models \rho)} \cdot \mathbbm{1}_{(\xi \models \Finally \tuple{ s,t_\star })}
    \nonumber
    \\
    &= \Prob_\mdp^\sched \big( \Finally \tuple{ s,t_\star } \cap [\pi \models \rho] \big).
    \label{eq:Lemma2:proof8}
\end{align}
Observe \cref{eq:Lemma2:proof8} is the desired expression in \cref{eq:lemma2a}, so we conclude the proof.
\end{proof}


\section{Proof of \cref{thm:sandwich}}
\label{appendix:proof:sandwich}

Let $H \colon Q \to \tilde{Q}$ be a function that maps every state of MDP $\mdp_{|\imphist}$ to a state of iMDP $\imdp$, such that $H(\tuple{s,t}) = \tuple{s,\impT} \in \tilde{Q}$, where $t \in \impT$.
The mapping $H$ is well-defined as $\tilde{Q}$ represents a proper partition of $Q$.
We prove \cref{thm:sandwich} by showing that for every MDP state $\tuple{s,t} \in Q$, the corresponding iMDP state $H(\tuple{s,t}) = \tuple{s,\impT} \in \tilde{Q}$ overapproximates its behavior.
Formally, for the conditioned MDP, take any transition from state $\tuple{s,t} \in Q$ via (enabled) action $t' \in A(\tuple{s,t})$ to state $\tuple{s',t'} \in Q$.
For any such transition, there exists an iMDP transition $\tuple{s,\impT} \in \tilde{Q}$ via $\impT' \in A(\tuple{s,\impT})$ to state $\tuple{s',\impT'} \in \tilde{Q}$ such that
\begin{enumerate}
    \item there exists $\tilde{P} \in \calP$ such that $P\big( \tuple{s,t}, t', \tuple{s',t'} \big) = \tilde{P}\big( \tuple{s,\impT}, \impT', \tuple{s',\impT'} \big)$, and
    \item it holds that $H(\tuple{s,t}) = \tuple{s,\impT} \in \tilde{Q}$ and $H(\tuple{s',t'}) = \tuple{s',\impT'} \in \tilde{Q}$.
\end{enumerate}
Observe that the converse also holds: for any iMDP transition, there exists a corresponding MDP transition such that the conditions above hold.
These conditions formalize that there always exists a transition function $\tilde{P} \in \calP$ such that the induced MDP $\imdp[\tilde{P}]$ is a \emph{probabilistic bisimulation} of the conditioned MDP $\mdp_{|\imphist}$, similar as in~\cite{DBLP:conf/lics/JonssonL91}.
Hence, there exists a $\tilde{P} \in \calP$ such that 
\begin{equation}
    \label{eq:thm:sandwich:proof1}
    \max_{\sched \in \Sched_\imdp^\mathsf{con}} \Weight_\imdp(\tilde{P}, \sched) = \Weight(\imphist).
\end{equation}
The upper and lower bounds in \cref{eq:thm:sandwich} follow directly from \cref{eq:thm:sandwich:proof1}, so we conclude the proof.

    \section{Computing the Probability for Given Evidence}
\label{appendix:path_probabilities}

We discuss the variation from \cref{remark:variation} of computing the probability for observing the given precise evidence $\rho$ in more detail.
Specifically, we show that, with minor modifications to our unfolding procedure, we can compute the probability that a CTMC generates the given (precise) evidence $\rho$.
Instead of looping all states $\tuple{ s,t } \in Q_\mathsf{reset}$ inconsistent with the evidence (defined in \cref{eq:inconsistent_states}) back to the initial state, we now create self-loops for those states.
Formally, given an unfolded MDP $\mdp = \Unfold(\ctmc, \timegraph_\rho) = \tuple{ Q, q_I, A, P }$ for precise evidence $\rho$, we define the modified MDP $\mdp_\rho = \tuple{ Q, q_I, A, P_\rho }$ with transition function $P_\rho$ defined for all $\tuple{ s, t }, \tuple{ s', t' } \in Q$ as
\begin{align*}
    \label{eq:path_MDP}
    P_\rho \big( \tuple{ s, t }, t', \tuple{ s', t' } \big) = \begin{cases}
        P \big( \tuple{ s, t }, t', \tuple{ s', t' } \big) & \text{if } \tuple{ s,t } \notin Q_\mathsf{reset}(\imphist),
        \\
        \mathbbm{1}_{(\tuple{ s, t } = \tuple{ s', t' })} & \text{if } \tuple{ s,t } \in Q_\mathsf{reset}(\imphist), 
    \end{cases}
\end{align*}
with $Q_\mathsf{reset}$ defined by \cref{eq:inconsistent_states}.
This transformation of the unfolded MDP is shown in \cref{fig:path_probability} for two different precisely timed evidences.
Then, the probability $\Prob_\ctmc(\pi \models \rho)$ that CTMC $\ctmc$ generates the evidence $\rho$ is the probability that $\mdp_\rho$ reaches a state $\tuple{s,t_\star}$ for time $t_\star$ and any CTMC state $s \in S$:
\begin{equation}
    \label{eq:probability_of_history}
    \Prob_\ctmc(\pi \models \rho) = \sum_{s \in S} \Prob_{\mdp_\rho}( \Finally\tuple{ s, t_\star } ). 
\end{equation}
Intuitively, \cref{eq:probability_intervals} computes the probability of ever reaching a terminal state at time $t_\star$.
Because all paths inconsistent with the evidence $\rho$ are trapped by the self-loops (in non-terminal states), \cref{eq:probability_intervals} thus computes the probability that the CTMC generates a path that is consistent with $\rho$.
For imprecise evidence $\imphist$, we can also ask for the \emph{worst-case} probability to obtain any instance $\rho \in \imphist$, by modifying the unfolded MDP $\mdp = \Unfold(\ctmc, \timegraph_\imphist)$ in an analogous manner.

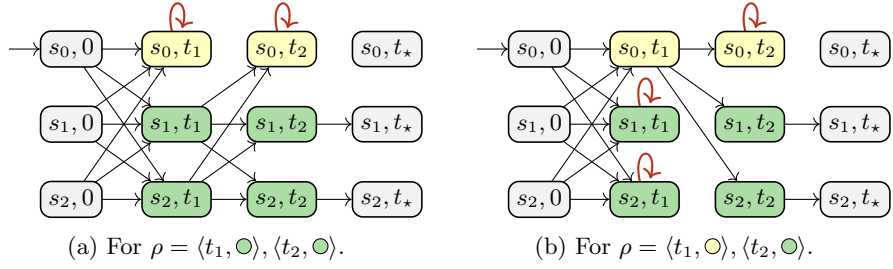
\begin{figure}[t]
\begin{subfigure}[b]{0.49\linewidth}
    \centering
    \iftikzcompile
    \begin{tikzpicture}[node distance=1.0cm]

    \def\xshift{0.4cm}
      \node[state, initial,initial where=left, initial text=] (s0t0) {$s_0, 0$};
      \node[state, below of=s0t0] (s1t0) {$s_1, 0$};
      \node[state, below of=s1t0] (s2t0) {$s_2, 0$};
      \node[state, empty label, right of=s0t0, xshift=\xshift] (s0t1) {$s_0, t_1$};
      \node[state, nonempty label, below of=s0t1] (s1t1) {$s_1, t_1$};
      \node[state, nonempty label, below of=s1t1] (s2t1) {$s_2, t_1$};
      \node[state, empty label, right of=s0t1, xshift=\xshift] (s0t2) {$s_0, t_2$};
      \node[state, nonempty label, below of=s0t2] (s1t2) {$s_1, t_2$};
      \node[state, nonempty label, below of=s1t2] (s2t2) {$s_2, t_2$};
      \node[state, right of=s0t2, xshift=\xshift] (s0star) {$s_0, t_\star$};
      \node[state, below of=s0star] (s1star) {$s_1, t_\star$};
      \node[state, below of=s1star] (s2star) {$s_2, t_\star$};
      \begin{pgfonlayer}{bg}    
      \path[->,shorten <= -2pt+\pgflinewidth] 
                (s0t0) edge [] node [above] {} (s0t1)
                (s0t0) edge [] node [above] {} (s1t1)
                (s0t0) edge [] node [above] {} (s2t1)
                (s1t0) edge [] node [above] {} (s0t1)
                (s1t0) edge [] node [above] {} (s1t1)
                (s1t0) edge [] node [above] {} (s2t1)
                (s2t0) edge [] node [above] {} (s0t1)
                (s2t0) edge [] node [above] {} (s1t1)
                (s2t0) edge [] node [above] {} (s2t1)
                %
                (s0t1) edge [loop above, red, thick] node [] {} (s0t1)
                (s1t1) edge [] node [above] {} (s0t2)
                (s1t1) edge [] node [above] {} (s1t2)
                (s1t1) edge [] node [above] {} (s2t2)
                (s2t1) edge [] node [above] {} (s0t2)
                (s2t1) edge [] node [above] {} (s1t2)
                (s2t1) edge [] node [above] {} (s2t2)
                %
                (s0t2) edge [loop above, red, thick] node [] {} (s0t2)
                (s1t2) edge [] node [above] {} (s1star)
                (s2t2) edge [] node [above] {} (s2star);
    \end{pgfonlayer}
    \end{tikzpicture}
    \fi
    \caption{For $\rho = \langle t_1, \bigcdot{nonempty_color} \rangle, \langle t_2, \bigcdot{nonempty_color} \rangle$.}
    \label{fig:path_probability1}
    \end{subfigure}
\hfill 
\begin{subfigure}[b]{0.49\linewidth}
    \centering
    \iftikzcompile
    \begin{tikzpicture}[node distance=1.0cm]

    \def\xshift{0.4cm}
      \node[state, initial,initial where=left, initial text=] (s0t0) {$s_0, 0$};
      \node[state, below of=s0t0] (s1t0) {$s_1, 0$};
      \node[state, below of=s1t0] (s2t0) {$s_2, 0$};
      \node[state, empty label, right of=s0t0, xshift=\xshift] (s0t1) {$s_0, t_1$};
      \node[state, nonempty label, below of=s0t1] (s1t1) {$s_1, t_1$};
      \node[state, nonempty label, below of=s1t1] (s2t1) {$s_2, t_1$};
      \node[state, empty label, right of=s0t1, xshift=\xshift] (s0t2) {$s_0, t_2$};
      \node[state, nonempty label, below of=s0t2] (s1t2) {$s_1, t_2$};
      \node[state, nonempty label, below of=s1t2] (s2t2) {$s_2, t_2$};
      \node[state, right of=s0t2, xshift=\xshift] (s0star) {$s_0, t_\star$};
      \node[state, below of=s0star] (s1star) {$s_1, t_\star$};
      \node[state, below of=s1star] (s2star) {$s_2, t_\star$};
      \begin{pgfonlayer}{bg}    
      \path[->,shorten <= -2pt+\pgflinewidth] 
                (s0t0) edge [] node [above] {} (s0t1)
                (s0t0) edge [] node [above] {} (s1t1)
                (s0t0) edge [] node [above] {} (s2t1)
                (s1t0) edge [] node [above] {} (s0t1)
                (s1t0) edge [] node [above] {} (s1t1)
                (s1t0) edge [] node [above] {} (s2t1)
                (s2t0) edge [] node [above] {} (s0t1)
                (s2t0) edge [] node [above] {} (s1t1)
                (s2t0) edge [] node [above] {} (s2t1)
                (s0t1) edge [] node [above] {} (s0t2)
                (s0t1) edge [] node [above] {} (s1t2)
                (s0t1) edge [] node [above] {} (s2t2)
                (s1t1) edge [loop above, red, thick] node [] {} (s1t1)
                (s2t1) edge [loop above, red, thick] node [] {} (s2t1)
                %
                (s0t2) edge [loop above, red, thick] node [] {} (s0t2)
                (s1t2) edge [] node [above] {} (s1star)
                (s2t2) edge [] node [above] {} (s2star);
    \end{pgfonlayer}
    \end{tikzpicture}
    \fi
    \caption{For $\rho = \langle t_1, \bigcdot{empty_color} \rangle, \langle t_2, \bigcdot{nonempty_color} \rangle$.}
    \label{fig:path_probability2}
\end{subfigure}
    \caption{Using the unfolded MDP to compute the probability for two precisely timed evidences. States that do not agree with the evidence are made absorbing.}
    \label{fig:path_probability}
\end{figure}
    \section{Details on Numerical Experiments}
\label{appendix:additional_results}

In this appendix, we provide additional details on the benchmarks used in \cref{sec:experiments}, and we provide more detailed results.

\subsection{Benchmarks}
\label{appendix:additional_results:benchmarks}

We describe each of the benchmarks used in \cref{sec:experiments} in more detail in the following.
\Cref{tab:benchmarks_evidence} provides the evidence for each example.
\begin{table*}[tp]
\setlength{\tabcolsep}{3pt}

\centering
\caption{Evidences for each benchmark.}
\label{tab:benchmarks_evidence}

\scalebox{1}{
\begin{tabular}{lp{0.85\textwidth}}
	\toprule
    Example & Evidence\\
    \midrule
    \ex{Invent-1} & 
        $\tuple{[0,0],\neg \texttt{empty}}$,
        $\tuple{[0.9,1.1],\neg \texttt{empty}}$,
        $\tuple{[1.9,2.1],\texttt{empty}}$,\\&
        $\tuple{[2.9,3.1],\neg \texttt{empty}}$\\
    \ex{Invent-2} & 
        $\tuple{[0,0],\neg \texttt{empty}}$,
        $\tuple{[0.9,1.1],\neg \texttt{empty}}$,
        $\tuple{[1.9,2.1],\neg \texttt{empty}}$,\\&
        $\tuple{[2.9,3.1],\neg \texttt{empty}}$\\
    \ex{Invent-3} & 
        $\tuple{[0,0],\neg \texttt{empty}}$,
        $\tuple{[0.9,1.1],\neg \texttt{empty}}$,
        $\tuple{[1.9,2.1],\neg \texttt{empty}}$,\\&
        $\tuple{[2.9,3.1],\neg \texttt{empty}}$,
        $\tuple{[3.9,4.1],\neg \texttt{empty}}$,
        $\tuple{[4.9,5.1],\neg \texttt{empty}}$,\\&
        $\tuple{[5.9,6.1],\neg \texttt{empty}}$,
        $\tuple{[6.9,7.1],\neg \texttt{empty}}$,
        $\tuple{[7.9,8.1],\neg \texttt{empty}}$\\
    \ex{Invent-4} & 
        $\tuple{[0,0],\neg \texttt{empty}}$,
        $\tuple{[0.9,1.1],\neg \texttt{empty}}$,
        $\tuple{[1.9,2.1],\neg \texttt{empty}}$,\\&
        $\tuple{[2.9,3.1],\neg \texttt{empty}}$,
        $\tuple{[3.9,4.1],\neg \texttt{empty}}$,
        $\tuple{[4.9,5.1],\texttt{empty}}$,\\&
        $\tuple{[5.9,6.1],\neg \texttt{empty}}$,
        $\tuple{[6.9,7.1],\neg \texttt{empty}}$,
        $\tuple{[7.9,8.1],\neg \texttt{empty}}$,\\&
        $\tuple{[8.9,9.1],\neg \texttt{empty}}$,
        $\tuple{[9.9,10.1],\texttt{empty}}$,
        $\tuple{[10.9,11.1],\texttt{empty}}$,\\&
        $\tuple{[11.9,12.1],\neg \texttt{empty}}$,
        $\tuple{[12.9,13.1],\neg \texttt{empty}}$,
        $\tuple{[13.9,14.1],\neg \texttt{empty}}$\\
    \ex{Ahrs-1} & 
        $\tuple{[50,60],\texttt{init}}$,\\
        &$\tuple{[100,110],\texttt{A1}_f}$,\\
        &$\tuple{[200,210],\texttt{A1}_f \land \texttt{Spare}_f}$,\\
        &$\tuple{[300,310],\texttt{A1}_f \land \texttt{Spare}_f \land \texttt{B}_f}$\\
    \ex{Ahrs-2} & 
        $\tuple{[50,60],\texttt{init}}$,\\
        &$\tuple{[100,110],\neg \texttt{A1}_f \land \neg \texttt{Spare}_f \land \neg \texttt{B}_f}$,\\
        &$\tuple{[200,210],\texttt{A1}_f \land \neg \texttt{Spare}_f \land \neg \texttt{B}_f}$,\\
        &$\tuple{[300,310],\texttt{A1}_f \land \texttt{A2}_f \land \neg \texttt{Spare}_f \land \neg \texttt{B}_f}$\\
    \ex{Ahrs-3} & 
    $\tuple{[30,60],\texttt{init}}$,\\
        &$\tuple{[90,110],\texttt{A1}_f}$,\\
        &$\tuple{[190,205],\texttt{A1}_f \land \texttt{Spare}_f}$,\\
        &$\tuple{[300,310],\texttt{A1}_f \land \texttt{Spare}_f \land \texttt{B}_f}$\\
    \ex{Ahrs-4} & 
        $\tuple{[50,60],\neg \texttt{Spare}_f}$,\\
        &$\tuple{[100,110],\neg \texttt{Spare}_f}$,\\
        &$\tuple{[200,210],\neg \texttt{Spare}_f}$,\\
        &$\tuple{[300,310],\texttt{A1}_f \land \neg \texttt{Spare}_f}$\\
    \ex{Phil-1} & 
        $\tuple{[0.9,1.1],\neg \texttt{fork}_1}$,\\
        &$\tuple{[1.9,2.1],\neg \texttt{fork}_1 \land \neg \texttt{fork}_2}$,\\
        &$\tuple{[2.9,3.1],\neg \texttt{fork}_1 \land \neg \texttt{fork}_2 \land \neg \texttt{fork}_3}$,\\
        &$\tuple{[3.9,4.1],\neg \texttt{fork}_1 \land \neg \texttt{fork}_2 \land \neg \texttt{fork}_3 \land \neg \texttt{fork}_4}$\\
    \ex{Phil-2} & 
        $\tuple{[0.9,1.1],\neg \texttt{fork}_1}$,\\
        &$\tuple{[1.9,2.1],\neg \texttt{fork}_1}$,\\
        &$\tuple{[2.9,3.1],\neg \texttt{fork}_1}$,\\
        &$\tuple{[3.9,4.1],\texttt{fork}_1}$\\
    \ex{Tandem-1} & 
        $\tuple{[2.5,3.0],\neg \texttt{first\_full}}$,\\
        &$\tuple{[5.0,5.5],\texttt{first\_full} \land \neg \texttt{second\_full}}$\\
    \ex{Tandem-2} & 
        $\tuple{[2.5,3.0],\texttt{first\_full}}$,\\
        &$\tuple{[5.0,5.5],\texttt{first\_full} \land \texttt{second\_full}}$\\
    \ex{Polling-1} & 
        $\tuple{[0.9,1.1],\texttt{waiting1}}$,\\
        &$\tuple{[1.9,2.1],\texttt{waiting2}}$,\\
        &$\tuple{[2.9,3.1],\texttt{waiting3}}$\\
    \bottomrule
\end{tabular}
}
\end{table*}%

\textsc{Invent} is the inventory model from \cref{fig:ctmc} with the label \texttt{empty} if the inventory is empty and $\neg\texttt{empty}$ otherwise.
The state-weight function is defined by the probability of reaching an empty inventory within a time bound of $0.1$.

\textsc{Ahrs} is a dynamic fault tree model of an Active Heat Rejection System~\cite{BoudaliD05}.
The model was taken from the \ex{Ffort} fault tree collection~\cite{Ffort19} and converted into a CTMC using \storm~\cite{DBLP:journals/tii/VolkJK18}.
The evidence is given by observations of the failures of sub-systems and components, for instance $\texttt{A1}_f$ and $\texttt{Spare}_f$.
The state-weight function is given by the probability of system failure within the next $50$ time units.

\textsc{Phil} models a variant of the dining philosophers~\cite{DBLP:journals/acta/Dijkstra71} and was taken from the \ex{QComp} benchmark collection~\cite{DBLP:conf/tacas/HartmannsKPQR19}.
As evidence, we can observe for each fork whether it is currently in use ($\neg \texttt{fork}_i$) or available.
The state-weight function is given by the probability of reaching a deadlock within $1$ time unit.

\textsc{Tandem} models a tandem queuing network consisting of a Coxian distribution with two phases sequentially composed with a M/M/1-queue~\cite{hermanns1999multi}.
As evidence, we observe whether any of the two queues is full.
The state-weight function is given by the probability that both queues will be full within 10 time units.

\textsc{Polling} models a cyclic server polling system~\cite{DBLP:conf/infocom/ChoiT92}.
Six stations are are handled by one polling server, which processes the jobs of the stations with a given rate.
As evidence, we observe whether stations are empty, i.e., have no jobs.
The state-weight function is given by the probability that all stations have no jobs within 10 time units.

\subsection{Additional results}
\begin{figure}[tp]
{\newcommand{\xlabelShift}{0.1cm}
\newcommand{\plotheight}{4cm}
\newcommand{\plotwidth}{0.325\linewidth}
\begin{subfigure}[b]{\plotwidth}
\begin{tikzpicture}

\definecolor{crimson2143940}{RGB}{214,39,40}
\definecolor{darkgray176}{RGB}{176,176,176}
\definecolor{gray127}{RGB}{127,127,127}
\definecolor{steelblue31119180}{RGB}{31,119,180}

\begin{axis}[
    width=\textwidth,
    height=\plotheight,
    xtick={0.1, 1, 10, 60, 600},
    xticklabels={0.1, 1, 10, 60, 600},
    ytick={0, 0.05, 0.1, 0.15},
    yticklabels={0, 0.05, 0.1, 0.15},
    scaled ticks=false,
    tick label style={font=\scriptsize},
    label style={font=\scriptsize},
    major tick length=\MajorTickLength,
    xlabel style={yshift=\xlabelShift},
log basis x={10},
tick align=outside,
tick pos=left,
x grid style={darkgray176},
xlabel={Time [s]},
xmin=0.05, xmax=600,
xmode=log,
xtick style={color=black},
y grid style={darkgray176},
ylabel={Cond. reach. prob.},
ymin=0.00, ymax=0.15,
ytick style={color=black}
]
\path [draw=gray127, fill=gray127, opacity=0.5]
(axis cs:0.05,0.0709661857196395)
--(axis cs:0.05,0.0722810926218195)
--(axis cs:0.02,0.0722810926218195)
--(axis cs:0.04,0.0722810926218195)
--(axis cs:0.12,0.0722810926218195)
--(axis cs:0.31,0.0722810926218195)
--(axis cs:0.71,0.0722810926218195)
--(axis cs:1.87,0.0722810926218195)
--(axis cs:4.08,0.0722810926218195)
--(axis cs:7.46,0.0722810926218195)
--(axis cs:12.15,0.0722810926218195)
--(axis cs:18.75,0.0722810926218195)
--(axis cs:26.96,0.0722810926218195)
--(axis cs:37.15,0.0722810926218195)
--(axis cs:49.71,0.0722810926218195)
--(axis cs:64.48,0.0722810926218195)
--(axis cs:81.44,0.0722810926218195)
--(axis cs:101.86,0.0722810926218195)
--(axis cs:125.82,0.0722810926218195)
--(axis cs:152.13,0.0722810926218195)
--(axis cs:181.7,0.0722810926218195)
--(axis cs:214.69,0.0722810926218195)
--(axis cs:253.8,0.0722810926218195)
--(axis cs:296.86,0.0722810926218195)
--(axis cs:345.21,0.0722810926218195)
--(axis cs:398.2,0.0722810926218195)
--(axis cs:454.79,0.0722810926218195)
--(axis cs:516.87,0.0722810926218195)
--(axis cs:581.41,0.0722810926218195)
--(axis cs:650.74,0.0722810926218195)
--(axis cs:650.74,0.0709661857196395)
--(axis cs:650.74,0.0709661857196395)
--(axis cs:581.41,0.0709661857196395)
--(axis cs:516.87,0.0709661857196395)
--(axis cs:454.79,0.0709661857196395)
--(axis cs:398.2,0.0709661857196395)
--(axis cs:345.21,0.0709661857196395)
--(axis cs:296.86,0.0709661857196395)
--(axis cs:253.8,0.0709661857196395)
--(axis cs:214.69,0.0709661857196395)
--(axis cs:181.7,0.0709661857196395)
--(axis cs:152.13,0.0709661857196395)
--(axis cs:125.82,0.0709661857196395)
--(axis cs:101.86,0.0709661857196395)
--(axis cs:81.44,0.0709661857196395)
--(axis cs:64.48,0.0709661857196395)
--(axis cs:49.71,0.0709661857196395)
--(axis cs:37.15,0.0709661857196395)
--(axis cs:26.96,0.0709661857196395)
--(axis cs:18.75,0.0709661857196395)
--(axis cs:12.15,0.0709661857196395)
--(axis cs:7.46,0.0709661857196395)
--(axis cs:4.08,0.0709661857196395)
--(axis cs:1.87,0.0709661857196395)
--(axis cs:0.71,0.0709661857196395)
--(axis cs:0.31,0.0709661857196395)
--(axis cs:0.12,0.0709661857196395)
--(axis cs:0.04,0.0709661857196395)
--(axis cs:0.02,0.0709661857196395)
--(axis cs:0.05,0.0709661857196395)
--cycle;

\addplot [semithick, black]
table {%
0.05 0.0722810926218195
0.02 0.0722810926218195
0.04 0.0722810926218195
0.12 0.0722810926218195
0.31 0.0722810926218195
0.71 0.0722810926218195
1.87 0.0722810926218195
4.08 0.0722810926218195
7.46 0.0722810926218195
12.15 0.0722810926218195
18.75 0.0722810926218195
26.96 0.0722810926218195
37.15 0.0722810926218195
49.71 0.0722810926218195
64.48 0.0722810926218195
81.44 0.0722810926218195
101.86 0.0722810926218195
125.82 0.0722810926218195
152.13 0.0722810926218195
181.7 0.0722810926218195
214.69 0.0722810926218195
253.8 0.0722810926218195
296.86 0.0722810926218195
345.21 0.0722810926218195
398.2 0.0722810926218195
454.79 0.0722810926218195
516.87 0.0722810926218195
581.41 0.0722810926218195
650.74 0.0722810926218195
};
\addplot [semithick, crimson2143940]
table {%
0.02 0.023507
0.04 0.037037
0.12 0.050444
0.31 0.059599
0.71 0.060193
1.87 0.062749
4.08 0.062863
7.46 0.062971
12.15 0.064078
18.75 0.064137
26.96 0.064249
37.15 0.064303
49.71 0.064354
64.48 0.064405
81.44 0.064756
101.86 0.064786
125.82 0.064844
152.13 0.064872
181.7 0.0649
214.69 0.064928
253.8 0.064955
296.86 0.064981
345.21 0.065007
398.2 0.065059
454.79 0.065084
516.87 0.065095
581.41 0.065109
650.74 0.06511
};
\addplot [semithick, crimson2143940, dashed]
table {%
0.02 0.126406
0.04 0.107001
0.12 0.092237
0.31 0.082885
0.71 0.077574
1.87 0.074736
4.08 0.073268
7.46 0.072522
12.15 0.072145
18.75 0.071956
26.96 0.071861
37.15 0.071814
49.71 0.07179
64.48 0.071778
81.44 0.071773
101.86 0.07177
125.82 0.071768
152.13 0.071767
181.7 0.071767
214.69 0.071767
253.8 0.071767
296.86 0.071767
345.21 0.071767
398.2 0.071767
454.79 0.071767
516.87 0.071767
581.41 0.071767
650.74 0.071767
};
\addplot [semithick, steelblue31119180, dashed]
table {%
0.02 0.023507
0.04 0.037386
0.12 0.051175
0.31 0.060549
0.71 0.065925
1.87 0.068788
4.08 0.070264
7.46 0.071013
12.15 0.07139
18.75 0.071579
26.96 0.071674
37.15 0.071721
49.71 0.071745
64.48 0.071757
81.44 0.071763
101.86 0.071765
125.82 0.071767
152.13 0.071768
181.7 0.071768
214.69 0.071768
253.8 0.071768
296.86 0.071768
345.21 0.071768
398.2 0.071768
454.79 0.071768
516.87 0.071768
581.41 0.071768
650.74 0.071768
};
\addplot [semithick, steelblue31119180]
table {%
0.02 0.126406
0.04 0.107271
0.12 0.092712
0.31 0.083469
0.71 0.083168
1.87 0.080589
4.08 0.080529
7.46 0.080465
12.15 0.079258
18.75 0.079228
26.96 0.079164
37.15 0.079131
49.71 0.079097
64.48 0.079062
81.44 0.078582
101.86 0.078566
125.82 0.078535
152.13 0.07852
181.7 0.078504
214.69 0.078487
253.8 0.078471
296.86 0.078454
345.21 0.078437
398.2 0.078402
454.79 0.078384
516.87 0.078366
581.41 0.078348
650.74 0.078328
};
\draw (axis cs:0.1,0.075) node[
  scale=0.5,
  anchor=base west,
  text=black,
  rotate=0.0
]{$W(\Omega)'$};
\end{axis}%
\end{tikzpicture}%
    \caption{\ex{Invent} with evidence 2.}
    \label{subfig:invent_2}
\end{subfigure}
\hfill
\begin{subfigure}[b]{\plotwidth}
\begin{tikzpicture}

\definecolor{crimson2143940}{RGB}{214,39,40}
\definecolor{darkgray176}{RGB}{176,176,176}
\definecolor{gray127}{RGB}{127,127,127}
\definecolor{steelblue31119180}{RGB}{31,119,180}

\begin{axis}[
    width=\textwidth,
    height=\plotheight,
    xtick={0.1, 1, 10, 60, 600},
    xticklabels={0.1, 1, 10, 60, 600},
    ytick={0, 0.05, 0.1, 0.15},
    yticklabels={0, 0.05, 0.1, 0.15},
    tick label style={font=\scriptsize},
    label style={font=\scriptsize},
    major tick length=\MajorTickLength,
    xlabel style={yshift=\xlabelShift},
log basis x={10},
tick align=outside,
tick pos=left,
x grid style={darkgray176},
xlabel={Time [s]},
xmin=0.8, xmax=600,
xmode=log,
xtick style={color=black},
y grid style={darkgray176},
ylabel=\empty, 
ymin=0.00, ymax=0.15,
ytick style={color=black}
]
\path [draw=gray127, fill=gray127, opacity=0.5]
(axis cs:0.05,0.0709305666449216)
--(axis cs:0.05,0.0722726800754794)
--(axis cs:0.03,0.0722726800754794)
--(axis cs:0.08,0.0722726800754794)
--(axis cs:0.25,0.0722726800754794)
--(axis cs:0.86,0.0722726800754794)
--(axis cs:2.51,0.0722726800754794)
--(axis cs:7.41,0.0722726800754794)
--(axis cs:16.52,0.0722726800754794)
--(axis cs:36.38,0.0722726800754794)
--(axis cs:65.34,0.0722726800754794)
--(axis cs:113.46,0.0722726800754794)
--(axis cs:184.07,0.0722726800754794)
--(axis cs:278.53,0.0722726800754794)
--(axis cs:413.12,0.0722726800754794)
--(axis cs:579.69,0.0722726800754794)
--(axis cs:786.57,0.0722726800754794)
--(axis cs:786.57,0.0709305666449216)
--(axis cs:786.57,0.0709305666449216)
--(axis cs:579.69,0.0709305666449216)
--(axis cs:413.12,0.0709305666449216)
--(axis cs:278.53,0.0709305666449216)
--(axis cs:184.07,0.0709305666449216)
--(axis cs:113.46,0.0709305666449216)
--(axis cs:65.34,0.0709305666449216)
--(axis cs:36.38,0.0709305666449216)
--(axis cs:16.52,0.0709305666449216)
--(axis cs:7.41,0.0709305666449216)
--(axis cs:2.51,0.0709305666449216)
--(axis cs:0.86,0.0709305666449216)
--(axis cs:0.25,0.0709305666449216)
--(axis cs:0.08,0.0709305666449216)
--(axis cs:0.03,0.0709305666449216)
--(axis cs:0.05,0.0709305666449216)
--cycle;

\addplot [semithick, black]
table {%
0.05 0.0722726800754794
0.03 0.0722726800754794
0.08 0.0722726800754794
0.25 0.0722726800754794
0.86 0.0722726800754794
2.51 0.0722726800754794
7.41 0.0722726800754794
16.52 0.0722726800754794
36.38 0.0722726800754794
65.34 0.0722726800754794
113.46 0.0722726800754794
184.07 0.0722726800754794
278.53 0.0722726800754794
413.12 0.0722726800754794
579.69 0.0722726800754794
786.57 0.0722726800754794
};
\addplot [semithick, crimson2143940]
table {%
0.03 0.023506
0.08 0.037033
0.25 0.050436
0.86 0.059591
2.51 0.060187
7.41 0.062743
16.52 0.062857
36.38 0.062965
65.34 0.064073
113.46 0.064131
184.07 0.064244
278.53 0.064297
413.12 0.064349
579.69 0.064399
786.57 0.064751
};
\addplot [semithick, crimson2143940, dashed]
table {%
0.03 0.126416
0.08 0.107008
0.25 0.09224
0.86 0.082884
2.51 0.077571
7.41 0.074733
16.52 0.073265
36.38 0.072518
65.34 0.072142
113.46 0.071952
184.07 0.071858
278.53 0.07181
413.12 0.071787
579.69 0.071775
786.57 0.071769
};
\addplot [semithick, steelblue31119180, dashed]
table {%
0.03 0.023506
0.08 0.037383
0.25 0.051169
0.86 0.060542
2.51 0.065918
7.41 0.068782
16.52 0.070258
36.38 0.071007
65.34 0.071384
113.46 0.071573
184.07 0.071668
278.53 0.071715
413.12 0.071739
579.69 0.071751
786.57 0.071757
};
\addplot [semithick, steelblue31119180]
table {%
0.03 0.126416
0.08 0.107276
0.25 0.092713
0.86 0.083468
2.51 0.083166
7.41 0.080586
16.52 0.080525
36.38 0.080462
65.34 0.079254
113.46 0.079224
184.07 0.07916
278.53 0.079127
413.12 0.079094
579.69 0.079059
786.57 0.078577
};
\draw (axis cs:1,0.073) node[
  scale=0.5,
  anchor=base west,
  text=black,
  rotate=0.0
]{$W(\Omega)'$};
\end{axis}

\end{tikzpicture}
    \caption{\ex{Invent} with evidence 3.}
    \label{subfig:invent_3}
\end{subfigure}
\hfill
\begin{subfigure}[b]{\plotwidth}
\begin{tikzpicture}

\definecolor{crimson2143940}{RGB}{214,39,40}
\definecolor{darkgray176}{RGB}{176,176,176}
\definecolor{gray127}{RGB}{127,127,127}
\definecolor{steelblue31119180}{RGB}{31,119,180}

\begin{axis}[
    width=\textwidth,
    height=\plotheight,
    xtick={0.1, 1, 10, 60, 600},
    xticklabels={0.1, 1, 10, 60, 600},
    ytick={0, 0.05, 0.1, 0.15},
    yticklabels={0, 0.05, 0.1, 0.15},
    tick label style={font=\scriptsize},
    label style={font=\scriptsize},
    major tick length=\MajorTickLength,
    xlabel style={yshift=\xlabelShift},
log basis x={10},
tick align=outside,
tick pos=left,
x grid style={darkgray176},
xlabel={Time [s]},
xmin=0.8, xmax=600,
xmode=log,
xtick style={color=black},
y grid style={darkgray176},
ylabel=\empty,
ymin=0.0, ymax=0.15,
ytick style={color=black}
]
\path [draw=gray127, fill=gray127, opacity=0.5]
(axis cs:0.05,0.0709451974300085)
--(axis cs:0.05,0.0722938112018405)
--(axis cs:10.96,0.0722938112018405)
--(axis cs:19.69,0.0722938112018405)
--(axis cs:29.48,0.0722938112018405)
--(axis cs:51.91,0.0722938112018405)
--(axis cs:107.65,0.0722938112018405)
--(axis cs:251.25,0.0722938112018405)
--(axis cs:505.12,0.0722938112018405)
--(axis cs:1017.5,0.0722938112018405)
--(axis cs:1017.5,0.0709451974300085)
--(axis cs:1017.5,0.0709451974300085)
--(axis cs:505.12,0.0709451974300085)
--(axis cs:251.25,0.0709451974300085)
--(axis cs:107.65,0.0709451974300085)
--(axis cs:51.91,0.0709451974300085)
--(axis cs:29.48,0.0709451974300085)
--(axis cs:19.69,0.0709451974300085)
--(axis cs:10.96,0.0709451974300085)
--(axis cs:0.05,0.0709451974300085)
--cycle;

\addplot [semithick, black]
table {%
0.05 0.0722938112018405
10.96 0.0722938112018405
19.69 0.0722938112018405
29.48 0.0722938112018405
51.91 0.0722938112018405
107.65 0.0722938112018405
251.25 0.0722938112018405
505.12 0.0722938112018405
1017.5 0.0722938112018405
};
\addplot [semithick, crimson2143940]
table {%
10.96 0.010971
19.69 0.035245
29.48 0.049944
51.91 0.059344
107.65 0.059975
251.25 0.062555
505.12 0.062678
1017.5 0.062788
};
\addplot [semithick, crimson2143940, dashed]
table {%
10.96 0.126417
19.69 0.107
29.48 0.092213
51.91 0.082829
107.65 0.077491
251.25 0.074636
505.12 0.073158
1017.5 0.072405
};
\addplot [semithick, steelblue31119180, dashed]
table {%
10.96 0.010971
19.69 0.036095
29.48 0.050835
51.91 0.060375
107.65 0.065803
251.25 0.068686
505.12 0.070171
1017.5 0.070924
};
\addplot [semithick, steelblue31119180]
table {%
10.96 0.126417
19.69 0.10727
29.48 0.092691
51.91 0.083428
107.65 0.083121
251.25 0.080535
505.12 0.080473
1017.5 0.080409
};
\draw (axis cs:1,0.073) node[
  scale=0.5,
  anchor=base west,
  text=black,
  rotate=0.0
]{$W(\Omega)'$};
\end{axis}

\end{tikzpicture}
    \caption{\ex{Invent} with evidence 4.}
    \label{subfig:invent_4}
\end{subfigure}
}

\newcommand{\plotheight}{5.5cm}
\begin{subfigure}[b]{0.5\linewidth}
\begin{tikzpicture}

\definecolor{crimson2143940}{RGB}{214,39,40}
\definecolor{darkgray176}{RGB}{176,176,176}
\definecolor{gray127}{RGB}{127,127,127}
\definecolor{steelblue31119180}{RGB}{31,119,180}

\begin{axis}[
    width=\textwidth,
    height=\plotheight,
    xtick={0.1, 1, 10, 60, 600},
    xticklabels={0.1, 1, 10, 60, 600},
    ytick={0.85, 0.9, 0.95, 1},
    yticklabels={0.85, 0.9, 0.95, 1},
log basis x={10},
tick align=outside,
tick pos=left,
x grid style={darkgray176},
xlabel={Time [s]},
xmin=0.8, xmax=600,
xmode=log,
xtick style={color=black},
y grid style={darkgray176},
ylabel={Cond. reach. prob.},
ymin=0.85, ymax=1.0,
ytick style={color=black}
]
\path [draw=gray127, fill=gray127, opacity=0.5]
(axis cs:0.05,0.957680733645)
--(axis cs:0.05,0.96580619582208)
--(axis cs:0.53,0.96580619582208)
--(axis cs:2.31,0.96580619582208)
--(axis cs:8.63,0.96580619582208)
--(axis cs:29.6,0.96580619582208)
--(axis cs:94.83,0.96580619582208)
--(axis cs:284.51,0.96580619582208)
--(axis cs:902,0.96580619582208)
--(axis cs:902,0.957680733645)
--(axis cs:902,0.957680733645)
--(axis cs:284.51,0.957680733645)
--(axis cs:94.83,0.957680733645)
--(axis cs:29.6,0.957680733645)
--(axis cs:8.63,0.957680733645)
--(axis cs:2.31,0.957680733645)
--(axis cs:0.53,0.957680733645)
--(axis cs:0.05,0.957680733645)
--cycle;

\addplot [semithick, black]
table {%
0.05 0.96580619582208
0.53 0.96580619582208
2.31 0.96580619582208
8.63 0.96580619582208
29.6 0.96580619582208
94.83 0.96580619582208
284.51 0.96580619582208
902 0.96580619582208
};
\addplot [semithick, crimson2143940]
table {%
0.53 0.861047
2.31 0.90718
8.63 0.928668
29.6 0.941207
94.83 0.94636
284.51 0.947352
902 0.948505
};
\addplot [semithick, crimson2143940, dashed]
table {%
0.53 0.994669
2.31 0.984359
8.63 0.973559
29.6 0.966009
94.83 0.96164
284.51 0.9593
902 0.95809
};
\addplot [semithick, steelblue31119180, dashed]
table {%
0.53 0.861047
2.31 0.920276
8.63 0.945178
29.6 0.955944
94.83 0.960777
284.51 0.963382
902 0.964936
};
\addplot [semithick, steelblue31119180]
table {%
0.53 0.994668
2.31 0.986196
8.63 0.977911
29.6 0.973028
94.83 0.971345
284.51 0.970409
902 0.969535
};
\draw (axis cs:1,0.97) node[
  scale=0.5,
  anchor=base west,
  text=black,
  rotate=0.0
]{$W(\Omega)'$};
\end{axis}

\end{tikzpicture}
    \caption{\ex{Ahrs} with evidence 3.}
    \label{subfig:ahrs_3}
\end{subfigure}
\hfill
\begin{subfigure}[b]{0.5\linewidth}
\begin{tikzpicture}

\definecolor{crimson2143940}{RGB}{214,39,40}
\definecolor{darkgray176}{RGB}{176,176,176}
\definecolor{gray127}{RGB}{127,127,127}
\definecolor{steelblue31119180}{RGB}{31,119,180}

\begin{axis}[
    width=\textwidth,
    height=\plotheight,
    xtick={0.1, 1, 10, 60, 600},
    xticklabels={0.1, 1, 10, 60, 600},
    ytick={0.15, 0.20, 0.25, 0.30},
    yticklabels={0.15, 0.20, 0.25, 0.30},
log basis x={10},
tick align=outside,
tick pos=left,
x grid style={darkgray176},
xlabel={Time [s]},
xmin=0.8, xmax=600,
xmode=log,
xtick style={color=black},
y grid style={darkgray176},
ylabel=\empty,
ymin=0.15, ymax=0.30,
ytick style={color=black}
]
\path [draw=gray127, fill=gray127, opacity=0.5]
(axis cs:0.05,0.206331807339844)
--(axis cs:0.05,0.210001959573082)
--(axis cs:0.44,0.210001959573082)
--(axis cs:1.8,0.210001959573082)
--(axis cs:5.11,0.210001959573082)
--(axis cs:16.9,0.210001959573082)
--(axis cs:54.51,0.210001959573082)
--(axis cs:176.13,0.210001959573082)
--(axis cs:547.54,0.210001959573082)
--(axis cs:1913.18,0.210001959573082)
--(axis cs:1913.18,0.206331807339844)
--(axis cs:1913.18,0.206331807339844)
--(axis cs:547.54,0.206331807339844)
--(axis cs:176.13,0.206331807339844)
--(axis cs:54.51,0.206331807339844)
--(axis cs:16.9,0.206331807339844)
--(axis cs:5.11,0.206331807339844)
--(axis cs:1.8,0.206331807339844)
--(axis cs:0.44,0.206331807339844)
--(axis cs:0.05,0.206331807339844)
--cycle;

\addplot [semithick, black]
table {%
0.05 0.210001959573082
0.44 0.210001959573082
1.8 0.210001959573082
5.11 0.210001959573082
16.9 0.210001959573082
54.51 0.210001959573082
176.13 0.210001959573082
547.54 0.210001959573082
1913.18 0.210001959573082
};
\addplot [semithick, crimson2143940]
table {%
0.44 0.158496
1.8 0.180208
5.11 0.191827
16.9 0.197751
54.51 0.200123
176.13 0.201279
547.54 0.20212
1913.18 0.202454
};
\addplot [semithick, crimson2143940, dashed]
table {%
0.44 0.255566
1.8 0.232449
5.11 0.219851
16.9 0.213184
54.51 0.209797
176.13 0.208061
547.54 0.20719
1913.18 0.206753
};
\addplot [semithick, steelblue31119180, dashed]
table {%
0.44 0.158496
1.8 0.18369
5.11 0.196826
16.9 0.203406
54.51 0.206703
176.13 0.208353
547.54 0.209179
1913.18 0.209591
};
\addplot [semithick, steelblue31119180]
table {%
0.44 0.255566
1.8 0.235434
5.11 0.224317
16.9 0.21847
54.51 0.216125
176.13 0.21499
547.54 0.214146
1913.18 0.21382
};
\draw (axis cs:1,0.22) node[
  scale=0.5,
  anchor=base west,
  text=black,
  rotate=0.0
]{$W(\Omega)'$};
\end{axis}

\end{tikzpicture}
    \caption{\ex{Ahrs} with evidence 4.}
    \label{subfig:ahrs_4}
\end{subfigure}

\begin{subfigure}[b]{0.5\linewidth}
\begin{tikzpicture}

\definecolor{crimson2143940}{RGB}{214,39,40}
\definecolor{darkgray176}{RGB}{176,176,176}
\definecolor{gray127}{RGB}{127,127,127}
\definecolor{steelblue31119180}{RGB}{31,119,180}

\begin{axis}[
    width=\textwidth,
    height=\plotheight,
    xtick={0.1, 1, 10, 60, 600},
    xticklabels={0.1, 1, 10, 60, 600},
    ytick={0, 0.2, 0.4},
    yticklabels={0, \phantom{0}0.2, 0.4},
log basis x={10},
tick align=outside,
tick pos=left,
x grid style={darkgray176},
xlabel={Time [s]},
xmin=0.8, xmax=600,
xmode=log,
xtick style={color=black},
y grid style={darkgray176},
ylabel={Cond. reach. prob.},
ymin=0.0, ymax=0.4,
ytick style={color=black}
]
\path [draw=gray127, fill=gray127, opacity=0.5]
(axis cs:0.05,0.235406380121471)
--(axis cs:0.05,0.240415459154772)
--(axis cs:0.34,0.240415459154772)
--(axis cs:1.29,0.240415459154772)
--(axis cs:4.91,0.240415459154772)
--(axis cs:17.13,0.240415459154772)
--(axis cs:58.93,0.240415459154772)
--(axis cs:209.5,0.240415459154772)
--(axis cs:796.7,0.240415459154772)
--(axis cs:796.7,0.235406380121471)
--(axis cs:796.7,0.235406380121471)
--(axis cs:209.5,0.235406380121471)
--(axis cs:58.93,0.235406380121471)
--(axis cs:17.13,0.235406380121471)
--(axis cs:4.91,0.235406380121471)
--(axis cs:1.29,0.235406380121471)
--(axis cs:0.34,0.235406380121471)
--(axis cs:0.05,0.235406380121471)
--cycle;

\addplot [semithick, black]
table {%
0.05 0.240415459154772
0.34 0.240415459154772
1.29 0.240415459154772
4.91 0.240415459154772
17.13 0.240415459154772
58.93 0.240415459154772
209.5 0.240415459154772
796.7 0.240415459154772
};
\addplot [semithick, crimson2143940]
table {%
0.34 0.103458
1.29 0.158484
4.91 0.194385
17.13 0.213827
58.93 0.222458
209.5 0.226798
796.7 0.229521
};
\addplot [semithick, crimson2143940, dashed]
table {%
0.34 0.37545
1.29 0.312596
4.91 0.274594
17.13 0.25556
58.93 0.246306
209.5 0.24191
796.7 0.239949
};
\addplot [semithick, steelblue31119180, dashed]
table {%
0.34 0.103458
1.29 0.163812
4.91 0.200048
17.13 0.219898
58.93 0.229628
209.5 0.234389
796.7 0.236734
};
\addplot [semithick, steelblue31119180]
table {%
0.34 0.37545
1.29 0.31525
4.91 0.278972
17.13 0.260694
58.93 0.252601
209.5 0.248723
796.7 0.246067
};
\draw (axis cs:1,0.25) node[
  scale=0.5,
  anchor=base west,
  text=black,
  rotate=0.0
]{$W(\Omega)'$};
\end{axis}%
\end{tikzpicture}%
    \caption{\ex{Phil} with evidence 2.}
    \label{subfig:phil_2}
\end{subfigure}
\hfill
\begin{subfigure}[b]{0.5\linewidth}
\begin{tikzpicture}

\definecolor{crimson2143940}{RGB}{214,39,40}
\definecolor{darkgray176}{RGB}{176,176,176}
\definecolor{gray127}{RGB}{127,127,127}
\definecolor{steelblue31119180}{RGB}{31,119,180}

\begin{axis}[
    width=\textwidth,
    height=\plotheight,
    xtick={0.1, 1, 10, 60, 600},
    xticklabels={0.1, 1, 10, 60, 600},
    ytick={0, 0.25, 0.5, 0.75, 1},
    yticklabels={0, 0.25, 0.5, 0.75, 1},
log basis x={10},
tick align=outside,
tick pos=left,
x grid style={darkgray176},
xlabel={Time [s]},
xmin=0.8, xmax=600,
xmode=log,
xtick style={color=black},
y grid style={darkgray176},
ylabel=\empty,
ymin=0.0, ymax=1.0,
ytick style={color=black}
]
\path [draw=gray127, fill=gray127, opacity=0.5]
(axis cs:0.05,0.13319693505324)
--(axis cs:0.05,0.133546368421131)
--(axis cs:4.45,0.133546368421131)
--(axis cs:7.54,0.133546368421131)
--(axis cs:14.04,0.133546368421131)
--(axis cs:32.08,0.133546368421131)
--(axis cs:75.79,0.133546368421131)
--(axis cs:183.89,0.133546368421131)
--(axis cs:407.36,0.133546368421131)
--(axis cs:876.96,0.133546368421131)
--(axis cs:876.96,0.13319693505324)
--(axis cs:876.96,0.13319693505324)
--(axis cs:407.36,0.13319693505324)
--(axis cs:183.89,0.13319693505324)
--(axis cs:75.79,0.13319693505324)
--(axis cs:32.08,0.13319693505324)
--(axis cs:14.04,0.13319693505324)
--(axis cs:7.54,0.13319693505324)
--(axis cs:4.45,0.13319693505324)
--(axis cs:0.05,0.13319693505324)
--cycle;

\addplot [semithick, black]
table {%
0.05 0.133546368421131
4.45 0.133546368421131
7.54 0.133546368421131
14.04 0.133546368421131
32.08 0.133546368421131
75.79 0.133546368421131
183.89 0.133546368421131
407.36 0.133546368421131
876.96 0.133546368421131
};
\addplot [semithick, crimson2143940]
table {%
4.45 0.084494
7.54 0.086299
14.04 0.092748
32.08 0.103213
75.79 0.107556
183.89 0.107982
407.36 0.107982
876.96 0.114281
};
\addplot [semithick, crimson2143940, dashed]
table {%
4.45 0.959545
7.54 0.62584
14.04 0.311273
32.08 0.200959
75.79 0.162314
183.89 0.146749
407.36 0.139832
876.96 0.136578
};
\addplot [semithick, steelblue31119180, dashed]
table {%
4.45 0.084494
7.54 0.086432
14.04 0.093191
32.08 0.103926
75.79 0.114613
183.89 0.122554
407.36 0.127439
876.96 0.130187
};
\addplot [semithick, steelblue31119180]
table {%
4.45 0.959544
7.54 0.629157
14.04 0.312435
32.08 0.202226
75.79 0.181864
183.89 0.181386
407.36 0.181386
876.96 0.162762
};
\draw (axis cs:1,0.14) node[
  scale=0.5,
  anchor=base west,
  text=black,
  rotate=0.0
]{$W(\Omega)'$};
\end{axis}%
\end{tikzpicture}%
    \caption{\ex{Tandem} with evidence 2.}
    \label{subfig:tandem_2}
\end{subfigure}

\caption{Additional results for different CTMCs and evidences.}
\label{fig:additional_results}
\end{figure}
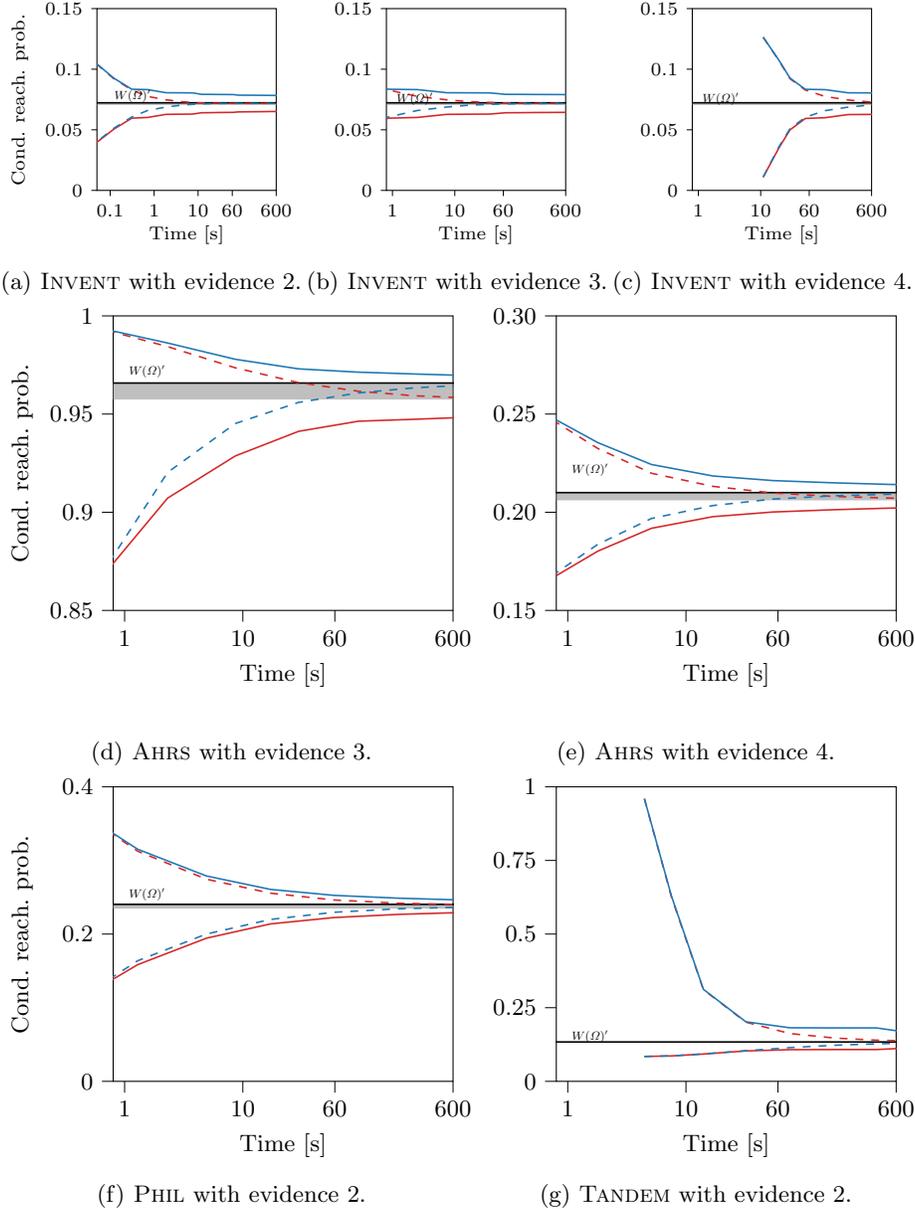

\Cref{fig:additional_results} provides additional plots for the benchmarks of \cref{sec:experiments} (see \cref{tab:benchmarks_summary} for the benchmark statistics).
Overall, we see the same results as observed in \cref{sec:experiments}: our method is able to find reasonably tight bounds on the weighted reachability within the used time limit of 10 minutes.

One major factor regarding scalability can be seen for \ex{Invent} with evidence 4 in \cref{subfig:invent_4}.
The evidence consists of 15 observations and as a result, our approach requires more than 10 seconds to obtain the first result.
However, within a minute and performing a few refinement steps, we still obtain a reasonable bound on the weighted reachability.

\Cref{subfig:tandem_2} shows that the coarse partitioning of the timings can initially lead to coarse bounds on the weighted reachability.
However, the bounds become again tighter after a few refinement steps.

\fi

\end{document}